%% file: hmotifs.tex
\documentclass{article}
\usepackage{graphicx} 
\usepackage{subcaption}
\usepackage{authblk}
\usepackage{amsmath}
\usepackage{amsfonts}
\usepackage{todonotes}
\usepackage{amssymb}
\usepackage{mathtools}
\usepackage[]{mdframed}
\usepackage{algorithm}
\usepackage{algpseudocode}
\usepackage{tikz}
\usetikzlibrary{intersections}
\usetikzlibrary{patterns}
\usepackage{xcolor}
\usepackage[english]{babel}

\usepackage{amsthm}

\theoremstyle{plain}
\newtheorem{theorem}{Theorem}[section]
\newtheorem{lemma}[theorem]{Lemma}
\newtheorem{corollary}[theorem]{Corollary}
\newtheorem{obs}[theorem]{Observation}
\newtheorem{conj}[theorem]{Conjecture}

\theoremstyle{definition}
\newtheorem{definition}[theorem]{Definition}
\newtheorem{remark}[theorem]{Remark}

\usepackage{url}
\usepackage[margin=1in]{geometry}

\newcommand{\motif}[2]{{\mathcal{M}(#1\to #2)}} 
\newcommand{\nel}{north east lines}
\renewcommand{\hom}[2]{\mathsf{Hom}(#1 \to #2)}
\newcommand{\homs}[2]{\#\mathsf{Hom}(#1 \to #2)}
\newcommand{\cphom}[3]{\mathsf{Hom}_\mathsf{cp}(#1 \to_{#2} #3)}
\newcommand{\cphoms}[3]{\#\mathsf{Hom}_\mathsf{cp}(#1 \to_{#2} #3)}
\newcommand{\emb}[2]{\mathsf{Emb}(#1 \to #2)}
\newcommand{\embs}[2]{\#\mathsf{Emb}(#1 \to #2)}
\newcommand{\cpemb}[3]{\mathsf{Emb}_\mathsf{cp}(#1 \to_{#2} #3)}
\newcommand{\cpembs}[3]{\#\mathsf{Emb}_\mathsf{cp}(#1 \to_{#2} #3)}
\newcommand{\sub}[2]{\mathsf{Sub}(#1 \to #2)}
\newcommand{\subs}[2]{\#\mathsf{Sub}(#1 \to #2)}

\newcommand{\motifs}[2]{\#\mathcal{M}(#1 \to #2)}
\newcommand{\colmotifs}[3]{\#\mathcal{M}_\mathsf{col}(#1 \to_{#2} #3)}
\newcommand{\colmotif}[3]{\mathcal{M}_\mathsf{col}(#1 \to_{#2} #3)}
\newcommand{\motifprob}[1]{\#\textsc{HyperMotif}(#1)}
\newcommand{\colmotifprob}[1]{\#\textsc{ColHyperMotif}(#1)}
\newcommand{\aut}[1]{\mathsf{Aut}(#1)}
\newcommand{\auts}[1]{\#\mathsf{Aut}(#1)}
\newcommand{\vecrho}{\vec{\rho}}
\newcommand{\vecsigma}{\vec{\sigma}}

\newcommand{\cycle}{\Delta}
\newcommand{\hyperclique}{\Gamma}
\newcommand{\coeff}{\mathsf{coeff}}

\newcommand{\fptlinred}{\leq^\mathsf{FPT}_\mathsf{LIN}}

\title{The Fine-Grained Complexity of Counting Hypergraph Motifs} 
\author[1]{Madhumitha Krishnakumar}
\author[1]{Marc Roth}
\affil[1]{School of Electronic Engineering \& Computer Science, Queen Mary University of London}

\begin{document}

\maketitle

\begin{abstract}
Introduced by Lee, Ko, and Shin (VLDB 2020), a hypergraph motif is a connected subhypergraph consisting of three hyperedges whose intersections satisfy a prescribed pattern. Such patterns are represented by Venn diagrams $\mathcal V\in\{0,1\}^7$, indicating which of the seven regions determined by three sets must be empty or non-empty. Lee et al.\ designed and implemented exact and approximate algorithms for counting, in a hypergraph $G$, the motifs specified by $\mathcal{V}$; their algorithms run in worst-case cubic time in the number of hyperedges of $G$. This cubic worst case can occur even for hypergraphs of bounded rank, and already for $2$-uniform hypergraphs, that is, for simple graphs.

In this work, we give a complete fine-grained picture of the parameterised complexity of exact hypergraph motif counting with respect to the rank of the input hypergraph. We use $\tilde{O}$ to hide polylogarithmic factors in the input size. First, we show that every Venn diagram $\mathcal{V}$ admits an exact counting algorithm running in FPT-near-quadratic time,
\[
f(\mathsf{rank}(G))\cdot \tilde{O}(|E(G)|^2),
\]
for some computable function $f$. Second, we precisely characterise when this can be improved to FPT-near-linear time. We prove that such an algorithm exists exactly for the \emph{degenerate} Venn diagrams, namely those that force one of the three hyperedges to be fully contained in another. For all non-degenerate Venn diagrams, we show that no FPT-near-linear-time algorithm exists unless either the Triangle Hypothesis or the Hyperclique Hypothesis fails.
Exact hypergraph motif counting is thus always fixed-parameter near-quadratic in the rank, and the degenerate Venn diagrams are precisely the cases admitting fixed-parameter near-linear time. In particular, for bounded-rank hypergraphs, where $f(\mathsf{rank}(G))=O(1)$, our results give near-quadratic upper bounds in general and a tight classification of the near-linear-time cases.

We also study \emph{generalised} hypergraph motifs defined by emptiness constraints on all intersections of $k$ hyperedges. We establish new parameterised upper and lower bounds when parameterised by $k+\mathsf{rank}(G)$. We further argue that obtaining matching bounds is closely tied to the hypergraph homomorphism counting problem, whose complexity remains open for unbounded-rank hypergraphs.

Our results build on the recently established hypergraph homomorphism basis of Bressan, Brinkmann, Dell, Roth, and Wellnitz (SODA 2026). The main technical challenge is the proof of the fine-grained lower bound which requires us to express the number of hypergraph motifs as a linear combination of homomorphism counts and to show that cyclic terms survive with non-zero coefficients for all non-degenerate Venn diagrams. This task is complicated significantly by the fact that multiple non-isomorphic subhypergraphs can satisfy the same Venn diagram. We overcome this challenge by introducing and analysing an intermediate coloured version of the problem. 
\end{abstract}
\pagebreak

\section{Introduction}
Motif counting refers to the problem of computing the number of occurrences of a small pattern in a large host network. Examples include subgraph counting, induced subgraph counting, as well as counting answers to a query in a relational database. A 2002 landmark result of Milo, Shen-Orr, Itzkovitz, Kashtan,  Chklovskii, and Alon\ \cite{milo_network_2002} identified a variety of subgraph patterns, subsequently called \emph{network motifs}, that appear with significantly higher frequency in real-life networks than their expected frequency in random networks, and that increased frequencies of certain patterns correlate with global features of networks. In the years since then, a significant amount of research has been undertaken to improve our understanding of the practical and theoretical aspects of counting network motifs:

From the practical side, network motifs have found applications in computational biology~\cite{Nogaetat08,Schilleretal15}, cancer detection~\cite{gomez-sanchez_clustering_2022}, social network analysis~\cite{UganderBK13}, genetics~\cite{ShenOrrMM02,Tranetal13}, phylogeny~\cite{Kuchaievetal10}, and data mining~\cite{AhmedNRD15,Babis17}, only to name a few. 
At the same time, the theoretical understanding of the inherent computational complexity of motif counting problems has seen a flurry of results in the last twenty years, not only with respect to classical complexity (``$\mathrm{P}$ vs.\ $\mathrm{NP}$''), but also in view of more modern frameworks such as fixed-parameter tractability and fine-grained complexity theory; examples include classifications for homomorphism counting ``from the other side''~\cite{DalmauJ04,RothW20}, subgraph counting~\cite{CurticapeanM14,CurticapeanDM17}, induced subgraph counting~\cite{ChenTW08,FockeR24,DoringMW24}, and motif counting in somewhere dense graphs~\cite{BressanGMR24}.

So far, the study of motif counting has almost exclusively been focused on graphs, that is, data organised in binary connections. However, a growing body of literature highlights the practical relevance of motif counting in \emph{hypergraphs}, including, for example, applications in computer vision~\cite{CV1,CV2,CV3}, clustering~\cite{CL1,CL2}, link prediction~\cite{LP1,LP2,LP3}, and bioinformatics~\cite{hwang_learning_2008}. 

In an attempt for the systematic study of motif counting in hypergraphs, Lee, Ko, and Shin~\cite{Lee_conference}\footnote{See also the extended journal version with Yoon and Kim~\cite{Lee}.} introduce \emph{hypergraph motifs} as subhypergraphs consisting of $k$ hyperedges, the intersections of which are described by a $2^k-1$-dimensional vector as follows:  Let $\mathfrak{h}\in \{0,1\}^{2^k-1}$ be a binary vector indexed by non-empty subsets of $[k]$, that is \[\mathfrak{h}=(\mathfrak{h}_J)_{\emptyset \neq J \subseteq [k]}\,.\]
A hypergraph $G$ with $k$ (distinct) hyperedges satisfies $\mathfrak{h}$ if, and only if, there is an ordering $e_1,\dots,e_k$ of the hyperedges such that, for each non-empty $J \subseteq [k]$, we have
\[\mathfrak{h}_J=1 \Leftrightarrow \left(\bigcap_{i \in J}e_i \right)\setminus \left(\bigcup_{i \in [k]\setminus J} e_i\right) \neq \emptyset\,.\]
Given a hypergraph motif $\mathfrak{h}$, and a hypergraph $G$, we write $\motifs{\mathfrak{h}}{G}$ for the number of subhypergraphs of $G$ that satisfy $\mathfrak{h}$.

The majority of the work of Lee et al.\ focuses on the case $k=3$ for which we can view $\mathfrak{h}$ as a \emph{Venn diagram} --- an example is provided in Figure~\ref{fig:intro_Venndiagram}. For the remainder of the paper, we will use $\mathfrak{h}$ to denote (general) hypergraph motifs with $k$ hyperedges, and we will use $\mathcal{V}$ for the special case $k=3$, i.e., for Venn diagrams.
\begin{figure}
    \centering
     \includegraphics[width=0.2\linewidth]{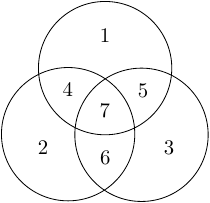} \hspace{1.7cm}\scalebox{-1}[1]{\includegraphics[width=0.5\linewidth]{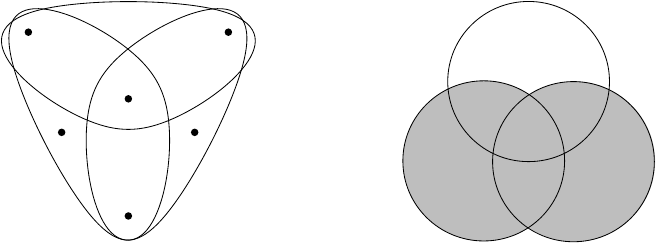}}
    \caption{\emph{(Left:)} An indexing of the intersections of three sets. \emph{(Centre):} Illustration of the Venn diagram $\mathcal{V}=(0,1,1,1,1,1,1)$. \emph{(Right):} A hypergraph satisfying $\mathcal{V}$.}
    \label{fig:intro_Venndiagram}
\end{figure}
While there are $2^7$ Venn diagrams in total, Lee et al.\ proposed to only keep one representative per permutation of the three sets, and to delete those that are not connected or that cannot be obtained from distinct hyperedges~\cite[Section 2.2]{Lee}. This leaves us with the $26$ Venn diagrams illustrated in Figure~\ref{fig:allVenns}. In their paper, Lee et al.\ discover that the number of subhypergraphs specified by those Venn diagrams are capable of characterising a variety of features of real-world hypergraphs such as domain-specific similarities, and that they have applications in Machine Learning such as clustering and hyperedge prediction~\cite[Section 6]{Lee}. Their results strongly indicate that even for the case $k=3$, hypergraph motifs promise to be a natural and effective generalisation of Milo et al.'s (\cite{milo_network_2002}) network motifs from graphs to hypergraphs.  

\paragraph*{Complexity Theoretic Challenges}
Lee et al.\ provide and implement algorithms for both exactly and approximately computing $\motifs{\mathcal{V}}{G}$. Although those algorithms are sufficiently efficient for the benchmark hypergraphs considered in their paper, the theoretical worst case running time is cubic in the number of hyperedges, and this behaviour appears already for a star \emph{graph}, that is, a for $2$-uniform hypergraph.\footnote{Their approximation algorithm can run asymptotically faster if fewer samples are taken, which yields however worse probabilistic guarantees for correctness in worst case inputs.}
In the present work we address the fine-grained complexity of counting hypergraph motifs and ask whether it is possible to beat cubic worst-case running time. Moreover, we also investigate the complexity of general hypergraph motifs $\mathfrak{h}$ on more than $3$ hyperedges.

\subsection{Our Results}
For our first result, we aim to determine, for each Venn diagram $\mathcal{V}$, the optimal exponent in the running time for any algorithm that computes the function $G \mapsto \motifs{\mathcal{V}}{G}$. 

\paragraph*{Parameterising by the rank}
While we do not want to restrict ourselves to input hypergraphs of constant rank --- the rank denotes the maximum size of a hyperedge --- we still wish to exploit the assumption that hypergraph motifs are \emph{small} patterns in a \emph{large} host hypergraph. Concretely, we do not aim for algorithms that run efficiently on instances $G$ containing hyperedges of size close to the order of $G$ itself. Instead, we will assume that the rank of $G$ is significantly smaller than $G$ itself --- we claim that this is a very natural feature of real-world hypergraphs and note that a similar assumption is made implicitly by Lee et al.\ when selecting their benchmark instances~\cite[Section 6]{Lee}; for example, they consider the DBLP co-authorship hypergraph which contains a node for each author, and a hyperedge for each paper including all of its authors, yielding a hypergraph with approximately 2 million nodes, 2.5 million hyperedges, but with rank only $25$. 

We make this assumption formal by relying on the toolkit of parameterised complexity theory and define the problem of counting hypergraph motifs as follows --- note that each Venn diagram $\mathcal{V}$ defines a unique problem:

\begin{mdframed}
\#\textsc{HyperMotif}($\mathcal{V}$) \\
\textbf{Input:} A hypergraph $G$ \\
\textbf{Parameter:} $\mathsf{rank}(G)$ \\
\textbf{Output:} $\#\motif{\mathcal{V}}{G}$
\end{mdframed}  
We say that $\#\textsc{HyperMotif}(\mathcal{V})$ is solvable in \emph{FPT}\footnote{\textbf{F}ixed-\textbf{P}arameter \textbf{T}ractable; see Section~\ref{sec:param_and_finegrained} for a concise introduction to parameterised complexity theory}\emph{-near-linear time} if it can be solved in time
\[f(\mathsf{rank}(G))\cdot \tilde{O}(|G|)\]
for some computable function $f$; here $\tilde{O}$ hides polylogarithmic factors. Note that, an FPT-near-linear time algorithm yields a near-linear-time algorithm when the problem is restricted to hypergraphs of constant rank, as $f(\mathsf{rank}(G))\in O(1)$ in that case. Similarly, we say that $\#\textsc{HyperMotif}(\mathcal{V})$ is solvable in \emph{FPT-near-quadratic time} if it can be solved in time 
\[f(\mathsf{rank}(G))\cdot \tilde{O}(|G|^2)\,.\]
\pagebreak

\noindent Our initial question can now be stated formally as follows:
\begin{enumerate}
    \item For which Venn diagrams $\mathcal{V}$ is $\#\textsc{HyperMotif}(\mathcal{V})$ solvable in FPT-near-linear time? 
    \item For which Venn diagrams $\mathcal{V}$ is $\#\textsc{HyperMotif}(\mathcal{V})$ solvable in FPT-near-quadratic time?
    \item Can we prove our results to be optimal under standard lower bound assumptions from fine-grained complexity theory?
\end{enumerate}
We are able to answer all of those questions comprehensively. For stating our results, we say that a Venn-diagram is degenerate if it enforces one edge to be fully contained in another edge; the formal definition can be found in Section~\ref{sec:venns}, but we encourage the reader to just consult Figure~\ref{fig:allVenns} for now.
\newcommand{\redish}{red!70!black}
\newcommand{\greenish}{green!50!black}

\begin{figure}[t]
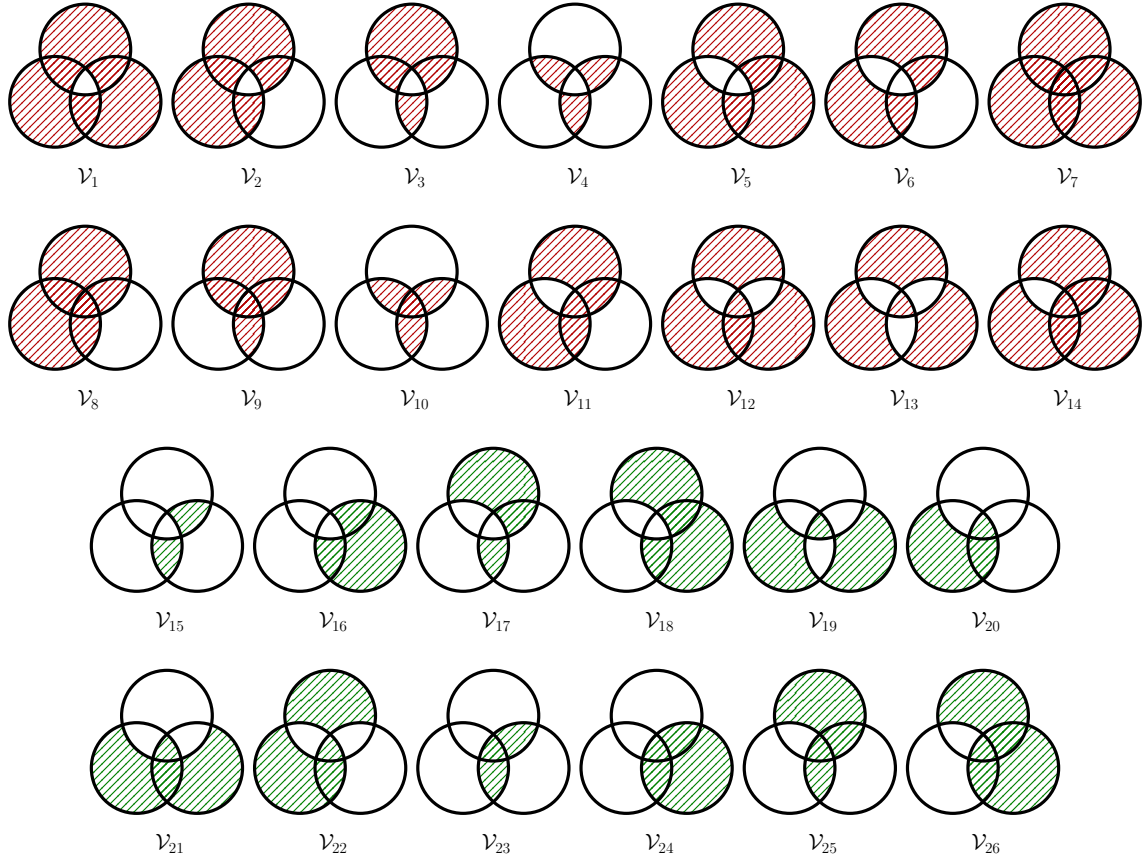

\include{big_venn_diagram_picture}
\caption{\label{fig:allVenns}The Venn diagrams $\mathcal{V}_1,\dots,\mathcal{V}_{14}$ represent the \emph{non-degenerate} cases. We show that hypergraph motifs corresponding to those Venn diagrams cannot be counted in (FPT-)near-linear time under standard lower bound assumptions. In contrast, the Venn diagrams $\mathcal{V}_{15},\dots,\mathcal{V}_{26}$ are degenerate, and we will see that counting hypergraph motifs corresponding to those Venn diagrams is possible in (FPT-)near-linear time.}
\end{figure}

\begin{theorem}[Main Theorem, Upper Bound]\label{intro:upper}
    Let $\mathcal{V}$ be a Venn diagram. The problem $\motifprob{\mathcal{V}}$ can be solved in FPT near-quadratic time, that is, there is a computable function $f$ such that the problem can be solved in time
    \[ f(\mathsf{rank}(G))\cdot \tilde{O}(|G|^2) \,.\]
    Moreover, if $\mathcal{V}$ is degenerate, then $\motifprob{\mathcal{V}}$ can be solved in FPT near-linear time, that is, there is a computable function $f$ such that the problem can be solved in time
    \[ f(\mathsf{rank}(G))\cdot \tilde{O}(|G|) \,.\]
\end{theorem}

\noindent For fixed rank, we obtain the following direct implication.
\begin{corollary}\label{intro:upper_cor}
    When restricted to input hypergraphs of bounded rank, $\motifprob{\mathcal{V}}$ is solvable in near-quadratic time for arbitrary Venn diagrams, and it is solvable in near-linear time for degenerate Venn diagrams.
\end{corollary}
\noindent We highlight that the above result implies immediately that the cubic worst-case running time of the algorithm of Lee et al.\ \cite{Lee}, exhibited even when restricted to hypergraphs of rank $2$, is not optimal for any Venn diagram.

Next, the second part of our main result states that our algorithms are optimal under the assumption of both the Triangle Hypothesis and the Hyperclique Hypothesis (see Section~\ref{sec:param_and_finegrained}): 
\begin{theorem}[Main Theorem; Lower Bound]\label{intro:lower}
    Let $\mathcal{V}$ be a non-degenerate Venn diagram. Then the problem $\motifprob{\mathcal{V}}$ is not solvable in FPT-near-linear time (and thus also not in near-linear time) unless either the Triangle Hypothesis or the Hyperclique Hypothesis fails.
\end{theorem}

\begin{corollary}[Fine-Grained Complexity Dichotomy]
    Assume that both the Triangle Hypothesis and the Hyperclique Hypothesis are true, and let~$\mathcal{V}$ be a Venn diagram. Then $\motifprob{\mathcal{V}}$ is solvable in FPT-near-linear time if and only if $\mathcal{V}$ is degenerate. Moreover, when restricted to input hypergraphs of bounded rank, $\motifprob{\mathcal{V}}$ is solvable in near-linear time if and only if $\mathcal{V}$ is degenerate.
\end{corollary}

\paragraph*{Generalised Hypergraph Motifs}
We also investigate the case of $k>3$. To avoid any ambiguities, for the remainder of this paper, we will refer to hypergraph motifs with $k>3$ hyperedges as \emph{generalised hypergraph motifs} ``GHMs'' of order $k$. 

For any GHM $\mathfrak{H}$ of order $k$, we can compute the function $G \mapsto \#\motif{\mathfrak{H}}{G}$ in time $|G|^{O(k)}$ up to a factor only depending on $\mathfrak{H}$ by brute force, leading to a polynomial-time algorithm if $\mathfrak{H}$ (and thus $k$) are fixed. For this reason, we ask for which generalised hypergraph motifs we can obtain a running time in which the exponent of the host graph $G$ does not depend on the order of the motif. In order to state and answer this question formally, we follow the standard approach for parameterised motif detection and counting problems (cf.\ \cite{Grohe07,Marx13,Bressanetal26}) and restrict the problem to classes $\mathfrak{C}$ of allowed hypergraph motifs:

\begin{mdframed}
\#\textsc{GeneralisedHyperMotif}($\mathfrak{C}$) \\
\textbf{Input:} A GHM $\mathfrak{H}\in \mathfrak{C}$, and a hypergraph $G$ \\
\textbf{Parameter:} $|\mathfrak{H}|+\mathsf{rank}(G)$ \\
\textbf{Output:} $\#\motif{\mathfrak{H}}{G}$
\end{mdframed}

Formally, we ask for which classes $\mathfrak{C}$ the problem $\#\textsc{GeneralisedHyperMotif}(\mathfrak{C})$ can be solved in time
\[f(|\mathfrak{h}|+\mathsf{rank}(H))\cdot |G|^{O(1)}\,,\]
for some computable function $f$. We call $\#\textsc{GeneralisedHyperMotif}(\mathfrak{C})$ \emph{fixed-parameter tractable} if an algorithm with the above running time exists.

We are able to provide a partial answer to this question: We introduce two width measures of GHMs: The \emph{hereditary fractional hypertreewidth}, and the \emph{hereditary adaptive width}. The formal definitions of hereditary fractional hypertreewidth and hereditary adaptive width require a significant amount of technical set-up which is why we defer their definitions to Section~\ref{sec6:general_case}. For the purpose of presenting our result on $\#\textsc{GeneralisedHyperMotif}(\mathfrak{C})$ we will say that a class $\mathfrak{C}$ of GHMs has \emph{bounded} hereditary fractional hypertreewdith if there is a constant upper bound on the hereditary fractional hypertreewidth of all members $\mathfrak{h}\in \mathfrak{C}$; bounded hereditary adaptive width is defined similarly. We show the following:

\begin{theorem}\label{thm:main_GHMs_intro}
    Let $\mathfrak{C}$ be a recursively enumerable class of generalised hypergraph motifs.
    \begin{itemize}
        \item If $\mathfrak{C}$ has bounded hereditary fractional hypertreewidth, then $\#\textsc{GeneralisedHyperMotif}(\mathfrak{C})$ is fixed-parameter tractable.
        \item If $\mathfrak{C}$ has unbounded hereditary adaptive width, then $\#\textsc{GeneralisedHyperMotif}(\mathfrak{C})$ is not fixed-parameter tractable, unless ETH\footnote{The \textbf{E}xponential \textbf{T}ime \textbf{H}ypothesis asserts that $3$-$\textsc{SAT}$ cannot be solved in subexponential time; see Section~\ref{sec:param_and_finegrained} for the formal statement.} fails.
    \end{itemize}
\end{theorem}
We wish to emphasise that Theorem~\ref{thm:main_GHMs_intro} follows rather easily from existing works on linear combinations of hypergraph homomorphism counts~\cite{Bressanetal26} --- a comprehensive introduction will be provided in the next part. Specifically, we hence view Theorem~\ref{thm:main_GHMs_intro} as an initial result on the complexity of counting GHMs, and not as its final answer. However, we also highlight that the gap between the upper and the lower bound of Theorem~\ref{thm:main_GHMs_intro} is not due to an artifact in our proofs, but due to the unresolved open problem on the tractability criterion for parameterised homomorphism counting on hypergraphs of unbounded rank ($\#\textsc{Hom}$): in a nutshell, hereditary fractional hypertreewidth and hereditary adaptive width are, unsurprisingly, closely related to fractional hypertreewidth and adaptive width. It follows from the work of Grohe and Marx~\cite{GroheM14} that $\#\textsc{Hom}$ is fixed-parameter tractable when restricted to pattern hypergraphs of bounded fractional hypertreewidth, and it follows from Marx~\cite{Marx13} that $\#\textsc{Hom}$ is not fixed-parameter tractable when restricted to pattern hypergraphs of unbounded adaptive width (unless ETH fails), leaving a gap for classes of pattern hypergraphs of unbounded fractional hypertreewidth, but bounded adaptive width. It follows from our proofs that any result closing this gap will also close the gap exhibited by Theorem~\ref{thm:main_GHMs_intro}.

\subsection{Technical Overview}
The main mathematical tool leveraged in this work is the framework of linear combinations of hypergraph homomorphism counts. This should not come to any surprise as the majority of recent results on pattern counting in \emph{graphs} rely on linear combinations of graph homomorphism counts (see e.g.\ \cite{CurticapeanDM17,DorflerRSW22,FockeR24,DoringMW24,BressanGMR24}). Specifically, we use the extension of graph motif parameters, originally introduced by Curticapean, Dell and Marx~\cite{CurticapeanDM17}, to \emph{hypergraph} motif parameters due to Bressan, Brinkmann, Dell, Roth, and Wellnitz~\cite{Bressanetal26}. 

For the purpose of presenting said framework, we are required to first introduce hypergraph homomorphisms: Given hypergraphs $H$ and $G$, a \emph{homomorphism} from $H$ to $G$ is a mapping $\varphi: V(H) \to V(G)$, such that for every hyperedge $e=\{v_1,\dots,v_r\}$ of $H$, its image $\varphi(e)=\{\varphi(v_1),\dots,\varphi(v_r)\}$ must be a hyperedge of $G$.\footnote{In some literature, hypergraph homomorphisms $\varphi$ from $H$ to $G$ only require that $\varphi(e)$ is a subset of an edge of $e$. This is sometimes referred to as a \emph{trimmed} homomorphism (cf.\ \cite{Bressanetal26}). We emphasise that, throughout this paper, we always require that $\varphi(e)$ is a hyperedge, and not just the subset of a hyperedge.} We write $\homs{H}{G}$ for the number of homomorphisms from $H$ to $G$. Following Bressan et al.\ \cite{Bressanetal26}, a \emph{hypergraph motif parameter} is a function $\Phi$ from hypergraphs to rationals that can be expressed as a finite linear combination of homomorphism counts, that is, there is a function $\alpha_\Phi$ from hypergraphs to rationals such that $\alpha_\Phi$ has finite support, and for each hypergraph $G$ we have
\begin{equation}
    \Phi(G)= \sum_H \alpha_\Phi(H)\cdot \homs{H}{G}\,,
\end{equation}
where the sum is over all (isomorphism types of) hypergraphs. Using an interpolation argument that is based on Dedekind's theorem and that has independently been discovered by Curticapean, Dell, and Marx~\cite{CurticapeanDM17} and by Chen and Mengel~\cite{ChenM16} about a decade ago, Bressan et al.\ \cite{Bressanetal26} show that computing a hypergraph motif parameter $\Phi$ is \emph{precisely as hard as} computing the hardest term $\homs{H}{G}$ with $\alpha_\Phi(H)\neq 0$. Specifically, the interreducibility result provided by this interpolation technique, often referred to as ``Complexity Monotonicity'' or ``Dedekind Interpolation'', implies that the best exponent of $|G|$ in the running time of any algorithm that computes $\Phi$ is \emph{equal} to the best exponent of $|G|$ in the running time of any algorithm computing the hardest term $G \mapsto \homs{H}{G}$ among all $H$ with $\alpha_\Phi(H)\neq 0$.

For establishing FPT-near-linear-time algorithms, or proving that such algorithms do not exist under standard assumptions from fine-grained complexity theory, we rely on the fact that $G\mapsto \homs{H}{G}$ is computable in near-linear time if and only if $H$ is $\alpha$-acyclic; the ``only if'' direction assumes both the Triangle Hypothesis and the Hyperclique Hypothesis. The upper bound of this classification follows from the counting version of Yannakakis algorithm~\cite{Yannakakis}, and the lower bound is due to Brault-Baron~\cite{BraultBaron13} --- we provide details in Section~\ref{sec5:proof}. Moreover, for establishing FPT-near-quadratic time, we rely on the extension of Yannakakis algorithm to hypergraphs of bounded generalised hypertreewidth; specifically, the function $G \mapsto \homs{H}{G}$ can be computed in near-quadratic time if $H$ has generalised hypertreewidth at most two (cf.\ \cite{Scarcello05,Mengel_full}). We introduce $\alpha$-acyclicity and generalised hypertreewidth in Section~\ref{sec:ghtw} and note that, for the purpose of this introduction, it will be sufficient to rely on the previous results in a black-box manner. \pagebreak

Following known transformations that date back to early works of Lov{\'{a}}sz (see~\cite[Chapter 5.3]{Lov}) and their generalisations from graphs to hypergraphs~\cite{Bressanetal26}, we show that for each Venn diagram $\mathcal{V}$ and hypergraph $G$ of rank at most $r$, we have
\begin{equation}\label{eq:intro_uncoloured_hombasis}
\motifs{\mathcal{V}}{G} = \sum_{F} \alpha_{\mathcal{V},r}(F) \cdot \homs{F}{G} \,,
\end{equation}
where
\begin{equation}\label{eq:intro_uncoloured_coeff}
    \alpha_{\mathcal{V},r}(F) = \sum_{H \in S(\mathcal{V},r)} \auts{H}^{-1} \cdot \sum_{\substack{\rho \in P(V(H))\\H/\rho \cong F}} \mu(\rho)\,.
\end{equation}
Here $S(\mathcal{V},r)$ denotes the set of all hypergraphs of rank at most $r$ that satisfy $\mathcal{V}$, $P(V(H))$ denotes the set of all partitions of $V(H)$, and $H/\rho$ denotes the \emph{quotient hypergraph} obtained from $H$ by identifying all pairs of vertices contained in a common block of $\rho$; consider Figure~\ref{fig:quotient_example} for an example. Moreover, $\mu(\rho)$ denotes the M\"obius function of the partition lattice of $V(H)$ (cf.\ \cite[Chapter A.1]{Lov} and~\cite[Chapter 3]{Stanley}). 
\begin{figure}
    \centering
    \includegraphics[width=0.8\linewidth]{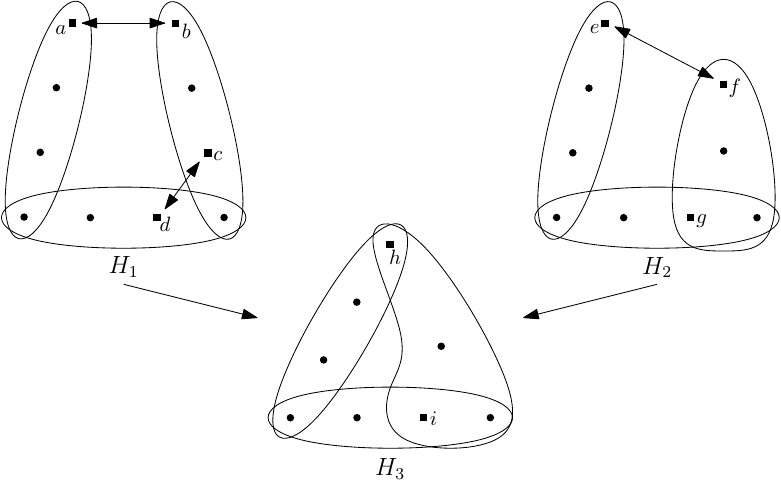}
    \caption{A pair of non-isomorphic $4$-uniform hypergraphs, $H_1$ and $H_2$, that satisfy the same Venn diagram ($\mathcal{V}_5$ in Figure~\ref{fig:allVenns}) and that have a common quotient $H_3$. Specifically, let $\rho_1$ be the partition of $V(H_1)$ that contains blocks $\{a,b\}$, $\{c,d\}$, and a singleton block for each remaining vertex. Let furthermore $\rho_2$ be the partition of $V(H_2)$ that contains block $\{e,f\}$, and a singleton block for each remaining vertex. Then $H_1/\rho_1 \cong H_2/\rho_2 \cong H_3.$}
    \label{fig:quotient_example}
\end{figure} 

\noindent At first glance, in combination with Dedekind Interpolation and the fine-grained complexity results on homomorphism counting, Equations~\eqref{eq:intro_uncoloured_hombasis} and~\eqref{eq:intro_uncoloured_coeff} appear to yield a direct strategy to prove our main result: 
\begin{enumerate}
    \item Show that, for each $\mathcal{V}$ and $r>0$, we have that $\alpha_{\mathcal{V},r}(F)\neq 0$ implies that $F$ has generalised hypertreewidth at most $2$.
    \item Show that, for each \emph{degenerate} $\mathcal{V}$ and $r>0$, we have that $\alpha_{\mathcal{V},r}(F)\neq 0$ implies that $F$ is $\alpha$-acyclic.
    \item Show that, for each \emph{non-degenerate} $\mathcal{V}$ and $r>0$, there is a non-$\alpha$-acyclic $F$ with $\alpha_{\mathcal{V},r}(F)\neq 0$. 
\end{enumerate}

Indeed, we are able to prove our upper bounds (Theorem~\ref{intro:upper}) via 1.\ and 2.\ by bounding the generalised hypertreewidth of the quotients of hypergraphs that satisfy the respective Venn diagrams. 

However, the majority of the work of this paper, and our main technical contribution, is the proof of the lower bound\footnote{This resembles the situation in the realm of graph motif counting: establishing lower bounds via the framework of graph motif parameters turned out to be the main challenge in the majority of recent works, see e.g.\ \cite{FockeR24,DoringMW24,BressanGMR24}.} (Theorem~\ref{intro:lower}): Step 3.\ above presents the challenge of determining which hypergraphs $F$, given $\mathcal{V}$ and $r$, have a non-zero coefficient $\alpha_{\mathcal{V},r}(F)$. Taking a closer look at Equation~\eqref{eq:intro_uncoloured_coeff}, we observe that each term
\[ \auts{H}^{-1}\cdot\sum_{\substack{\rho \in P(V(H))\\H/\rho \cong F}} \mu(\rho)\]
contributes to the sum with a sign of $(-1)^{|V(H)|-|V(F)|}$ as $\mathsf{sign}(\mu(\rho))=(-1)^{|V(H)|-|V(F)|}$~\cite[Equation A.2]{Lov}. However, this means that the overall coefficient $\alpha_{\mathcal{V},r}(F)$ is an alternating sum, potentially allowing for cancellations, as there might be graphs $H_1$ and $H_2$ in $S(\mathcal{V},r)$ with $|V(H_1)|=|V(H_2)|\pm 1$ that both have $F$ as a quotient. Note that this is possible even if we restrict ourselves to uniform hypergraphs: Figure~\ref{fig:quotient_example} contains an example of $4$-uniform hypergraphs $H_1$ and $H_2$ with  $|V(H_1)|=|V(H_2)|+ 1$ that satisfy the same Venn diagram and  have a common quotient $H_3$.

This subtle, but crucial, observation makes the analysis of the support of the coefficient function $\alpha_{\mathcal{V},r}$ much more challenging.\footnote{We point out that Bressan et al.\ \cite{Bressanetal26} did not have to overcome this problem as they were working with individual subgraph patterns which allowed them to rule out any cancellations by relying solely on the sign of the M\"obius function.}  We ultimately overcome this problem by introducing an edge-colourful intermediate version of $\#\textsc{HyperMotif}(\mathcal{V})$, denoted by $\#\textsc{ColHyperMotif}(\mathcal{V})$, and we show that
\[\#\textsc{ColHyperMotif}(\mathcal{V}) \fptlinred \#\textsc{HyperMotif}(\mathcal{V})\,,\]
where $\fptlinred$ denotes parameterised linear-time reducibility. This will allow us to work solely with the colourful version $\#\textsc{ColHyperMotif}(\mathcal{V})$ for the purpose of our lower bounds. The technical details of the coloured version are provided in Section~\ref{sec2:prelims}; in a nutshell, we operate on input hypergraphs $G$ for which each hyperedge has one of $3$ colours, and we only count the subhypergraphs that satisfy the specified Venn diagram and contain precisely one edge per colour.

We are then able to express the number of colourful hypergraph motifs as a linear combination of colourful homomorphism counts via so-called \emph{hypergraph  exas}, defined in Section~\ref{sec:fractures}. The concept of fractures (on simple graphs) was originally introduced by the authors of~\cite{Peyerimhoffetal23} in order to express edge-colourful subgraph counts as linear combinations of edge-colourful homomorphism counts. While we are eventually able to translate this method from graphs to hypergraphs, we are, in the process, required to circumvent a variety of technical problems that arise, among others, due to the fact that~\cite{Peyerimhoffetal23} considers only simple graphs and, specifically, omits graphs with self-loops. 

Ultimately, the colourful approach allows us to drastically simplify the coefficient formula~\eqref{eq:intro_uncoloured_coeff}; specifically, we will be able to drop the terms $\auts{H}^{-1}$, which, while still yielding an alternating sum, gives us much more control of analysing the coefficients. Eventually, we show that for each non-degenerate Venn diagram $\mathcal{V}$, there exists a non-$\alpha$-acyclic hypergraph $F$ in the (colourful) homomorphism expansion of $\#\textsc{ColHyperMotif}(\mathcal{V})$.

Last but not least, we invoke Dedekind Interpolation~\cite[Appendix A]{Bressanetal26full} to show that linear combinations of \emph{colourful} hypergraph homomorphism counts also satisfy the property that the computation of the linear combination is precisely as hard as the computation of its hardest term. In summary, we establish the following chain of parameterised linear-time reductions: 

\begin{lemma}\label{intro:red_chain}
    Let $\mathcal{V}$ be a non-degenerate Venn diagram. There exists a non-$\alpha$-acyclic hypergraph $F$ such that
    \[\#\textsc{Hom}(\{F\}) \fptlinred \#\textsc{cpHom}(\{F\}) \fptlinred \colmotifprob{\mathcal{V}} \fptlinred \motifprob{\mathcal{V}}\,.\]
\end{lemma}
\noindent In the reduction chain above, $\#\textsc{cpHom}(\{F\})$ denotes the colourful version of counting hypergraph homomorphisms from $F$.

Theorem~\ref{intro:lower} then follows from the aforementioned work of Brault-Baron~\cite{BraultBaron13} (see also Mengel's survey~\cite{Mengel,Mengel_full}) stating that homomorphisms from non-$\alpha$-acyclic hypergraphs cannot be counted in near-linear time unless either one of the Triangle Hypothesis or the Hyperclique Hypothesis fail.

\subsection{Conclusion and Future Work}
We have provided an exhaustive and explicit classification of the fine-grained complexity of counting hypergraph motifs with $3$ hyperedges, precisely determining which cases admit (FPT-)near-linear-time algorithms under standard lower bound assumptions. Moreover, by translating the counting problem to the evaluation of a finite linear combination of homomorphism counts, we have obtained an (FPT-)near-quadratic-time algorithm for \emph{all} Venn diagrams that ultimately relies on an extension of Yannakakis' celebrated algorithm.

While we consider the case of $k=3$ solved, we strongly believe that the generalised hypergraph motif counting problem will benefit from further work. Specifically, we propose the search for an explicit tractability criterion as a promising future research direction. To facilitate the next steps, we emphasise that the framework of fractured hypergraphs introduced in our paper applies with only trivial modifications to generalised hypergraph motif counting problems, and we suspect that it will serve as a useful tool towards a full, and explicit, classification of $\#\textsc{GeneralisedHyperMotif}(\mathfrak{C})$.

\subsection{Organisation of the paper}
Section~\ref{sec2:prelims} introduces the necessary definitions and concepts used throughout this paper. Section~\ref{sec3:col_hom_basis} includes the set up of the coloured version of our problem and an introduction to Dedekind Interpolation for isolating (colourful) homomorphism counts with non-zero coefficients in finite linear combinations. In Section~\ref{sec4:non_deg_Venn}, we analyse the coefficients in said linear combinations for all non-degenerate Venn diagrams. Section~\ref{sec5:proof} provides the statement and proof of the obtained bounds for solving $\motifprob{\mathcal{V}}$. Finally, Section~\ref{sec6:general_case} addresses generalised hypergraph motifs and includes the proof of Theorem~\ref{thm:main_GHMs_intro}.

{\small
\tableofcontents
}
\pagebreak

\section{Preliminaries} \label{sec2:prelims}

For $k\in\mathbb{N}$, we write $[k]$ for the set $\{1,2,...,k\}$.
The cardinality of a set $S$ is denoted by $|S|$ or $\#S$. Given a function $f: A \to B$ and a subset $C \subseteq A$, we write $f|_C$ for the restriction of $f$ to $C$. Given a function $f: A \times B \to C$ and an element $a \in A$, we write $f(a,\star):B \to C$ for the function that maps $b$ to $f(a,b)$. Following standard notation in set theory, given a set of sets $J$, we write $\cup J$ for the union $\bigcup_{S \in J}S$.

A \textit{partition} $\rho$ of a finite\footnote{Naturally, partitions can also be defined for infinite sets; however, we will only consider partitions of finite sets in the present work.} set $S$ is a set of pairwise disjoint blocks $B_i\subseteq S$ such that $\bigcup B_i=S$. For two partitions $\rho$ and $\sigma$, we say $\rho$ \textit{refines} or is \textit{finer} than $\sigma$ if every block in $\rho$ is a subset of a block in $\sigma$. This is denoted by $\rho\leq \sigma$. On the other hand, $\sigma\geq \rho$ if $\sigma$ is \textit{coarser} than $\rho$. The partition that has all elements of the set contained in one block is $\{S\}=\top_S$, while the partition with each element in its own block is denoted $\bot_S$; we drop $S$ from the subscript if it is clear from the context. Finally, given a set $S$, we write $P(S)$ for the set of all partitions of $S$. The partially ordered set $(P(S),\leq)$ is called the \emph{partition lattice}\footnote{We will not need any specific properties of lattices and refer the reader to \cite{Stanley} for a comprehensive introduction to lattices.} of $S$.

\subsection{Hypergraphs}
We will use capital letters $F$, $G$, and $H$ to denote hypergraphs. A hypergraph $G$ is a pair of a finite set of vertices $V$ and a set of hyperedges $E$, where each hyperedge $e\in E$ is a non-empty subset of $V$. For the remainder of the paper, we will refer to hyperedges just as ``edges''. We write $V(G)$ and $E(G)$, respectively, to denote the vertex and edge set of $G$. For the purpose of this paper, we assume that hypergraphs do not have isolated vertices, that is, every vertex must be contained in at least one edge. Given a subset of edges $A \subseteq E(G)$, we write $G[A]$ for the hypergraph with vertices $\cup_{e\in A} e$ and edges $A$. Moreover, given a hypergraph $H$, we write $H\downarrow$ for the hypergraph obtained from $H$ by deleting all non-maximal edges, that is, we delete all edges $e$ such that $e\subsetneq e'$ for some $e' \in E(H)$.
The \emph{rank} of a hypergraph $G$, denoted by $\mathsf{rank}(G)$, is the maximum number of vertices contained in any edge of $G$.

We will also rely on quotients of hypergraphs as defined by Bressan et al.\ \cite{Bressanetal26}:
\begin{definition}[Hypergraph Quotients]
    \label{def:hypergraph_quotient}
    Let $G=(V,E)$ be a hypergraph and let $\sigma=\{B_1,\ldots,B_\ell\}$
    be a partition of
    $V$.
    For any set $X \subseteq V$, define
    \(X/\sigma := \{B_i \in \sigma \mid B_i \cap X \ne \emptyset\}.\)

    The \emph{quotient of~$G$ with respect to $\sigma$}, denoted by $G/\sigma$, has vertices $V(G/\sigma)=\sigma$, that is, each block of $\sigma$ becomes a vertex of $G/\sigma$. Moreover, $E(G/\sigma):=\{e/\sigma \mid e \in E(G)\}$, that is, a set of blocks becomes an edge of the quotient if there is an edge in $G$ that has a non-empty intersection with each of the blocks.
\end{definition}

\subsubsection{Homomorphisms, Embeddings, and Subgraphs}

A \emph{homomorphism} from a hypergraph $H$ to a hypergraph $G$ is a mapping $\varphi:V(H)\rightarrow V(G)$ such that for all edges $e=\{v_1,\dots,v_\ell\}\in E(H)$, we have $\varphi(e)=\{\varphi(v_1),\dots,\varphi(v_\ell)\}\in E(G)$. The term $\hom{H}{G} $ denotes the set of all homomorphisms from $H$ to $G$.

A homomorphism $\varphi$ from $H$ to $G$ is called an \emph{embedding} if it is injective; we write $\emb{H}{G}$ for the set of all embeddings from $H$ to $G$. Moreover, we call $\varphi$ an \emph{isomorphism} if it is bijective and additionally satisfies that $e \in E(H)$ if and only if $\varphi(e) \in E(G)$; $H$ and $G$ are called isomorphic, denoted by $H \cong G$, if an isomorphism from $H$ to $G$ exists. An isomorphism from a hypergraph $H$ to itself is called an automorphism, and we write $\aut{H}$ for the set of all automorphisms of $H$.

A \emph{subhypergraph} of $G=(V,E)$ is a hypergraph $G'=(V',E')$ with $V' \subseteq V$ and $E' \subseteq E$. We write $\sub{H}{G}$ for the set of all subgraphs of $G$ that are isomorphic to $H$.
The following well-known identity (cf.\ \cite[Section 3]{Bressanetal26full}) relates the number of subgraphs and the number of embeddings via the number of automorphisms:
 \begin{equation}
     \embs{H}{G}=\auts{H}\cdot\subs{H}{G}. \label{sub=emb/aut}
 \end{equation}

Moreover, lifting a classical result of Lov{\'{a}}sz (see~\cite[Chapter 5.2.3]{Lov}) from graphs to hypergraphs, the following transformation was shown by Bressan et al.\ \cite{Bressanetal26}:
 \begin{equation}
    \embs{H}{G}=\sum_{\rho\in P(V(H))} \mu(\rho)\cdot\homs{H/\rho}{G}, \label{embintermsofhom}
\end{equation}
where $\rho$ is a partition of the vertices of $H$ from the set of all partitions $P(V(H))$, and $\mu$ is the Möbius function of the partition lattice (see e.g.\ \cite[Chapter A.1]{Lov} and~\cite[Chapter 3]{Stanley}).

\subsubsection{Generalised Hypertreewidth and Hypergraph Acyclicity}
\label{sec:ghtw}
Hypertree decompositions lift the notion of tree decompositions from graphs to hypergraphs. In this work, we will rely on \emph{generalised} hypertreewidth:

\begin{definition}[Hypertree Decompositions and Generalised Hypertreewidth]
A \emph{hypertree decomposition} of a hypergraph $H$ is a pair $\mathcal{T}=(T,\{B_t\}_{t \in V(T)})$ where $T$ is a tree, and $\{B_t\}_{t \in V(T)}$ is a collection of subsets of $V(H)$, called \emph{bags}, such that the following conditions are satisfied:
\begin{enumerate}
    \item $\bigcup_{t \in V(T)} B_t = V(H)$.
    \item For each edge $e \in E(H)$, there is a bag $B_t$ such that $e \subseteq B_t$.
    \item For each vertex $v \in V(H)$, the subgraph $T[\{t \in V(T) \mid v \in B_t\}]$ is connected.
\end{enumerate}
The \emph{width} of a bag $B_t$ is the minimum number of edges of $H$ that cover $B_t$, that is
\[ \mathsf{width}(B_t)=\min\{|A| \mid A \subseteq E(H)~\wedge~B_t \subseteq \cup A\}\,.\]
The width of $\mathcal{T}$ is the maximum width of any bag of $\mathcal{T}$, and the \emph{generalised hypertreewidth} of $H$, denoted by $\mathsf{ghtw}(H)$, is the minimum width over all hypertree decompositions of $H$.  
\end{definition}

While there are several notions of hypergraph acyclicity present in the literature (see e.g.\ \cite{BeeriFMY83,Brault-Baron16}), we focus in this work on $\alpha$-\emph{acyclicity}. There are various equivalent definitions of $\alpha$-acyclicity, using e.g.\ GYO reductions~\cite{Graham1979,YuO79,BeeriFMY83} or join-trees~\cite{Yannakakis}. For this work, we rely on the characterisation of $\alpha$-acyclicity via generalised hypertreewidth~\cite{GottlobLS02}:
\begin{definition}[$\alpha$-acyclic hypergraphs]
    A hypergraph $H$ is called $\alpha$-acyclic if $\mathsf{ghtw}(H)\leq 1$.
\end{definition}

\begin{remark}
For the remainder of the paper, to avoid notational clutter, we will drop the $\alpha$ and just refer to \emph{acyclic} hypergraphs.
\end{remark}

\subsubsection{Colourings of Hypergraphs}
An $H$-coloured hypergraph is a pair ($G, c$) of a hypergraph $G$ and homomorphism $c\in\mathsf{Hom}(G\to H)$, called the $H$\emph{-colouring} of $G$.
Let $(F,c_F)$ and $(G,c_G)$ be $H$-coloured hypergraphs. We say that $(F,c_F)$ and $(G,c_G)$ are \emph{isomorphic} if there exists an isomorphism $\pi$ from $F$ to $G$ that additionally satisfies $c_G(v)=v_F(\pi(v))$ for all vertices $v \in V(G)$. A homomorphism $\varphi\in \hom{F}{G}$ is called \emph{colour-prescribed} if $c_F(v)=c_G(\varphi(v))$ for all $v\in V(F)$. We write $\cphom{(F,c_F)}{H}{(G,c_G)}$ for the set of all colour-prescribed homomorphisms from $(F,c_F)$ to $(G,c_G)$. Similarly as in the uncoloured case, we call a colour-prescribed homomorphism a \emph{colour-prescribed embedding} if it is injective (on the vertices); we denote the set of all colour-prescribed embeddings from $(F,c_F)$ to $(G,c_G)$ as $\cpemb{(F,c_F)}{H}{(G,c_G)}$.

\begin{obs}\label{obs:induced_edge-colouring}
    Let $(G,c)$ be an $H$-coloured hypergraph. Then $c$ also induces an edge-colouring of $G$ into (at most) $|E(H)|$ distinct colours by assigning an edge $e \in E(G)$ the colour $c(e)$ --- recall that $c(e)\in E(H)$ as $c$ is a homomorphism.
\end{obs}
For the purpose of this work, we will mostly restrict ourselves to edgewise-injectively coloured hypergraphs, defined below:

\begin{definition} Let $(G,c)$ be an $H$-coloured hypergraph. We call $(G,c)$
an \emph{edgewise-injectively $H$-coloured hypergraph} if $c|_e$ is injective for all $e \in E(G)$, that is, no edge of $G$ contains two vertices with the same colour. We write $\mathbb{G}_H$ for the set of all isomorphism types of edgewise-injectively $H$-coloured graphs.
\end{definition}

\begin{remark}[$\cphom{H}{H}{(G,c)}$]
    We consider the graph $H$ itself as an edgewise-injectively $H$-coloured graph $(H,\mathsf{id}_H)$, where $\mathsf{id}_H$ is the identify on $V(H)$. To avoid notational clutter, we allow ourselves to just write $\cphom{H}{H}{(G,c)}$, rather than $\cphom{(H,\mathsf{id}_H)}{H}{(G,c)}$.
\end{remark}

\subsubsection{Fractures of Hypergraphs}
\label{sec:fractures}
Similar to fractures in graphs \cite{Peyerimhoffetal23}, we define hypergraph fractures and fractured hypergraphs.

\begin{definition}[Hypergraph Fractures]
    Let $H$ be a hypergraph. A \textit{fracture} of $H$ is a tuple of partitions
    $$\vecrho=(\rho_v)_{v\in V(H)},$$
    where $\rho_v$ is a partition of the set of edges incident to $v$. We write $\mathcal{L}(H)$ for the set of all fractures of $H$.

    We say that a fracture $\vecrho\in \mathcal{L}(H)$ \emph{refines} a fracture $\vecsigma\in \mathcal{L}(H)$, denoted by $\vecrho\leq \vecsigma$ if $\rho_v \leq \sigma_v$ for all $v \in V(H)$. The \emph{lattice of fractures} of a hypergraph $H$ is the partially ordered set $(\mathcal{L}(H),\leq)$.
\end{definition}

Each fracture of $H$ defines a fractured hypergraph; in a nutshell, this operation can be thought of as an inverse of the quotient operation where vertices are split, rather than merged, and the partitions of the fracture guides on how the splitting operation is executed. The formal definition is provided below:
\begin{definition}[Fractured Hypergraphs $H \natural \vecrho$]
Given a fracture $\vecrho=(\rho_v)_{v\in V(H)}$ of a hypergraph $H$, we define the \textit{fractured hypergraph} $H\natural\vecrho$ as follows:
\begin{itemize}
    \item The vertex set of $H\natural\vecrho$ is
    $V(H\natural\vecrho)=\{v_B\ |\ B\in\rho_v, v\in V(H)\}$.
    \item The edges of $H\natural\vecrho$ are defined as follows: For each edge $e=\{v^1,\dots,v^\ell\}$ of $H$ and $i\in [\ell]$, let $B_i$ be the block of $\rho_{v^i}$ that contains $e$. We then add $\{v^1_{B_1},\dots,v^\ell_{B_\ell}\}$ as an edge to $H\natural\vecrho$. 
\end{itemize}
\end{definition}

Figure \ref{fig:fractured_hypergraph_example} shows the fractured hypergraph $H\natural\vecrho$ obtained by taking the fracture of hypergraph $H$ with respect to $\vecrho=\big(\{\{e_1\},\{e_2\}\},\top,\dots,\top\big)$, that is $\vecrho_a=\{\{e_1\},\{e_2\}\}$, and $\vecrho_v = \top$ for all $v \in V(H)\setminus\{a\}$. The first partition of the fracture specifies that vertex $a$ is split into two copies, one appearing in each of its incident edges. The partitions corresponding to the remaining vertices are $\top$, meaning that all incident edges are contained in a single block. Thus, no other vertex is split.

\begin{figure}[t]
    \centering
    \includegraphics[width=0.65\linewidth]{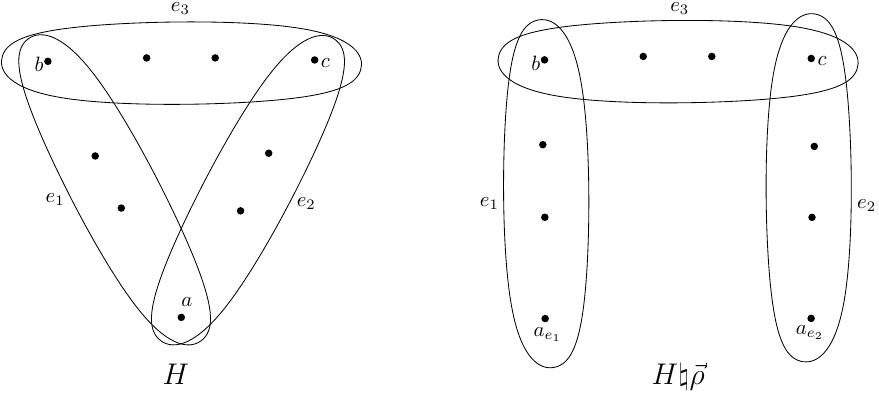}
    \caption{Hypergraph $H$ and the fractured hypergraph $H\natural\vecrho$ with respect to the fracture $\vecrho=\big(\{\{e_1\},\{e_2\}\},\top,\dots,\top\big).$}
    \label{fig:fractured_hypergraph_example}
\end{figure}

We observe that any fractured graph $H \natural \vecrho$ admits a canonical edgewise-injective $H$-colouring:
\begin{obs}
    Given a hypergraph $H$ and a fracture $\vecrho$ of $H$, the mapping 
    \begin{align*}
        c_{H,\vecrho}: V(H\natural \vecrho) &\to V(H)\\
        v_B & \mapsto v
    \end{align*}
    is an edgewise-injective $H$-colouring of $H \natural \vecrho$. 
\end{obs}

\begin{remark}[$\cphom{H \natural \vecrho}{H}{(G,c_G)}$]
    Following the previous observation, we will always consider fractured graphs $H \natural \vecrho$ as edgewise-injectively $H$-coloured graphs w.r.t.\ $c_{H,\vecrho}$. To streamline notation, we will thus omit explicitly writing $c_{H,\vecrho}$ and set
    \[ \cphom{H \natural \vecrho}{H}{(G,c_G)} := \cphom{(H \natural \vecrho,c_{H,\vecrho})}{H}{(G,c_G)}\]
\end{remark}

\subsubsection{Venn Diagrams and Hypergraph Motifs}\label{sec:venns}
Following \cite{Lee}, informally, a \emph{Venn diagram} specifies which intersections of three given sets are non-empty, and the \emph{motifs} in a hypergraph $G$ specified by a given Venn diagram are precisely those $3$-edge subsets $\{e_1,e_2,e_3\}$ of $G$ whose intersections match the Venn diagram. We make this more formal below; consider Figures~\ref{fig:vd} and~\ref{fig:vdex} for an illustration, including an example.
\begin{definition}[Venn diagrams and $\motif{\mathcal{V}}{G}$]
    A \emph{Venn diagram} is a vector $\mathcal{V}\in \{0,1\}^7$, indexed by non-empty subsets of $[3]$, that is
    \[\mathcal{V}=(\mathcal{V}_J)_{\emptyset \neq J\subseteq [3]}\,.\]
    We say that a $3$-edge hypergraph $G$ satisfies $\mathcal{V}$, denoted by $\mathcal{V}(G)=1$, if there is an ordering $e_1,e_2,e_3$ of $E(G)$ such that for all non-empty $J \subseteq [3]$ we have
    \[\mathcal{V}_J=1 \Leftrightarrow\left(\bigcap_{i \in J}e_i \right)\setminus \left(\bigcup_{i \in [3]\setminus J} e_i\right)  \,. \]
    Furthermore, we define
    \[\motif{\mathcal{V}}{G}:=\{A \subseteq E(G) \mid \#A=3~\wedge~\mathcal{V}(G[A])=1\}\]
    for the set of all \emph{motifs} of $G$ satisfying $\mathcal{V}$.
\end{definition}

\begin{figure}
    \centering
    \includegraphics[width=0.25\linewidth]{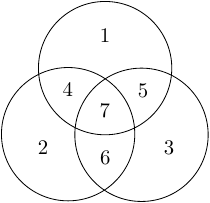}
    \caption{A Venn diagram with its 7 sections}
    \label{fig:vd}
\end{figure}

\begin{figure}
    \centering
    \includegraphics[width=0.5\linewidth]{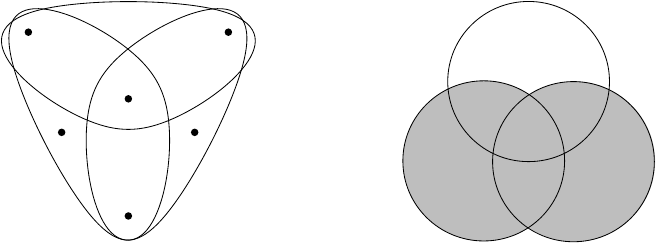}
    \caption{A hypergraph and its associated Venn diagram. It can also be encoded as the binary vector, $(0,1,1,1,1,1,1)$ assuming sections are ordered with respect to the numbering used in Figure \ref{fig:vd}.}
    \label{fig:vdex}
\end{figure}
Similarly as in the work of Lee et al.\ \cite{Lee}, we only consider Venn diagrams that correspond to connected motifs, and we exclude Venn diagrams that are equivalent up to reordering of sets and those with with repeated (identical) edges. This leaves us (as well as Lee at al.\ \cite{Lee}) with precisely $26$ Venn diagrams as depicted in Figure~\ref{fig:allVenns}.

For our lower bounds, we will also rely on \emph{colourful motifs}, defined below. To this end, recall from Observation~\ref{obs:induced_edge-colouring} that an $H$-colouring $c$ of a hypergraph $G$ induces also a colouring of the edges of $G$.
\begin{definition}[$\colmotif{\mathcal{V}}{H}{(G,c)}$]
Let $H$ be a hypergraph, and let $(G,c)$ be an edgewise-injectively $H$-coloured hypergraph. Let furthermore $\mathcal{V}$ be a Venn diagram. We define
    $$\colmotif{\mathcal{V}}{H}{(G,c)} := \{A \subseteq E(G) \mid c(A)=E(H) ~\wedge~|A|=3 ~\wedge~\mathcal{V}(G[A])=1 \} \,.$$
\end{definition}

Next, for our complexity classification, we will distinguish on whether or not a Venn diagram is degenerate; Figure~\ref{fig:allVenns} depicts all non-degenerate Venn diagrams $\mathcal{V}_1,\dots,\mathcal{V}_{14}$, highlighted in red, and all degenerate Venn diagrams $\mathcal{V}_{15},\dots,\mathcal{V}_{26}$, highlighted in green.
\begin{definition}[Degenerate Venn Diagrams]\label{def:degenerate}
    A Venn diagram $\mathcal{V}$ is called \emph{degenerate} if every hypergraph~$G$ with~$\mathcal{V}(G)=1$ satisfies that there are edges $e_1$ and $e_2$ of $G$ such that $e_1 \subseteq e_2$.
\end{definition}

\subsection{Parameterised and Fine-grained Complexity Theory}\label{sec:param_and_finegrained}
For this work, we use the standard notation $\tilde{O}$ for hiding poly-logarithmic factors. We call algorithms running in time $\tilde{O}(n)$ and $\tilde{O}(n^2)$, respectively, \emph{near-linear-time} and \emph{near-quadratic-time} algorithms.

\subsubsection{Fixed-Parameter Tractability and Reductions}
We provide a brief introduction to fixed-parameter tractability and refer the reader to the standard textbook \cite{Cygan} for a comprehensive introduction.
A \emph{parameterised counting problem} is a pair $(P,\kappa)$ of a function (the counting problem) $P:\{0,1\}^*\to\mathbb{Q}$ and a parameterisation $\kappa:\{0,1\}^*\to\mathbb{N}$. For example, the problem $\#\textsc{Clique}$ asks, given as input a pair of a (non-hyper)graph $G$ and a positive integer $k$, to compute the number of complete $k$-vertex subgraphs of $G$, and the parameterisation is given by $\kappa(G,k)=k$. Following standard conventions, we will say that $k$ is the parameter of the problem.

A \emph{fixed-parameter tractable} (``$\mathsf{FPT}$'') algorithm for a parameterised counting problem $(P,\kappa)$ is a deterministic algorithm that, on input $x$, computes $P(x)$ in time $f(\kappa(x))\cdot |x|^{O(1)}$ for some computable function~$f$. We say that $(P,\kappa)$ is fixed-parameter tractable (FPT) if it can be solved by an FPT algorithm.

We call an FPT algorithm an \emph{FPT-near-linear time algorithm} if it runs in time $f(\kappa(x))\cdot \tilde{O}(|x|)$, and we call it an \emph{FPT-near-quadratic time algorithm} if it runs in time $f(\kappa(x))\cdot \tilde{O}(|x|^2)$.

A \emph{parameterised Turing-reduction} from $(P_1,\kappa_1)$ to $(P_2,\kappa_2)$ is an FPT algorithm for $(P_1,\kappa_1)$ with oracle access to $P_2$, satisfying that there is a computable function $g$ such that, on input $x$, for any oracle query $y$ posed to $P_2$ it holds that $\kappa_2(y)\leq g(\kappa_1(x))$, that is, the parameter of each oracle query can only depend on the parameter of the input. We write $(P_1,\kappa_1) \leq^{\mathsf{FPT}} (P_2,\kappa_2)$ if a parameterised Turing-reduction exists.

We also need a linear-time version of parameterised Turing-reductions. To this end, we will make use of the notion introduced by Aivasiliotis, G\"obel, Roth, and, Schmitt\ \cite{paramHolants,paramHolantsArxiv}:
\begin{definition}[Parameterised Linear-Time Reductions]
A parameterised Turing-reduction from
$(P_1,\kappa_1)$ to $(P_2,\kappa_2)$ is called a \emph{linear-time} parameterised Turing-reduction if there exists a computable function $f$ such that, on input $x$, the reduction runs in time $f(\kappa_1(x))\cdot O(|x|)$ and poses at most $f(\kappa_1(x))$ oracle queries. We write $(P_1,\kappa_1)\fptlinred(P_2,\kappa_2)$ if a linear-time parameterised Turing-reduction exists.
\end{definition}
\noindent Aivasiliotis et al.\ \cite[Lemma 2.8]{paramHolantsArxiv} prove that, $(P_1,\kappa_1)\fptlinred(P_2,\kappa_2)$ and $(P_2,\kappa_2)$ being solvable in FPT-near-linear time implies that $(P_1,\kappa_1)$ is solvable in FPT-near-linear time as well.

\subsubsection{Lower Bounds and Hardness Assumptions}
A parameterised counting problem $(P,\kappa)$ is $\#\mathsf{W}[1]$-hard if $\#\textsc{Clique}\leq^{\mathsf{FPT}}_\mathsf{T}(P,\kappa)$, and it is known that $\#\mathsf{W}[1]$-hard are not fixed-parameter tractable, unless the Exponential Time Hypothesis, defined below, fails~\cite{Chenetal05,Chenetal06}.

\begin{conj}[ETH~\cite{ImpagliazzoP01}]
    The \emph{Exponential Time Hypothesis} (ETH) asserts that the problem $3\textsc{-SAT}$ cannot be solved in time $\exp(o(n))$ where $n$ is the number of variables of the input formula.
\end{conj}

While ETH is sufficient for ruling out FPT algorithms, we need to rely on more specific lower bound assumptions from fine-grained complexity theory to rule out (near-)linear-time algorithms.

\begin{conj}[Triangle Hypothesis; cf.\ \cite{AbboudW14}]
    The \emph{Triangle Hypothesis} asserts that there is no algorithm that, given an input graph $G$ with $m$ edges, decides in time $\tilde{O}(m)$ if $G$ contains a triangle.
\end{conj}

\begin{conj}[Hyperclique Hypothesis; cf.\ \cite{Mengel_full,BraultBaron13}]
The \emph{Hyperclique Hypothesis} asserts that, for no pair $k>h>2$ of integers, there is an $\varepsilon>0$ and an algorithm that, given a $h$-uniform hypergraph $H$ with $n$ vertices, decides in time $\tilde{O}(n^{k-\varepsilon})$ whether $H$ contains a hyperclique\footnote{A hyperclique of size $k$ in an $h$-uniform hypergraph is a $k$-vertex subset $V'$ of $H$ such that ever subset $S\subseteq V'$ of size $h$ is an edge of $H$.} of size $k$.
\end{conj}

\subsubsection{Problem Definitions}

We associate each Venn diagram $\mathcal{V}$ with its (parameterised) motif counting problem, including the coloured version, defined below:

\begin{mdframed}
\#\textsc{HyperMotif}($\mathcal{V}$) \\
\textbf{Input:} A hypergraph $G$ \\
\textbf{Parameter:} $\mathsf{rank}(G)$ \\
\textbf{Output:} $\#\motif{\mathcal{V}}{G}$
\end{mdframed}  

\begin{mdframed}
\#\textsc{ColHyperMotif}($\mathcal{V}$) \\
\textbf{Input:} An edgewise-injectively $H$-coloured hypergraph $(G,c)$ \\
\textbf{Parameter:} $\mathsf{rank}((G,c))$ \\
\textbf{Output:} $\colmotifs{\mathcal{V}}{H}{(G,c)}$
\end{mdframed}

\noindent For the purpose of our proofs we will also consider the following two homomorphism counting problems. For both problems, $\mathcal{H}$ denotes a recursively enumerable class of graphs and, similarly to the Venn diagram~$\mathcal{V}$ for the above problems, $\mathcal{H}$ is part of the problem definition, and not part of the input; specifically, each class~$\mathcal{H}$ defines a unique problem:

\begin{mdframed}
\#\textsc{Hom}($\mathcal{H}$) \\
\textbf{Input:} A hypergraph $G$, and a hypergraph $H \in \mathcal{H}$ \\
\textbf{Parameter:} $|H|$\\
\textbf{Output:} $\homs{H}{G}$
\end{mdframed}  

\begin{mdframed}
\#\textsc{cpHom}($\mathcal{H}$) \\
\textbf{Input:} A hypergraph $H \in \mathcal{H}$, and an edgewise-injectively $H$-coloured hypergraph $(G,c)$ \\
\textbf{Parameter:} $|H|$ \\
\textbf{Output:} $\cphoms{H}{H}{(G,c)}$
\end{mdframed}

\section{A Colourful Homomorphism Basis for Fractured Hypergraphs}
\label{sec3:col_hom_basis}
Recall the definition of the edge-colourful version of $\motifs{\mathcal{V}}{G}$:
$$\colmotifs{\mathcal{V}}{H}{(G,c)} := \#\{A \subseteq E(G) \mid c(A)=E(H) ~\wedge~|A|=3 ~\wedge~\mathcal{V}(G[A])=1 \} \,.$$
Our goal in the section is twofold: first, we demonstrate how the function $(G,c) \mapsto \colmotifs{\mathcal{V}}{H}{(G,c)}$ can be translated into a linear combination of colour-prescribed homomorphism counts. Second, we invoke Dedekind Interpolation to show that the evaluation of said linear combination is tractable in linear time if and only if all colour-prescribed homomorphism counts with a non-zero coefficient can be computed in linear time.

As an initial step towards the first goal, given an an edgewise-injectively $H$-coloured hypergraph $(G,c)$, we show that, for any fracture $\vec{\sigma}$ of $H$, the term $\cphoms{H\natural\vec{\sigma}}{H}{(G,c)}$ can be expressed as a zeta-transformation of injective colour-prescribed homomorphisms over the lattice of fractures of $H$:
\begin{lemma} \label{lem:hom=sum(emb)}
    Let $H$ be a hypergraph with $|E(H)|=3$, and $G$ an edgewise-injectively $H$-coloured hypergraph equipped with $H$-colouring $c$. For all fractures $\vec{\sigma}$ of $H$, we can write that for the fractured hypergraph $H\natural\vec{\sigma}$,
    \begin{equation}
    \cphoms{H\natural\vec{\sigma}}{H}{(G,c)} = \sum_{\vec{\rho}\geq\vec{\sigma}}\cpembs{H\natural\vec{\rho}}{H}{(G,c)}, \label{hom=sum(emb)}
    \end{equation}
    where the relation $\geq$, is over the lattice of fractures $\mathcal{L}(H)$. 
\end{lemma}
\begin{proof}
    The proof is largely analogous to the graph case (\cite{Peyerimhoffetal23}, Claim 4.2). Define an equivalence relation $\sim$ on the set $\cphom{H\natural\vec{\sigma}}{H}{(G,c)}$, where $\varphi_1\sim\varphi_2$ if and only if
    \begin{equation}\label{eq:zeta_help}
        \forall u,v\in V(H\natural\vec{\sigma}): \varphi_1(u)=\varphi_1(v)\iff \varphi_2(u)=\varphi_2(v)\,.
    \end{equation}
    Writing $[\varphi]$ for the equivalence class represented by $\varphi$ we obtain:
    \begin{equation*}
        \cphoms{H\natural\vec{\sigma}}{H}{(G,c)} = \sum_{[\varphi]}|[\varphi]|.
    \end{equation*}
    Note that, since any $\varphi\in \cphom{H\natural\vec{\sigma}}{H}{(G,c)}$ is colour-prescribed, $\varphi(u)=\varphi(v)$ is only possible if $u=x^B$ and $v=x^{B'}$ for some vertex $x \in V(H)$ and blocks $B,B' \in \vec{\sigma}_x$. Consequently, any $\varphi \in \cphom{H\natural\vec{\sigma}}{H}{(G,c)}$ induces a coarsening $\vec{\rho}^\varphi \geq \vec{\sigma}$ as follows:
    \begin{itemize}
        \item For any $x \in V(H)$ and blocks $B,B' \in \vec{\sigma}_x$, we merge $B$ and $B'$ if and only if $\varphi(x^B)=\varphi(x^{B'})$. The resulting partition is $\vec{\rho}^\varphi_x$.
        \item $\vec{\rho}^\varphi=(\vec{\rho}_x^\varphi)_{x \in V(H)}$.
    \end{itemize}
    Now observe that $\varphi_1 \sim\varphi_2$ (see \eqref{eq:zeta_help}) if and only if $\vec{\rho}^{\varphi_1} = \vec{\rho}^{\varphi_2}$. Moreover, as argued before, all $\vec{\rho}^{\varphi}$ are coarsenings of $\vec{\sigma}$. Hence, we can represent the equivalence classes of $\sim$ by $\vec{\rho}\geq \vec{\sigma}$ as follows: \[C\left(\vec{\rho}\right) = \{\varphi \in \cphom{H\natural\vec{\sigma}}{H}{(G,c)} \mid \vec{\rho}^{\varphi}=\vec{\rho}\}\,.\] 
    This yields
    \begin{equation*}
        \cphoms{H\natural\vec{\sigma}}{H}{(G,c)} = \sum_{\vec{\rho}\geq \vec{\sigma}}|C(\vec{\rho})|.
    \end{equation*}
    Finally, we observe that there is a one-to-one correspondence between homomorphisms in $C(\vec{\rho})$ and injective homomorphisms in $\cpemb{H \natural \vec{\rho}}{H}{(G,c)}$: the bijection between those sets is given by the mapping that assigns $\varphi \in C(\vec{\rho})$ the injective homomorphism $b(\varphi) \in \cpemb{H \natural \vec{\rho}}{H}{(G,c)}$ that maps a vertex $x^{\hat{B}}$ to $\varphi(x^B)$, where $B$ is one of the blocks of $\vec{\sigma}_x$ that has been merged into $\hat{B}$.

    As a consequence, we obtain
    \begin{equation*}
        \cphoms{H\natural\vec{\sigma}}{H}{(G,c)} = \sum_{\vec{\rho}\geq \vec{\sigma}}\cpembs{H \natural \vec{\rho}}{H}{(G,c)}\,,
    \end{equation*}
    concluding the proof.
\end{proof}

Next, for expressing $\colmotifs{\mathcal{V}}{H}{(G,c)}$ as a linear combination of colour-prescribed homomorphisms, we will consider a subset of the lattice of fractures depending on which fractures graphs satisfy a given Venn diagram:
\begin{definition}[$\mathcal{L}(H,\mathcal{V})$]
    Let $H$ be a hypergraph with $|E(H)|=3$, and let $\mathcal{V}$ be a Venn diagram. We write $\mathcal{L}(H,\mathcal{V})$ for the set of fractures $\vec{\sigma}$ of $H$ such that $H\natural\vec{\sigma}$ satisfies $\mathcal{V}$.
\end{definition}

For the next step, we first note that every element $A\in \colmotif{\mathcal{V}}{H}{(G,c)}$ induces a fracture $\vec{\sigma}^A \in \mathcal{L}(H,\mathcal{V})$ as follows: Let $x \in V(H)$ and let $e,f$ be two edges of $H$ incident to $x$. From $A\in \colmotif{\mathcal{V}}{H}{(G,c)}$ is follows specifically that $c(A)=E(H)$ and $|A|=|E(H)|=3$. Hence $c$ induces a bijection from $A$ to $E(H)$ and $e$ and $f$ have unique pre-images $\hat{e}$ and $\hat{f}$ under $c$. We include $e$ and $f$ in the same block of $\vec{\sigma}^A_x$ if and only if $\hat{e}$ and $\hat{f}$ are connected to the same endpoint in $c^{-1}(x)$ --- note that it is well-defined to speak of \emph{the} endpoint of $\hat{e}$ and $\hat{f}$ in $c^{-1}(x)$ as $(G,c)$ is edgewise-injectively $H$-coloured. We further observe that this construction yields a canonical isomorphism from $H\natural\vec{\sigma}^A$ to $G[A]$, showing that $\vec{\sigma}^A \in \mathcal{L}(H,\mathcal{V})$ as $\mathcal{V}(G[A])=1$. 

\begin{lemma}\label{lem:colmotifs_to_homs}
    Let $\mathcal{V}$ be a Venn diagram, let $H$ be a hypergraph with $|E(H)|=3$, and let $(G,c)$ be an edgewise-injectively $H$-coloured hypergraph. We have
    \begin{equation}
    \colmotifs{\mathcal{V}}{H}{(G,c)}=\sum_{\vec{\sigma}\in\mathcal{L}(H,\mathcal{V})}\ \sum_{\vec{\rho}\geq\vec{\sigma}}\mu(\vec{\sigma},\vec{\rho})\cdot\cphoms{H\natural\vec{\rho}}{H}{(G,c)},
    \end{equation}
    where the Möbius function $\mu(\vec{\sigma},\vec{\rho})$ is over the lattice of fractures.\label{lem:colmotif}
\end{lemma}
\begin{proof}
We define an equivalence relation on $\colmotif{\mathcal{V}}{H}{(G,c)}$ by setting
\[ A \simeq A' :\Leftrightarrow \vec{\sigma}^A = \vec{\sigma}^{A'}\,.\]
Now, given the equivalence class $C$ represented by $\vec{\sigma} \in \mathcal{L}(H,\mathcal{V})$, we note that the members of $C$ are precisely those $3$-edge subsets $A$ of $E(G)$ such that $c(A)=E(H)$ and $G[A]$ is isomorphic to $H \natural \vec{\sigma}$. Thus, there is a bijection from $C$ to $\cpemb{H\natural\vec{\sigma}}{H}{(G,c)}$ which maps $A \in C$ to the injective homomorphism $\varphi$ that maps the edges of $H\natural\vec{\sigma}$ to $A$; this uniquely identifies $\varphi$ as it is colour-prescribed. Consequently, we have
\begin{equation*}
    \colmotifs{\mathcal{V}}{H}{(G,c)}=\sum_{\vec{\sigma}\in\mathcal{L}(H,\mathcal{V})}\cpembs{H\natural\vec{\sigma}}{H}{(G,c)}\,.
\end{equation*}
Next, we apply M\"obius inversion (see~\cite[Chapter 3.7]{Stanley}) on the zeta-transformation given by Lemma~\ref{lem:hom=sum(emb)} and obtain:
\begin{equation*}
   \cpembs{H\natural\vec{\sigma}}{H}{(G,c)} =\sum_{\vec{\rho}\geq\vec{\sigma}}\mu(\vec{\sigma},\vec{\rho})\cdot\cphoms{H\natural\vec{\rho}}{H}{(G,c)}\,,
\end{equation*}
where $\mu$ is the M\"obius function of the lattice of fractures. This concludes the proof.
\end{proof}

\begin{corollary}\label{cor:top_coefficient}
Let $\mathcal{V}$ be a Venn diagram and let $H$ be a hypergraph with $|E(H)|=3$. There is a unique computable function $\coeff_{H,\mathcal{V}}: \mathcal{L}(H) \to \mathbb{Z}$ such that
\[\colmotifs{\mathcal{V}}{H}{\star} = \sum_{\vec{\sigma} \in\mathcal{L}(H)} \coeff_{H,\mathcal{V}}(\vec{\sigma}) \cdot \cphoms{H\natural\vec{\sigma}}{H}{\star}\,. \]
Moreover, for the coarsest fracture $\vec{\rho}=\vec{\top}$ we have that
\[\coeff_{H,\mathcal{V}}(\vec{\top}) =\sum_{\vecsigma\in\mathcal{L}(H,\mathcal{V})}~\prod_{v \in V(H)} (-1)^{|\sigma_v|-1} \cdot (|\sigma_v|-1)!\]
\end{corollary}
\begin{proof}
The existence and uniqueness of $\coeff_{H,\mathcal{V}}$ follows from Lemma~\ref{lem:colmotifs_to_homs} by collecting coefficients of isomorphic terms. For $\vec{\top}$, we collect the coefficients of $\cphoms{H \natural \vec{\top}}{H}{\star}$ in Lemma~\ref{lem:colmotifs_to_homs} and obtain
\[ \coeff_{H,\mathcal{V}}(\vec{\top}) = \sum_{\vecsigma\in\mathcal{L}(H,\mathcal{V})}\mu(\vecsigma,\vec\top)\] 
Similarly as for the case of graphs~\cite[Section 4]{Peyerimhoffetal23}, we use the facts\footnote{See, for instance, Chapter 3 in the standard textbook of Stanley~\cite{Stanley}} that the lattice of fractures is, by definition, the product of $|V(H)|$ partition lattices, and that the M\"obius function is multiplicative over products of posets. We thus have:
\[\sum_{\vecsigma\in\mathcal{L}(H,\mathcal{V})}\mu(\vecsigma,\vec\top)= \sum_{\vecsigma\in\mathcal{L}(H,\mathcal{V})}\prod_{v \in V(H)} \mu(\vecsigma_v,\vec\top_v) = \sum_{\vecsigma\in\mathcal{L}(H,\mathcal{V})}\prod_{v \in V(H)} (-1)^{|\sigma_v|-1} \cdot (|\sigma_v|-1)! \,,\]
where the last equation follows from the explicit formula for the M\"obius function of the partition lattice (cf.\ Chapter 3 in~\cite{Stanley}).
\end{proof}

\subsection*{Coefficient Isolation via Dedekind Interpolation}
Recall that our overall strategy for proving the lower bounds for $\motifs{\mathcal{V}}{G}$ relies in the first step on a translation of the coloured version $\colmotifs{\mathcal{V}}{H}{(G,c)}$ into a linear combination of colour-prescribed homomorphism counts as established in Corollary~\ref{cor:top_coefficient}. In the second step, our aim is to show that such a linear combination is at least as hard as computing the hardest term \[(G,c) \mapsto\cphoms{H \natural \vec\sigma}{H}{(G,c)}\,.\]
Said second step will be achieved with Dedekind Interpolation. To this end, we use the abstract circuit-based framework established by Bressan et al.\ (see Appendix A in the full version~\cite{Bressanetal26full} of \cite{Bressanetal26}), requiring us to take a brief detour to group theory.

A \emph{semigroup} $(\mathbb{G},\ast)$ consists of a ground set $\mathbb{G}$ and an associate operation $\ast: \mathbb{G}^2 \to \mathbb{G}$. We say that $(\mathbb{G},\ast)$ is \emph{computable} if $\mathbb{G}$ and $\ast$ are computable. For the purpose of this work, we call a mapping $\mathfrak{h}: \mathbb{G} \to \mathbb{Q}$ a \emph{semigroup homomorphism} if $\mathfrak{h}(g_1 \ast g_2)= \mathfrak{h}(g_1)\cdot \mathfrak{h}(g_2)$.

\begin{definition}[Dedekind Circuits (Definition A.1 in \cite{Bressanetal26full})]\label{def:dedekind_circuits}
    For a semigroup $(\mathbb{G},\ast)$ and pairwise distinct semigroup homomorphisms $\mathfrak{h}_1,\dots\mathfrak{h}_k:\mathbb{G} \to \mathbb{Q}$, a \emph{Dedekind Circuit} $\mathbb{D}(\mathfrak{h}_1,\dots,\mathfrak{h}_k)$ is an arithmetic circuit over $\mathbb{Q}$ with 
    \begin{itemize}
        \item $k$ output gates $\mathsf{out}_1,\dots,\mathsf{out}_k$; and
        \item for finitely many $g \in \mathbb{G}$, an \emph{input gate} $\mathsf{in}[g]$ labelled with $g$.
        \item Further, for each rational linear combination $F=a_1 \mathfrak{h}_1 + \dots + a_k \mathfrak{h}_k$, the circuit $\mathbb{D}(\mathfrak{h}_1,\dots,\mathfrak{h}_k)$ when each input gate $\mathsf{in}[g]$ is assigned the value $F(g)$, outputs $\mathsf{out}_i=a_i$ for all $i \in [k]$.
    \end{itemize}
\end{definition}

\begin{theorem}[Dedekind Interpolation (Theorem A.3 in \cite{Bressanetal26full})]\label{thm:dedekind}
  Let $(\mathbb{G},\ast)$ be a computable semigroup and let $\mathfrak{h}_1,\dots,\mathfrak{h}_k: \mathbb{G} \to \mathbb{Q}$ be computable semigroup homomorphisms. There exists an algorithm $\mathbb{A}$ with the following properties:
  \begin{enumerate}
      \item $\mathbb{A}$ receives as input
      \begin{itemize}
          \item semigroup elements $g_1,\dots,g_k\in \mathbb{Q}$ with $\mathfrak{h}_i(g_i)\neq 0$ for all $i \in[k]$, and
          \item for each pair $i,j \in [k]$ with $i<j$, a semigroup element $g_{i,j}\in \mathbb{G}$ with $\mathfrak{h}_i(g_{i,j})\neq \mathfrak{h}_j(g_{i,j})$.
      \end{itemize}
      \item $\mathbb{A}$ computes a Dedekind Circuit $\mathbb{D}(\mathfrak{h}_1,\dots,\mathfrak{h}_k)$ of depth $O(k)$ with constant fan-in.
  \end{enumerate}
\end{theorem}

\noindent For applying Dedekind Interpolation to our setting, we will rely on Tensor products: Tensor products have been defined for hypergraphs~\cite{Bressanetal26} and for coloured graphs~\cite{Peyerimhoffetal23}. We now define a version for $H$-coloured hypergraphs.

\begin{definition}\label{def:coloured_hypergraph_tensor}
    Let $H$ be a hypergraph and $(F,c_F)$ and $(G,c_G)$ be two edgewise-injectively $H$-coloured hypergraphs, with $H$-colourings $c_F$ and $c_G$, respectively. Then, the \textit{colour-preserving hypergraph tensor product} $(F,c_F)\otimes (G,c_G)$ is defined as follows: 
\begin{align*}
    V((F,c_F)\otimes(G,c_G)) &= \{(u,v)\in V((F,c_F))\times V((G,c_G))\ |\ c_F(u)=c_G(v)\} \\
    E((F,c_F)\otimes(G,c_G)) &= \{e\subseteq V((F,c_F))\times V((G,c_G)) \ |\ \pi_F(e)\in E((F,c_F)) \wedge \pi_G(e)\in E((G,c_G))\}
\end{align*}
where $\pi_F(e)$ and $\pi_G(e)$ are the sets of all first components of vertices in $e$ and second components of vertices in $e$, respectively. That is, for an element $(u,v)\in e$, $\pi_F((u,v))=u$ and $\pi_G((u,v))=v$.
\end{definition}

\begin{lemma}\label{lem:tensor_efficient}
    Given two edgewise-injectively $H$-coloured graphs, $(G,c_G)$ and $(F,c_F)$, the tensor product $(G,c_G)\otimes(F,c_F)$ can be computed in time $O(r^2\cdot\|G\|\cdot\|F\|),$ where $r=$max(rank$(G)$, rank$(F)$), and $\|G\|$ denotes the sum of the number of vertices and the number of edges in the hypergraph $G$. 
\end{lemma}
\begin{proof}
    The tensors are calculated as described in Algorithm \ref{tensorrunningtime}. Computing the vertex set $V((G,c_G)\otimes(F,c_F))$ requires time $|V(G)|\cdot|V(F)|$. Then going over all elements of this set and removing pairs that do not have the same colour introduces another factor of $|V(G)|\cdot|V(F)|$.
    The running time to loop over the set of edges of both hypergraphs will be $|E(G)|\cdot|E(F)|$. This is in turn bounded by $\|G\|\cdot\|F\|$. Finally, adding the corresponding edge $e$ into the tensor will contribute a factor of $|e_G|\cdot|e_F|=r^2$. Note that to preserve colours, we must have $|e_G|=|e_F|$.
    Hence our running time is given by
    $$2|V(G)|\cdot|V(F)| + r^2\cdot\|G\|\cdot\|F\|\leq O( r^2\cdot\|G\|\cdot\|F\|).$$
\end{proof}

\algrenewcommand\algorithmicrequire{\textbf{Input:}}
\algrenewcommand\algorithmicensure{\textbf{Output:}}
\begin{algorithm}[t]
\caption{Computing the Tensor product of edgewise-injectively $H$-coloured hypergraphs}
\label{tensorrunningtime}
\begin{algorithmic}[1]
\Require Edgewise-injectively $H$-coloured hypergraphs $(G,c_G)$ and $(F,c_F)$
\Ensure The tensor product $(G,c_G)\otimes(F,c_F)$
\State Compute the vertex set of $(G,c_G)\otimes(F,c_F)$ as $
V((G,c_G)) \times V((F,c_F))$.

\State Remove all pairs whose vertices have different colours.

\For{$e_G \in E(G)$}
    \For{$e_F \in E(F)$}
        \State Create an edge
        $e \in E((G,c_G)\otimes(F,c_F))$
        iff $\pi_G(e)=e_G$ and $\pi_F(e)=e_F$.
    \EndFor
\EndFor
\end{algorithmic}
\end{algorithm}

For what follows, recall that $\mathbb{G}_H$ denotes the set of all isomorphism types of edgewise-injectively $H$-coloured hypergraphs.

\begin{lemma}\label{lem:tensor_props}
For each $H$, the pair $(\mathbb{G}_H,\otimes)$ is a computable semigroup. Moreover, for each fracture $\vec\sigma$ of $H$, the mapping $\mathfrak{h}_{\vec\sigma} : (G,c_G) \mapsto \cphoms{H \natural\vec\sigma}{H}{(G,c_G)}$ is a semigroup homomorphism.
\end{lemma}
\begin{proof}
For the first claim, we observe that the mapping $(u,(v,w))\mapsto((u,v),w)$ is an isomorphism from $(G_1,c_1)\otimes((G_2,c_2) \otimes (G_3,c_3))$ to $((G_1,c_1)\otimes(G_2,c_2)) \otimes (G_3,c_3)$; hence $\otimes$ is associative and thus $(\mathbb{G}_H,\otimes)$ is a semigroup. Moreover, clearly, $\mathbb{G}_H$ and $\varphi_{\vec\sigma}$ are computable.

For the second claim, we need to show that
\[\cphoms{H\natural\vec{\sigma}}{H}{(G,c_G)\otimes(F,c_F)} = \cphoms{H\natural\vec{\sigma}}{H}{(G,c_G)} \cdot \cphoms{H\natural\vec{\sigma}}{H}{(F,c_F)} \]
Consider the mapping $b: \varphi \mapsto (\varphi_1,\varphi_2)$ where, for $i \in \{1,2\}$ , $\varphi_i(v):=\pi_i(\varphi(v))$ and $\pi_i$ is the projection of a pair to the $i$-th component. We claim that $b$ is a bijection from $\cphom{H\natural\vec{\sigma}}{H}{(G,c_G)\otimes(F,c_F)}$ to $ \cphom{H\natural\vec{\sigma}}{H}{(G,c_G)} \times \cphom{H\natural\vec{\sigma}}{H}{(F,c_F)}$.

To this end, let $\varphi \in \cphom{H\natural\vec{\sigma}}{H}{(G,c_G)\otimes(F,c_F)}$. For a vertex $u\in V(H\natural\vec{\sigma})$ let $(x,y)=\varphi(u)\in V((G,c_G)\otimes(F,c_F))$, that is $x=\varphi_1(u)$ and $y=\varphi_2(u)$. Then by definition of the colour-preserving Tensor product, $c_G(x)=c_F(y)=c_{H\natural\vec\sigma}(u)$ --- recall that the $H$-colouring $c_{H\natural\vec\sigma}$ maps vertices $x^B$, for $B$ being a block of $\vec{\sigma}_x$, to $x \in V(H)$. Thus both $\varphi_1$ and $\varphi_2$ preserve the $H$-colouring. We also have to show that $\varphi_1$ and $\varphi_2$ preserve edges:
let $e\in E(H\natural\vecsigma)$ be an edge. Since $\varphi$ is a homomorphism, we have that $\varphi(e)$ is an edge in $(G,c_G)\otimes(F,c_F)$. Again, applying the definition of the tensor product, we have $$\varphi_1(e)=\pi_1(\varphi(e))\in E((G,c_G)) \text{  and  }\varphi_2(e)=\pi_2(\varphi(e))\in E((F,c_F))\,,$$
implying that the co-domain of $b$ is, as intended, $\cphom{H\natural\vec{\sigma}}{H}{(G,c_G)} \times \cphom{H\natural\vec{\sigma}}{H}{(F,c_F)}$.

For proving that $b$ is injective, we consider the mapping $b'$ that maps a pair $$(\varphi_1,\varphi_2)\in \cphom{H\natural\vec{\sigma}}{H}{(G,c_G)} \times \cphom{H\natural\vec{\sigma}}{H}{(F,c_F)}$$ to the function $\varphi: u \mapsto (\varphi_1(u),\varphi_2(u))$. If $\varphi\in \cphom{H\natural\vec{\sigma}}{H}{(G,c_G)\otimes(F,c_F)}$ then the co-domain of $b'$ is $\cphom{H\natural\vec{\sigma}}{H}{(G,c_G)\otimes(F,c_F)}$. Then, clearly, $b \circ b'$ and $b \circ b'$ are both the identity, hence $b'=b^{-1}$ and $b$ is a bijection. Thus it remains to show that $\varphi$ is indeed contained in $\cphom{H\natural\vec{\sigma}}{H}{(G,c_G)\otimes(F,c_F)}$. 

To this end, we first observe that $\varphi$ preserves colours as both $\varphi_1$ and $\varphi_2$ do. Next, given an edge $e=\{x_1,\dots,x_\ell\} \in E(H \natural \vecsigma)$, let $(u_i,v_i)=\varphi(x_i)$ for all $i \in [\ell]$, that is $u_i=\varphi_1(x_i)$ and $v_i=\varphi_2(x_i)$. We have to show that $\hat{e}:=\{(u_1,v_1),\dots,(u_\ell,v_\ell)\}$ is an edge of $(G,c_G) \otimes (F,c_F)$. As $\varphi_1$ and $\varphi_2$ are both homomorphisms we have that $e_1:=\{u_1,\dots,u_\ell\}$ and $e_2:=\{v_1,\dots,v_\ell\}$ are edges of $(G,c_G)$ and $(F,c_F)$ respectively. However, note that $\pi_G(\hat{e})=\pi_1(\hat{e})=e_1$ and $\pi_F(\hat{e})=\pi_2(e)=e_2$, and hence $\hat{e}$ is, by definition, an edge of the Tensor product $(G,c_G) \otimes (F,c_F)$.
\end{proof}

We are now able to apply Dedekind Interpolation to $\colmotifs{\mathcal{V}}{H}{\star}$.

\begin{lemma}\label{lem:dedekind_applied}
    There is a computable function $f$ and an oracle algorithm $R$ with the following properties:
    \begin{enumerate}
        \item The input of $R$ is a Venn diagram $\mathcal{V}$, a hypergraph $H$ and an edgewise-injectively $H$-coloured hypergraph $(G,c_G)$.
        \item $R$ has oracle access to the function $\colmotifs{\mathcal{V}}{H}{\star}$.
        \item $R$ computes $\cphoms{H\natural\vecsigma}{H}{(G,c_G)}$ for all $\vecsigma$ with $\coeff_{H,\mathcal{V}}(\vecsigma)\neq 0$.
        \item The running time of $R$ is bounded by $O(f(\mathcal{V},H)\cdot |(G,c)|)$ and it makes at most $f(\mathcal{V},H)$ oracle queries.
    \end{enumerate}
\end{lemma}
\begin{proof}
    Let $\mathcal{V}$, $H$, and $(G,c_G)$ be the input of $R$ as specified above.
    Let $r=\mathsf{rank}(H)$ and observe that $r$ must also be an upper bound of the rank of $G$ as $(G,c_G)$ is an edgewise-injectively $H$-coloured graph. For each fracture $\vec{\sigma}$ of $H$, we consider the function $\mathfrak{h}_{\vecsigma}: \mathbb{G}_H \to \mathbb{Q}$
    \[\mathfrak{h}_{\vec{\sigma}}(\hat{G},\hat{c}) := \cphoms{H \natural \vec{\sigma}}{H}{(\hat{G},\hat{c})}\,.\]
    Observe that: 
    \begin{itemize}
        \item[(I)] $\mathfrak{h}_{\vec{\sigma}}(H\natural_{\vec\sigma},c_{H \natural \vecsigma})>0$ as the identity mapping is contained in $\cphom{H \natural \vec{\sigma}}{H}{(H\natural_{\vec\sigma},c_{H \natural \vecsigma})}$.
        \item[(II)] For $\vecsigma \neq \vecrho$ we have that $(H\natural_{\vec\sigma},c_{H \natural \sigma})$ and $(H\natural_{\vecrho},c_{H \natural \vecrho})$ are not isomorphic --- recall that isomorphisms between $H$-coloured graphs must preserve the colouring. Moreover, $\mathfrak{h}_{\vec{\sigma}}\neq\mathfrak{h}_{\vec{\rho}}$, and we can find, in time only depending on $H$, an edgewise-injectively $H$-coloured graph $(\hat{G},\hat{c})$ such that $$\mathfrak{h}_{\vec{\sigma}}(\hat{G},\hat{c})\neq\mathfrak{h}_{\vec{\rho}}(\hat{G},\hat{c})$$.
    \end{itemize}
    Consequently, in time only depending on $H$ and $\mathcal{V}$, our algorithm $R$ can run algorithm $\mathbb{A}$ from Theorem~\ref{thm:dedekind} and obtain a Dedekind Circuit $\mathbb{D}(\mathfrak{h}_{\vecsigma})_{\vecsigma \in \mathcal{L}(H)}$ of depth $O(|\mathcal{L}(H)|)$ and constant fan-in. Let $\ell$ be the number of input gates of $\mathbb{D}(\mathfrak{h}_{\vecsigma})_{\vecsigma \in \mathcal{L}(H)}$, and let $(G_1,c_1),\dots,(G_\ell,c_\ell)$ denote the labels of the input gates.

    In order to make use of $\mathbb{D}(\mathfrak{h}_{\vecsigma})_{\vecsigma \in \mathcal{L}(H)}$, we define the function 
    \[F: \mathbb{G}_H \to \mathbb{Q};~~(\hat{G},\hat{c})\mapsto \colmotifs{\mathcal{V}}{H}{(G,c_G)\otimes(\hat{G},\hat{c})} \,. \]
    Next we note that, by Lemma~\ref{lem:colmotifs_to_homs} and Lemma~\ref{lem:tensor_props}, we have
    \begin{align*}
        \forall (\hat{G},\hat{c})\in \mathbb{G}_H: F(\hat{G},\hat{c})&= \colmotifs{\mathcal{V}}{H}{(G,c_G)\otimes(\hat{G},\hat{c})}\\
        ~&= \sum_{\vecsigma \in \mathcal{L}(H)} \coeff_{H,\mathcal{V}}(\vecsigma) \cdot \cphoms{H \natural \vecsigma}{H}{(G,c)} \cdot \cphoms{H \natural \vecsigma}{H}{(G,c)}\\
        ~& = \sum_{\vecsigma \in \mathcal{L}(H)} a_{\vecsigma} \cdot \mathfrak{h}_{\vec{\sigma}}(\hat{G},\hat{c})\,,
    \end{align*}
    where $a_{\vecsigma}:= \coeff_{H,\mathcal{V}}(\vecsigma) \cdot \cphoms{H \natural \vecsigma}{H}{(G,c)}$.
    Using our oracle, we can compute the values $F(G_i,c_i)$ for $i\in[\ell]$ --- recall that $\ell$ only depends on $H$ and $\mathcal{V}$. Note that, for $n=|(G,c)|$, we can generously bound $F(G_i,c_i)\leq g(H,\mathcal{V})\cdot n^{g(H,\mathcal{V})}$ for some computable function $g$. Thus, the bit size required for storing those values is bounded by $\log(g(H,\mathcal{V})\cdot n^{g(H,\mathcal{V})}) \in O(g(H,\mathcal{V})\cdot \log n)$. Moreover, by Lemma~\ref{lem:tensor_efficient}, we can compute the required Tensor products in time $O(r^2 \cdot ||(G_i,c_i)|| \cdot n)$.

    Next, we insert the value $F(G_i,c_i)$ to the $i$-th input gate for all $i\in[\ell]$ and evaluate the circuit $\mathbb{D}(\mathfrak{h}_{\vecsigma})_{\vecsigma \in \mathcal{L}(H)}$. 
    As $\mathbb{D}(\mathfrak{h}_{\vecsigma})_{\vecsigma \in \mathcal{L}(H)}$ has constant fan-in and depth only depending on $H$ and $\mathcal{V}$, and as the bit size of the input gates is bounded by $O(g(H,\mathcal{V})\cdot \log n)$, the overall time required for the evaluation of the circuit is bounded generously by $O(g'(H,\mathcal{V})\cdot \mathsf{polylog}(n))$ for some computable function $g'$. 
    
    The evaluation of $\mathbb{D}(\mathfrak{h}_{\vecsigma})_{\vecsigma \in \mathcal{L}(H)}$ yields, by definition of Dedekind Circuits, the values $(a_{\vec\sigma})_{\vecsigma \in \mathcal{L}(H)}$ at the output gates. Finally, recalling that $a_{\vecsigma}:= \coeff_{H,\mathcal{V}}(\vecsigma) \cdot \cphoms{H \natural \vecsigma}{H}{(G,c)}$, the algorithm $R$ outputs, for all $\vecsigma$ with $\coeff_{H,\mathcal{H}}(\vec{\sigma})$,
    \[ \frac{a_{\vecsigma}}{\coeff_{H,\mathcal{V}}(\vecsigma)} = \cphoms{H \natural \vecsigma}{H}{(G,c)}\,, \]
    concluding the proof.
\end{proof}

\section{The Homomorphism Basis for Non-Degenerate Venn Diagrams}
\label{sec4:non_deg_Venn}
In Corollary~\ref{cor:top_coefficient}, we have shown that
\[ \coeff_{H,\mathcal{V}}(\vec{\top}) = \sum_{\vecsigma\in\mathcal{L}(H,\mathcal{V})}\prod_{v \in V(H)} (-1)^{|\sigma_v|-1} \cdot (|\sigma_v|-1)! \,.\]
In the current section, for each non-degenerate Venn diagram $\mathcal{V}$, we will find a non-$\alpha$-acyclic hypergraph $H$ satisfying that $\coeff_{H,\mathcal{V}}(\vec{\top})\neq 0$. As it turns out, we will be able to always choose $H$ to be obtainable from either a triangle or from a hyperclique minus one edge by adding degree-$1$ vertices to edges; we formalise those notions below:

\begin{definition}[$\Delta(j_1,j_2,j_3)$ and $\Gamma(j_1,j_2,j_3)$; cf.\ Figure~\ref{fig:hypergraphs_H}]
Let $j_1,j_2,j_3$ be non-negative integers. We write $\Delta$ to denote the triangle, that is, the (hyper)graph with vertices $\{a,b,c\}$ and edges $e_1=\{a,b\}$, $e_2=\{a,c\}$, and $e_3=\{b,c\}$. Moreover, we write $\Gamma$ to denote hypergraph obtained by deleting one edge from the $3$-uniform $4$-hyperclique, that is, $\Gamma$ is the hypergraph with vertices $\{a,b,c,d\}$ and edges $e_1=\{a,b,d\}$, $e_2=\{a,b,c\}$, and $e_3=\{a,c,d\}$.
\begin{itemize}
    \item The hypergraph $\Delta(j_1,j_2,j_3)$ is obtained from $\Delta$ by adding $j_i$ fresh vertices of degree $1$ to the edge $e_i$ for all $i \in \{1,2,3\}$.
    \item The hypergraph $\Gamma(j_1,j_2,j_3)$ is obtained from $\Gamma$ by adding $j_i$ fresh vertices of degree $1$ to the edge $e_i$ for all $i \in \{1,2,3\}$.
\end{itemize}
\end{definition}

\begin{figure}
    \centering
    \begin{subfigure}{0.35\linewidth}
        \includegraphics[width=\linewidth]{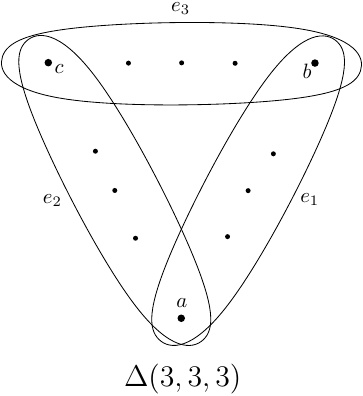}
        \caption{}
        \label{fig:Delta}
    \end{subfigure}
    ~~~~~~~~~
    \begin{subfigure}{0.31\linewidth}
        \includegraphics[width=\linewidth]{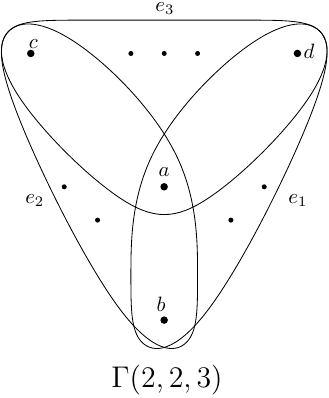}
        \caption{}
        \label{fig:Gamma}
    \end{subfigure}
    \caption{Examples of $\Delta(j_1,j_2,j_3)$ and $\Gamma(j_1,j_2,j_3)$.}
    \label{fig:hypergraphs_H}
\end{figure}

Note that neither of $\Delta(j_1,j_2,j_3)$ and $\Gamma(j_1,j_2,j_3)$ are $\alpha$-acyclic as they reduce, respectively, to the non-empty hypergraphs $\Delta(0,0,0)$ and $\Gamma(0,0,0)$, via GYO reductions\footnote{After Graham~\cite{Graham1979} and Yu and \"Ozsoyoglu~\cite{YuO79}; defined via successive deletion of vertices of degree $1$ and edges fully contained in other edges.} (see~\cite{BeeriFMY83}). Hence we obtain:
\begin{obs}\label{obs:delta_gamma_cyclic}
    For all $j_1,j_2,j_3 \geq 0$ we have $\mathsf{ghtw}(\Delta(j_1,j_2,j_3))>1$ and $\mathsf{ghtw}(\Gamma(j_1,j_2,j_3))>1$.
\end{obs}

In the subsequent subsections, we will show that for each of the $14$ non-degenerate Venn diagrams (see Figure~\ref{fig:allVenns}) either $\Delta(j_1,j_2,j_3)$ or $\Gamma(j_1,j_2,j_3)$ survives with a non-zero coefficient for some $j_1,j_2,j_3 \geq 0$.

\subsection{Warm up: Coefficients of $\mathcal{V}_1$, $\mathcal{V}_2$, $\mathcal{V}_3$, and $\mathcal{V}_4$}
\label{sec4.1}

Recall the following Venn diagrams:\\

\noindent \begin{tikzpicture}[scale=0.5, transform shape]
    \def\r{1.5}
    \coordinate (A) at (0,0);
    \coordinate (B) at (2,0);
    \coordinate (C) at (1,1.732);

    \begin{scope}
        \clip (A) circle (\r);
        \fill[pattern=\nel, pattern color=black] (A) circle (\r);
    \end{scope}
   
    \begin{scope}
        \clip (B) circle (\r);
        \fill[pattern=\nel, pattern color=black] (B) circle (\r);
    \end{scope}

    \begin{scope}
        \clip (C) circle (\r);
        \fill[pattern=\nel, pattern color=black] (C) circle (\r);
    \end{scope}

    \begin{scope}
        \clip (A) circle (\r);
        \clip (B) circle (\r);
        \fill[pattern=\nel, pattern color=black] (-2,-2) rectangle (4,4);
    \end{scope}

    \begin{scope}
        \clip (A) circle (\r);
        \clip (C) circle (\r);
        \fill[pattern=\nel, pattern color=black] (-1,-1) rectangle (4,4);
    \end{scope}

    \begin{scope}
        \clip (B) circle (\r);
        \clip (C) circle (\r);
        \fill[pattern=\nel, pattern color=black] (0,0) rectangle (4,4);
    \end{scope}

    \begin{scope}
        \clip (A) circle (\r);
        \clip (B) circle (\r);
        \clip (C) circle (\r);
        \fill[fill=white, pattern color=black] (0,0) rectangle (4,4);
    \end{scope}

    \draw[line width=1.2pt] (A) circle (\r);
    \draw[line width=1.2pt] (B) circle (\r);
    \draw[line width=1.2pt] (C) circle (\r);

    \node (V_1) at (1.1,-2.5) {\huge $\mathcal{V}_1$};
\end{tikzpicture}
\begin{tikzpicture}[scale=0.5, transform shape]
    \def\r{1.5}
    \coordinate (A) at (0,0);
    \coordinate (B) at (2,0);
    \coordinate (C) at (1,1.732);

    \begin{scope}
        \clip (A) circle (\r);
        \fill[pattern=\nel, pattern color=black] (A) circle (\r);
    \end{scope}
   
    \begin{scope}
        \clip (B) circle (\r);
        \fill[pattern=none, pattern color=black] (B) circle (\r);
    \end{scope}

    \begin{scope}
        \clip (C) circle (\r);
        \fill[pattern=\nel, pattern color=black] (C) circle (\r);
    \end{scope}

    \begin{scope}
        \clip (A) circle (\r);
        \clip (B) circle (\r);
        \fill[pattern=\nel, pattern color=black] (-2,-2) rectangle (4,4);
    \end{scope}

    \begin{scope}
        \clip (A) circle (\r);
        \clip (C) circle (\r);
        \fill[pattern=\nel, pattern color=black] (-1,-1) rectangle (4,4);
    \end{scope}

    \begin{scope}
        \clip (B) circle (\r);
        \clip (C) circle (\r);
        \fill[pattern=\nel, pattern color=black] (0,0) rectangle (4,4);
    \end{scope}

    \begin{scope}
        \clip (A) circle (\r);
        \clip (B) circle (\r);
        \clip (C) circle (\r);
        \fill[fill=white, pattern color=black] (0,0) rectangle (4,4);
    \end{scope}

    \draw[line width=1.2pt] (A) circle (\r);
    \draw[line width=1.2pt] (B) circle (\r);
    \draw[line width=1.2pt] (C) circle (\r);

    \node (V_2) at (1.1,-2.5) {\huge $\mathcal{V}_2$};
\end{tikzpicture}
\begin{tikzpicture}[scale=0.5, transform shape]
    \def\r{1.5}
    \coordinate (A) at (0,0);
    \coordinate (B) at (2,0);
    \coordinate (C) at (1,1.732);

    \begin{scope}
        \clip (A) circle (\r);
        \fill[pattern=none, pattern color=black] (A) circle (\r);
    \end{scope}
   
    \begin{scope}
        \clip (B) circle (\r);
        \fill[pattern=none, pattern color=black] (B) circle (\r);
    \end{scope}

    \begin{scope}
        \clip (C) circle (\r);
        \fill[pattern=\nel, pattern color=black] (C) circle (\r);
    \end{scope}

    \begin{scope}
        \clip (A) circle (\r);
        \clip (B) circle (\r);
        \fill[pattern=\nel, pattern color=black] (-2,-2) rectangle (4,4);
    \end{scope}

    \begin{scope}
        \clip (A) circle (\r);
        \clip (C) circle (\r);
        \fill[pattern=\nel, pattern color=black] (-1,-1) rectangle (4,4);
    \end{scope}

    \begin{scope}
        \clip (B) circle (\r);
        \clip (C) circle (\r);
        \fill[pattern=\nel, pattern color=black] (0,0) rectangle (4,4);
    \end{scope}

    \begin{scope}
        \clip (A) circle (\r);
        \clip (B) circle (\r);
        \clip (C) circle (\r);
        \fill[fill=white, pattern color=black] (0,0) rectangle (4,4);
    \end{scope}

    \draw[line width=1.2pt] (A) circle (\r);
    \draw[line width=1.2pt] (B) circle (\r);
    \draw[line width=1.2pt] (C) circle (\r);

    \node (V_3) at (1.1,-2.5) {\huge $\mathcal{V}_3$};
\end{tikzpicture}
\begin{tikzpicture}[scale=0.5, transform shape]
    \def\r{1.5}
    \coordinate (A) at (0,0);
    \coordinate (B) at (2,0);
    \coordinate (C) at (1,1.732);

    \begin{scope}
        \clip (A) circle (\r);
        \fill[pattern=none, pattern color=black] (A) circle (\r);
    \end{scope}
   
    \begin{scope}
        \clip (B) circle (\r);
        \fill[pattern=none, pattern color=black] (B) circle (\r);
    \end{scope}

    \begin{scope}
        \clip (C) circle (\r);
        \fill[pattern=none, pattern color=black] (C) circle (\r);
    \end{scope}

    \begin{scope}
        \clip (A) circle (\r);
        \clip (B) circle (\r);
        \fill[pattern=\nel, pattern color=black] (-2,-2) rectangle (4,4);
    \end{scope}

    \begin{scope}
        \clip (A) circle (\r);
        \clip (C) circle (\r);
        \fill[pattern=\nel, pattern color=black] (-1,-1) rectangle (4,4);
    \end{scope}

    \begin{scope}
        \clip (B) circle (\r);
        \clip (C) circle (\r);
        \fill[pattern=\nel, pattern color=black] (0,0) rectangle (4,4);
    \end{scope}

    \begin{scope}
        \clip (A) circle (\r);
        \clip (B) circle (\r);
        \clip (C) circle (\r);
        \fill[fill=white, pattern color=black] (0,0) rectangle (4,4);
    \end{scope}

    \draw[line width=1.2pt] (A) circle (\r);
    \draw[line width=1.2pt] (B) circle (\r);
    \draw[line width=1.2pt] (C) circle (\r);

    \node (V_4) at (1.1,-2.5) {\huge $\mathcal{V}_4$};
\end{tikzpicture}\\

These Venn diagrams all represent cyclic hypergraphs with three edges. Adjacent edges share one vertex, while each edge itself may contain vertices of degree one. Crucially, the Venn diagrams in this group have in common that they forbid the presence of a vertex in the intersection of all three edges. In that way, $V_1,\dots,V_4$ naturally model the hypergraphs $\Delta(j_1,j_2,j_3)$.

For the proof of the following lemma, we observe that fractures of $\Delta(j_1,j_2,j_3)$ are obtained by splitting any subset of the three degree-2 vertices (see vertices $\{a,b,c\}$ in Figure~\ref{fig:Delta}). Note also that any fracture $\vec{\sigma}$ of $\Delta(j_1,j_2,j_3)$ must satisfy that $\vec{\rho}_x=\{\{e\}\}$ for any degree-1 vertex $x$ with $E(x)=\{e\}$ of $\Delta(j_1,j_2,j_3)$, as the only partition of the set $\{e\}$ is $\{\{e\}\}$. Specifically, when investigating the fractured graphs of $\Delta(j_1,j_2,j_3)$, we only have to check which of the degree-2 vertices $\{a,b,c\}$ are split.

\begin{lemma}
    We have
    \begin{itemize}
        \item  $\coeff_{\cycle(1,1,1),\mathcal{V}_1}(\vec{\top})\neq 0$.
        \item  $\coeff_{\cycle(1,1,0),\mathcal{V}_2}(\vec{\top})\neq 0$.
        \item  $\coeff_{\cycle(1,0,0),\mathcal{V}_3}(\vec{\top})\neq 0$.
        \item  $\coeff_{\cycle(0,0,0),\mathcal{V}_4}(\vec{\top})\neq 0$.
    \end{itemize}
    \label{lem:coeff_V1_to_V4}
\end{lemma}
\begin{proof}
Let us label the three degree-2 vertices $a,b,c$, and the three edges $e_1, e_2, e_3$ as in Figure \ref{fig:Delta}.

For $\mathcal{V}_1$, we will pick the hypergraph $H=\cycle(1,1,1)$ and we observe that $\mathcal{V}_1(H)=1$. As $H=H\natural\vec{\top}$ we thus have $\vec\top \in \mathcal{L}(H,\mathcal{V})$. For computing the coefficient $\coeff_{\cycle(1,1,1),\mathcal{V}_1}(\vec{\top})$, we sum over all fractures of $H$ that satisfy the Venn diagram. However, we observe that any fracture $\vec{\sigma}$ of $H$ that splits\footnote{Formally, we say that $\vec\sigma$ splits a degree-2 vertex $a$ with $E(a)=\{e_1,e_2\}$ if $\vec\sigma_a=\{\{e_1\},\{e_2\}\}$.} $a,b$ or $c$ will no longer satisfy the Venn diagram; Figure \ref{fig:fractured_hypergraph_example} illustrates splitting vertex $a$ and observe that the resulting fractured hypergraph no longer has the three pairwise intersections that $\mathcal{V}_1$ requires. Hence, $\mathcal{L}(H,\mathcal{V}_1)=\{\vec\top\}$ and thus the only term that is needed to calculate the coefficient for $\mathcal{V}_1$ is $\vec\top$.
As $|\vec\top_v|=1$ for all $v\in V(H)$, we have 
\begin{align*}
    \coeff_{\cycle(1,1,1),\mathcal{V}_1}(\vec{\top})&= \sum_{\vecsigma\in\mathcal{L}(H,\mathcal{V}_1)}~\prod_{v \in V(H)} (-1)^{|\vec\sigma_v|-1} \cdot (|\vec\sigma_v|-1)!\\
    &= \prod_{v \in V(H)} (-1)^{|\vec\top_v|-1} \cdot (|\vec\top_v|-1)!\\
    &= ((-1)^0\cdot0!)^{|V(H)|} \\
    &=1 \neq 0.
\end{align*}

For $\mathcal{V}_2,\mathcal{V}_3,$ and $\mathcal{V}_4$, we proceed with an analogous argument; the only difference is that we choose $\Delta(1,1,0)$, $\Delta(1,0,0)$, and $\Delta(0,0,0)$ to accommodate the restrictions on the absence of degree-1 vertices enforced by the Venn diagrams.

\end{proof}

\subsection{Coefficients of $\mathcal{V}_5$ and $\mathcal{V}_6$}

\begin{tikzpicture}[scale=0.5, transform shape]
    \def\r{1.5}
    \coordinate (A) at (0,0);
    \coordinate (B) at (2,0);
    \coordinate (C) at (1,1.732);

    \begin{scope}
        \clip (A) circle (\r);
        \fill[pattern=\nel, pattern color=black] (A) circle (\r);
    \end{scope}
   
    \begin{scope}
        \clip (B) circle (\r);
        \fill[pattern=\nel, pattern color=black] (B) circle (\r);
    \end{scope}

    \begin{scope}
        \clip (C) circle (\r);
        \fill[pattern=\nel, pattern color=black] (C) circle (\r);
    \end{scope}

    \begin{scope}
        \clip (A) circle (\r);
        \clip (B) circle (\r);
        \fill[pattern=\nel, pattern color=black] (-2,-2) rectangle (4,4);
    \end{scope}

    \begin{scope}
        \clip (A) circle (\r);
        \clip (C) circle (\r);
        \fill[fill=white, pattern color=black] (-1,-1) rectangle (4,4);
    \end{scope}

    \begin{scope}
        \clip (B) circle (\r);
        \clip (C) circle (\r);
        \fill[pattern=\nel, pattern color=black] (0,0) rectangle (4,4);
    \end{scope}

    \begin{scope}
        \clip (A) circle (\r);
        \clip (B) circle (\r);
        \clip (C) circle (\r);
        \fill[fill=white, pattern color=black] (0,0) rectangle (4,4);
    \end{scope}

    \draw[line width=1.2pt] (A) circle (\r);
    \draw[line width=1.2pt] (B) circle (\r);
    \draw[line width=1.2pt] (C) circle (\r);

    \node (V_5) at (1.1,-2.5) {\huge $\mathcal{V}_5$};
\end{tikzpicture}
\begin{tikzpicture}[scale=0.5, transform shape]
    \def\r{1.5}
    \coordinate (A) at (0,0);
    \coordinate (B) at (2,0);
    \coordinate (C) at (1,1.732);

    \begin{scope}
        \clip (A) circle (\r);
        \fill[pattern=\nel, pattern color=black] (A) circle (\r);
    \end{scope}
   
    \begin{scope}
        \clip (B) circle (\r);
        \fill[pattern=none, pattern color=black] (B) circle (\r);
    \end{scope}

    \begin{scope}
        \clip (C) circle (\r);
        \fill[pattern=\nel, pattern color=black] (C) circle (\r);
    \end{scope}

    \begin{scope}
        \clip (A) circle (\r);
        \clip (B) circle (\r);
        \fill[pattern=\nel, pattern color=black] (-2,-2) rectangle (4,4);
    \end{scope}

    \begin{scope}
        \clip (A) circle (\r);
        \clip (C) circle (\r);
        \fill[fill=white, pattern color=black] (-1,-1) rectangle (4,4);
    \end{scope}

    \begin{scope}
        \clip (B) circle (\r);
        \clip (C) circle (\r);
        \fill[pattern=\nel, pattern color=black] (0,0) rectangle (4,4);
    \end{scope}

    \begin{scope}
        \clip (A) circle (\r);
        \clip (B) circle (\r);
        \clip (C) circle (\r);
        \fill[fill=white, pattern color=black] (0,0) rectangle (4,4);
    \end{scope}

    \draw[line width=1.2pt] (A) circle (\r);
    \draw[line width=1.2pt] (B) circle (\r);
    \draw[line width=1.2pt] (C) circle (\r);

    \node (V_6) at (1.1,-2.5) {\huge $\mathcal{V}_6$};
\end{tikzpicture} \\

These Venn diagrams denote hypergraphs with three edges that form a `path-like' structure (for example, see $H\natural\vecrho$ in Figure~\ref{fig:fractured_hypergraph_example}. There exists an edge that has a common vertex with each of the other two edges, resulting a sequence of three connected edges. We will obtain hypergraphs satisfying $\mathcal{V}_5$ and $\mathcal{V}_6$ by taking a fracture of $\cycle(1,1,1)$ and $\cycle(1,1,0)$, respectively.

\begin{lemma}
    We have
    \begin{itemize}
        \item  $\coeff_{\cycle(1,1,1),\mathcal{V}_5}(\vec{\top})\neq 0$.
        \item  $\coeff_{\cycle(1,1,0),\mathcal{V}_6}(\vec{\top})\neq 0$.
    \end{itemize}
    \label{lem:coeff_V5_to_V6}
\end{lemma}

\begin{proof}
For $\mathcal{V}_5$, we use $H=\cycle(1,1,1)$ with vertices and edges labelled as in Figure \ref{fig:Delta}. However, note that $H$ does not directly satisfy $\mathcal{V}_5$, which only has two pairwise intersections.
To satisfy $\mathcal{V}_5$, we must fracture any of the three degree-2 vertices in $H$. For example, splitting vertex $a$ corresponds to the fracture
\begin{equation*}
    \vecsigma=\big(\{\{e_1\},\{e_2\}\},{\top},\dots,\top\big).
\end{equation*}
Then, the contribution of this fracture to the coefficient is 
\begin{align*}
    (-1)^1\cdot1!\cdot\prod_{\substack{v \in V(H) \\ v\neq a}}(-1)^0\cdot0! =-1.
\end{align*}
The fractured hypergraph resulting from splitting vertex $b$ or $c$ will also satisfy $\mathcal{V}_5$. Both of these fractures give us the same coefficient due to having the same number and size of blocks. Finally, note that splitting any two (or all three) of $a,b,c$ yields a hypergraph not satisfying $\mathcal{V}_5$. Hence
\begin{align*}
    \coeff_{\cycle(1,1,1),\mathcal{V}_5}(\vec{\top})&= 3\cdot (-1)
    =-3\neq 0.
\end{align*}

For the case of $\mathcal{V}_6$, using $H=\Delta(1,1,0)$, first note that the edge $e_3$ in $\cycle(1,1,0)$ contains no vertex of degree $1$ (in fact, $e_3=\{c,b\}$), but both $e_1$ and $e_2$ contain one vertex of degree $1$. Similarly as in the previous cases, the only choice for picking a fracture is whether or not we split $a$, $b$, and $c$. Splitting either two or all three of them yields an isolated edge and the Venn diagram $\mathcal{V}_6$ is not satisfied. Splitting none of them yields the hypergraph $\cycle(1,1,0)$ itself, which also does not satisfy $\mathcal{V}_6$. Hence we can only split precisely one of $a$, $b$, and $c$ if we hope to satisfy $\mathcal{V}_6$.

We observe that splitting $b$ or (symmetrically) $c$, we obtain a fractured graph that satisfies $\mathcal{V}_5$ and not $\mathcal{V}_6$. In fact, only the splitting of $a$ yields a fractured graph satisfying $\mathcal{V}_6$, the reason for which is that $e_3$ keeps being the unique edge with no degree $1$ vertex in the fractures graph.
Hence we are only splitting one of the degree-2 vertices, so the corresponding fracture satisfies 
\[\vecsigma =\big(\{\{e_1\},\{e_2\}\},\top,\dots,\top\big)\,, \]
and therefore
\begin{equation*}
    \coeff_{\cycle(1,1,0),\mathcal{V}_6}(\vec{\top})= (-1)^1\cdot1!\cdot\prod_{\substack{v \in V(H) \\ v\neq a}}(-1)^0\cdot0! =-1\neq 0.
\end{equation*}
\end{proof}

\subsection{Coefficients of $\mathcal{V}_7$, $\mathcal{V}_8$, $\mathcal{V}_9$, and $\mathcal{V}_{10}$}
\begin{tikzpicture}[scale=0.5, transform shape]
    \def\r{1.5}
    \coordinate (A) at (0,0);
    \coordinate (B) at (2,0);
    \coordinate (C) at (1,1.732);

    \begin{scope}
        \clip (A) circle (\r);
        \fill[pattern=\nel, pattern color=black] (A) circle (\r);
    \end{scope}
   
    \begin{scope}
        \clip (B) circle (\r);
        \fill[pattern=\nel, pattern color=black] (B) circle (\r);
    \end{scope}

    \begin{scope}
        \clip (C) circle (\r);
        \fill[pattern=\nel, pattern color=black] (C) circle (\r);
    \end{scope}

    \begin{scope}
        \clip (A) circle (\r);
        \clip (B) circle (\r);
        \fill[pattern=\nel, pattern color=black] (-2,-2) rectangle (4,4);
    \end{scope}

    \begin{scope}
        \clip (A) circle (\r);
        \clip (C) circle (\r);
        \fill[pattern=\nel, pattern color=black] (-1,-1) rectangle (4,4);
    \end{scope}

    \begin{scope}
        \clip (B) circle (\r);
        \clip (C) circle (\r);
        \fill[pattern=\nel, pattern color=black] (0,0) rectangle (4,4);
    \end{scope}

    \begin{scope}
        \clip (A) circle (\r);
        \clip (B) circle (\r);
        \clip (C) circle (\r);
        \fill[pattern=\nel, pattern color=black] (0,0) rectangle (4,4);
    \end{scope}

    \draw[line width=1.2pt] (A) circle (\r);
    \draw[line width=1.2pt] (B) circle (\r);
    \draw[line width=1.2pt] (C) circle (\r);

    \node (V_7) at (1.1,-2.5) {\huge $\mathcal{V}_7$};
\end{tikzpicture}
\begin{tikzpicture}[scale=0.5, transform shape]
    \def\r{1.5}
    \coordinate (A) at (0,0);
    \coordinate (B) at (2,0);
    \coordinate (C) at (1,1.732);

    \begin{scope}
        \clip (A) circle (\r);
        \fill[pattern=\nel, pattern color=black] (A) circle (\r);
    \end{scope}
   
    \begin{scope}
        \clip (B) circle (\r);
        \fill[pattern=none, pattern color=black] (B) circle (\r);
    \end{scope}

    \begin{scope}
        \clip (C) circle (\r);
        \fill[pattern=\nel, pattern color=black] (C) circle (\r);
    \end{scope}

    \begin{scope}
        \clip (A) circle (\r);
        \clip (B) circle (\r);
        \fill[pattern=\nel, pattern color=black] (-2,-2) rectangle (4,4);
    \end{scope}

    \begin{scope}
        \clip (A) circle (\r);
        \clip (C) circle (\r);
        \fill[pattern=\nel, pattern color=black] (-1,-1) rectangle (4,4);
    \end{scope}

    \begin{scope}
        \clip (B) circle (\r);
        \clip (C) circle (\r);
        \fill[pattern=\nel, pattern color=black] (0,0) rectangle (4,4);
    \end{scope}

    \begin{scope}
        \clip (A) circle (\r);
        \clip (B) circle (\r);
        \clip (C) circle (\r);
        \fill[pattern=\nel, pattern color=black] (0,0) rectangle (4,4);
    \end{scope}

    \draw[line width=1.2pt] (A) circle (\r);
    \draw[line width=1.2pt] (B) circle (\r);
    \draw[line width=1.2pt] (C) circle (\r);

    \node (V_8) at (1.1,-2.5) {\huge $\mathcal{V}_8$};
\end{tikzpicture}
\begin{tikzpicture}[scale=0.5, transform shape]
    \def\r{1.5}
    \coordinate (A) at (0,0);
    \coordinate (B) at (2,0);
    \coordinate (C) at (1,1.732);

    \begin{scope}
        \clip (A) circle (\r);
        \fill[pattern=none, pattern color=black] (A) circle (\r);
    \end{scope}
   
    \begin{scope}
        \clip (B) circle (\r);
        \fill[pattern=none, pattern color=black] (B) circle (\r);
    \end{scope}

    \begin{scope}
        \clip (C) circle (\r);
        \fill[pattern=\nel, pattern color=black] (C) circle (\r);
    \end{scope}

    \begin{scope}
        \clip (A) circle (\r);
        \clip (B) circle (\r);
        \fill[pattern=\nel, pattern color=black] (-2,-2) rectangle (4,4);
    \end{scope}

    \begin{scope}
        \clip (A) circle (\r);
        \clip (C) circle (\r);
        \fill[pattern=\nel, pattern color=black] (-1,-1) rectangle (4,4);
    \end{scope}

    \begin{scope}
        \clip (B) circle (\r);
        \clip (C) circle (\r);
        \fill[pattern=\nel, pattern color=black] (0,0) rectangle (4,4);
    \end{scope}

    \begin{scope}
        \clip (A) circle (\r);
        \clip (B) circle (\r);
        \clip (C) circle (\r);
        \fill[pattern=\nel, pattern color=black] (0,0) rectangle (4,4);
    \end{scope}

    \draw[line width=1.2pt] (A) circle (\r);
    \draw[line width=1.2pt] (B) circle (\r);
    \draw[line width=1.2pt] (C) circle (\r);

    \node (V_9) at (1.1,-2.5) {\huge $\mathcal{V}_9$};
\end{tikzpicture}
\begin{tikzpicture}[scale=0.5, transform shape]
    \def\r{1.5}
    \coordinate (A) at (0,0);
    \coordinate (B) at (2,0);
    \coordinate (C) at (1,1.732);

    \begin{scope}
        \clip (A) circle (\r);
        \fill[pattern=none, pattern color=black] (A) circle (\r);
    \end{scope}
   
    \begin{scope}
        \clip (B) circle (\r);
        \fill[pattern=none, pattern color=black] (B) circle (\r);
    \end{scope}

    \begin{scope}
        \clip (C) circle (\r);
        \fill[pattern=none, pattern color=black] (C) circle (\r);
    \end{scope}

    \begin{scope}
        \clip (A) circle (\r);
        \clip (B) circle (\r);
        \fill[pattern=\nel, pattern color=black] (-2,-2) rectangle (4,4);
    \end{scope}

    \begin{scope}
        \clip (A) circle (\r);
        \clip (C) circle (\r);
        \fill[pattern=\nel, pattern color=black] (-1,-1) rectangle (4,4);
    \end{scope}

    \begin{scope}
        \clip (B) circle (\r);
        \clip (C) circle (\r);
        \fill[pattern=\nel, pattern color=black] (0,0) rectangle (4,4);
    \end{scope}

    \begin{scope}
        \clip (A) circle (\r);
        \clip (B) circle (\r);
        \clip (C) circle (\r);
        \fill[pattern=\nel, pattern color=black] (0,0) rectangle (4,4);
    \end{scope}

    \draw[line width=1.2pt] (A) circle (\r);
    \draw[line width=1.2pt] (B) circle (\r);
    \draw[line width=1.2pt] (C) circle (\r);

    \node (V_10) at (1.1,-2.5) {\huge $\mathcal{V}_{10}$};
\end{tikzpicture} \\

The Venn diagrams $\mathcal{V}_7,\dots,\mathcal{V}_{10}$ correspond to the hypergraphs $\Gamma(j_1,j_2,j_3)$. Fractures of $\Gamma(j_1,j_2,j_3)$ are obtained by splitting any subset of the vertices having degree more than one (i.e., vertices $\{a,b,c,d\}$ in Figure~\ref{fig:Gamma}). 

\begin{lemma}[Coefficients for $\mathcal{V}_7,\dots,\mathcal{V}_{10}$]
    We have
    \begin{itemize}
        \item  $\coeff_{\hyperclique(1,1,1),\mathcal{V}_7}(\vec{\top})\neq 0$.
        \item  $\coeff_{\hyperclique(1,1,0),\mathcal{V}_8}(\vec{\top})\neq 0$.
        \item  $\coeff_{\hyperclique(1,0,0),\mathcal{V}_9}(\vec{\top})\neq 0$.
        \item  $\coeff_{\hyperclique(0,0,0),\mathcal{V}_{10}}(\vec{\top})\neq 0$.
    \end{itemize}
    \label{lem:coeff_V7_to_V10}
\end{lemma}

\begin{proof}
The Venn diagrams $\mathcal{V}_7,\dots,\mathcal{V}_{10}$ are those that cannot be satisfied when taking any non-trivial fracture of the respective hypergraphs $H=\Gamma(j_1,j_2,j_3)$ that split vertices $a,b,c,d$ (as labelled in Figure~\ref{fig:Gamma}): taking any non-trivial fracture of vertex $a$ leaves the intersection of all three edges empty. Splitting any subset of the degree-2 vertices will leave at least one pairwise intersection empty. In both cases, we no longer satisfy any of the Venn diagrams. Therefore, for each of the four cases, the only fracture we can hope to yield a fractured graph satisfying the respective Venn diagram is the trivial fracture $\vec{\top}$. 

Observe that $\mathcal{V}_7(\Gamma(1,1,1)\natural \vec{\top})= \mathcal{V}_7(\Gamma(1,1,1))=1$; hence $\mathcal{L}(\mathcal{V}_7,\Gamma(1,1,1))=\{\vec{\top}\}$, and thus
\[\coeff_{\hyperclique(1,1,1),\mathcal{V}_7}(\vec{\top}) = \prod_{v \in V(\Gamma(1,1,1))}(-1)^0\cdot0! =1\neq 0.\]
The identical argument applies to $\mathcal{V}_8$, $\mathcal{V}_9$, and $\mathcal{V}_{10}$ and, respectively, $\Gamma(1,1,0)$, $\Gamma(1,0,0)$, and $\Gamma(0,0,0)$.
\end{proof}

\subsection{Coefficients of $\mathcal{V}_{11}$, $\mathcal{V}_{12}$, and $\mathcal{V}_{13}$}
\begin{tikzpicture}[scale=0.5, transform shape]
    \def\r{1.5}
    \coordinate (A) at (0,0);
    \coordinate (B) at (2,0);
    \coordinate (C) at (1,1.732);

    \begin{scope}
        \clip (A) circle (\r);
        \fill[pattern=\nel, pattern color=black] (A) circle (\r);
    \end{scope}
   
    \begin{scope}
        \clip (B) circle (\r);
        \fill[pattern=none, pattern color=black] (B) circle (\r);
    \end{scope}

    \begin{scope}
        \clip (C) circle (\r);
        \fill[pattern=\nel, pattern color=black] (C) circle (\r);
    \end{scope}

    \begin{scope}
        \clip (A) circle (\r);
        \clip (B) circle (\r);
        \fill[pattern=\nel, pattern color=black] (-2,-2) rectangle (4,4);
    \end{scope}

    \begin{scope}
        \clip (A) circle (\r);
        \clip (C) circle (\r);
        \fill[fill=white, pattern color=black] (-1,-1) rectangle (4,4);
    \end{scope}

    \begin{scope}
        \clip (B) circle (\r);
        \clip (C) circle (\r);
        \fill[pattern=\nel, pattern color=black] (0,0) rectangle (4,4);
    \end{scope}

    \begin{scope}
        \clip (A) circle (\r);
        \clip (B) circle (\r);
        \clip (C) circle (\r);
        \fill[pattern=\nel, pattern color=black] (0,0) rectangle (4,4);
    \end{scope}

    \draw[line width=1.2pt] (A) circle (\r);
    \draw[line width=1.2pt] (B) circle (\r);
    \draw[line width=1.2pt] (C) circle (\r);

    \node (V_11) at (1.1,-2.5) {\huge $\mathcal{V}_{11}$};
\end{tikzpicture} 
\begin{tikzpicture}[scale=0.5, transform shape]
    \def\r{1.5}
    \coordinate (A) at (0,0);
    \coordinate (B) at (2,0);
    \coordinate (C) at (1,1.732);

    \begin{scope}
        \clip (A) circle (\r);
        \fill[pattern=\nel, pattern color=black] (A) circle (\r);
    \end{scope}
   
    \begin{scope}
        \clip (B) circle (\r);
        \fill[pattern=\nel, pattern color=black] (B) circle (\r);
    \end{scope}

    \begin{scope}
        \clip (C) circle (\r);
        \fill[pattern=\nel, pattern color=black] (C) circle (\r);
    \end{scope}

    \begin{scope}
        \clip (A) circle (\r);
        \clip (B) circle (\r);
        \fill[pattern=\nel, pattern color=black] (-2,-2) rectangle (4,4);
    \end{scope}

    \begin{scope}
        \clip (A) circle (\r);
        \clip (C) circle (\r);
        \fill[fill=white, pattern color=black] (-1,-1) rectangle (4,4);
    \end{scope}

    \begin{scope}
        \clip (B) circle (\r);
        \clip (C) circle (\r);
        \fill[fill=white, pattern color=black] (0,0) rectangle (4,4);
    \end{scope}

    \begin{scope}
        \clip (A) circle (\r);
        \clip (B) circle (\r);
        \clip (C) circle (\r);
        \fill[pattern=\nel, pattern color=black] (0,0) rectangle (4,4);
    \end{scope}

    \draw[line width=1.2pt] (A) circle (\r);
    \draw[line width=1.2pt] (B) circle (\r);
    \draw[line width=1.2pt] (C) circle (\r);

    \node (V_12) at (1.1,-2.5) {\huge $\mathcal{V}_{12}$};
\end{tikzpicture}
\begin{tikzpicture}[scale=0.5, transform shape]
    \def\r{1.5}
    \coordinate (A) at (0,0);
    \coordinate (B) at (2,0);
    \coordinate (C) at (1,1.732);

    \begin{scope}
        \clip (A) circle (\r);
        \fill[pattern=\nel, pattern color=black] (A) circle (\r);
    \end{scope}
   
    \begin{scope}
        \clip (B) circle (\r);
        \fill[pattern=\nel, pattern color=black] (B) circle (\r);
    \end{scope}

    \begin{scope}
        \clip (C) circle (\r);
        \fill[pattern=\nel, pattern color=black] (C) circle (\r);
    \end{scope}

    \begin{scope}
        \clip (A) circle (\r);
        \clip (B) circle (\r);
        \fill[fill=white, pattern color=black] (-2,-2) rectangle (4,4);
    \end{scope}

    \begin{scope}
        \clip (A) circle (\r);
        \clip (C) circle (\r);
        \fill[fill=white, pattern color=black] (-1,-1) rectangle (4,4);
    \end{scope}

    \begin{scope}
        \clip (B) circle (\r);
        \clip (C) circle (\r);
        \fill[fill=white, pattern color=black] (0,0) rectangle (4,4);
    \end{scope}

    \begin{scope}
        \clip (A) circle (\r);
        \clip (B) circle (\r);
        \clip (C) circle (\r);
        \fill[pattern=\nel, pattern color=black] (0,0) rectangle (4,4);
    \end{scope}

    \draw[line width=1.2pt] (A) circle (\r);
    \draw[line width=1.2pt] (B) circle (\r);
    \draw[line width=1.2pt] (C) circle (\r);

    \node (V_13) at (1.1,-2.5) {\huge $\mathcal{V}_{13}$};
\end{tikzpicture} \\
All three of these Venn diagrams will require the hypergraph $H=\Gamma(0,0,0)$. $H$ simply contains vertices $a,b,c,d$ from Figure~\ref{fig:Gamma} and no degree-1 vertices. 
\begin{lemma}
    We have
    \begin{itemize}
        \item  $\coeff_{\hyperclique(0,0,0),\mathcal{V}_{11}}(\vec{\top})\neq 0$.
        \item  $\coeff_{\hyperclique(0,0,0),\mathcal{V}_{12}}(\vec{\top})\neq 0$.
        \item  $\coeff_{\hyperclique(0,0,0),\mathcal{V}_{13}}(\vec{\top})\neq 0$.
    \end{itemize}
    \label{lem:coeff_V11_to_V13}
\end{lemma}
\begin{proof}
First note that splitting the centre vertex $a$ in $\Gamma(0,0,0)$ yields a hypergraph without a vertex of degree $3$, hence any fracture $\vecsigma$ of $\Gamma(0,0,0)$ yielding a fractured graph satisfying any one of $\mathcal{V}_{11}$, $\mathcal{V}_{12}$ or $\mathcal{V}_{13}$ must satisfy $\vecsigma_a = \top$

The Venn diagram $\mathcal{V}_{11}$ is satisfied by taking the fracture of any one of the degree-2 vertices in $H=\hyperclique(0,0,0)$. Ordering the vertices of $\hyperclique(0,0,0)$ as $(a,b,c,d)$,
the following fracture, for example, corresponding to the splitting of vertex $d$ in Figure \ref{fig:Gamma} and yields a fractured graph satisfying $\mathcal{V}_{11}$: $$\big(\top,\top,\top,\{\{e_1\},\{e_3\}\}\big)\,.$$
The contribution to the coefficient $\coeff_{\hyperclique(0,0,0),\mathcal{V}_{11}}(\vec{\top})$ is thus
$$\big((-1)^0\cdot0!\big)^3\cdot\big((-1)^1\cdot1!\big)=-1.$$
Since we obtain isomorphic fractured graphs by splitting vertices $b$ or $c$, we have
$$ \coeff_{\hyperclique(0,0,0),\mathcal{V}_{11}}(\vec{\top})= -3. $$

Next, $\mathcal{V}_{12}$ is satisfied by taking the fracture of any pair of degree-2 vertices in $\hyperclique(0,0,0)$. For example, splitting $c$ and $d$ corresponds to
$$\big(\top,\top,\{\{e_2\},\{e_3\}\},\{\{e_1\},\{e_3\}\}\big)\,,$$
and the contribution to the coefficient $\coeff_{\hyperclique(0,0,0),\mathcal{V}_{12}}(\vec{\top})$ is
$$\big((-1)^0\cdot0!\big)^2\cdot\big((-1)^1\cdot1!\big)^2=1\,.$$
Since there are $3$ choices of picking a pair of the degree-2 vertices $b,c,d$, we obtain
$$\coeff_{\hyperclique(0,0,0),\mathcal{V}_{12}}(\vec{\top})= 3\,.$$

Finally, $\mathcal{V}_{13}$ is satisfied by splitting all three degree-2 vertices $b,c,d$, yielding the fracture
$$\big(\top,\{{\{e_1\},\{e_2\}}\},\{\{e_2\},\{e_3\}\},\{\{e_1\},\{e_3\}\}\big)\,.$$
This is the only fracture for which the fractured graph satisfies $\mathcal{V}_{13}$, and thus
\[\coeff_{\hyperclique(0,0,0),\mathcal{V}_{13}}(\vec{\top})=\big((-1)^0\cdot0!\big)\cdot\big((-1)^1\cdot1!\big)^3= -1\neq 0\,,\]
concluding the proof.
\end{proof}

\subsection{Coefficient of $\mathcal{V}_{14}$}
\label{sec4.5}
\begin{tikzpicture}[scale=0.5, transform shape]
    \def\r{1.5}
    \coordinate (A) at (0,0);
    \coordinate (B) at (2,0);
    \coordinate (C) at (1,1.732);

    \begin{scope}
        \clip (A) circle (\r);
        \fill[pattern=\nel, pattern color=black] (A) circle (\r);
    \end{scope}
   
    \begin{scope}
        \clip (B) circle (\r);
        \fill[pattern=\nel, pattern color=black] (B) circle (\r);
    \end{scope}

    \begin{scope}
        \clip (C) circle (\r);
        \fill[pattern=\nel, pattern color=black] (C) circle (\r);
    \end{scope}

    \begin{scope}
        \clip (A) circle (\r);
        \clip (B) circle (\r);
        \fill[pattern=\nel, pattern color=black] (-2,-2) rectangle (4,4);
    \end{scope}

    \begin{scope}
        \clip (A) circle (\r);
        \clip (C) circle (\r);
        \fill[fill=white, pattern color=black] (-1,-1) rectangle (4,4);
    \end{scope}

    \begin{scope}
        \clip (B) circle (\r);
        \clip (C) circle (\r);
        \fill[pattern=\nel, pattern color=black] (0,0) rectangle (4,4);
    \end{scope}

    \begin{scope}
        \clip (A) circle (\r);
        \clip (B) circle (\r);
        \clip (C) circle (\r);
        \fill[pattern=\nel, pattern color=black] (0,0) rectangle (4,4);
    \end{scope}

    \draw[line width=1.2pt] (A) circle (\r);
    \draw[line width=1.2pt] (B) circle (\r);
    \draw[line width=1.2pt] (C) circle (\r);

    \node (V_14) at (1.1,-2.5) {\huge $\mathcal{V}_{14}$};
\end{tikzpicture}\\
\begin{lemma}
    We have $\coeff_{\hyperclique(1,0,0),\mathcal{V}_{14}}(\vec{\top})\neq 0$.
    \label{lem:coeff_V14}
\end{lemma}
\begin{proof}
Recall that $\hyperclique(1,0,0)$ contains the four vertices $a,b,c,d$ and one degree-$1$ vertex --- we call it $x$ --- in the edge $e_1=\{a,b,d,x\}$. We show that there is precisely $1$ fracture yielding to a fractured graph satisfying $\mathcal{V}_{14}$. To this end, assume that $\vecsigma$ is such a fracture. Then \[\vecsigma=(\vecsigma_a,\vecsigma_b,\vecsigma_c,\vecsigma_d,\vecsigma_x)\,.\]
First, since $x$ has degree $1$, $\vecsigma_x=\top$. Next, for satisfying $\mathcal{V}_{14}$, a vertex of degree $3$ is necessary; this can only be achieved if $a$ is not split, hence $\vecsigma_a=\top$. Moreover, $\mathcal{V}_{14}$ also requires the presence of at least $2$ vertices of degree $2$, hence we can split at most one of the vertices $b,c,d$. If none is split, $\vecsigma=\vec\top$, but $\hyperclique(1,0,0)\natural \vec{\top} = \hyperclique(1,0,0)$ does not satisfy $\mathcal{V}_{14}$. Hence, the only possibility for obtaining a fractured graph satisfying $\mathcal{V}_{14}$ is to split precisely one of $b,c,d$, and to not split all other vertices. Finally, it can easily be verified that splitting $c$ yields a fractured graph satisfying $\mathcal{V}_{14}$, while splitting $b$ or $d$ does not --- see Figure~\ref{fig:V14_fractures} for a depiction of those three cases. 

\begin{figure}[t]
\centering
\includegraphics[width=1.0\linewidth]{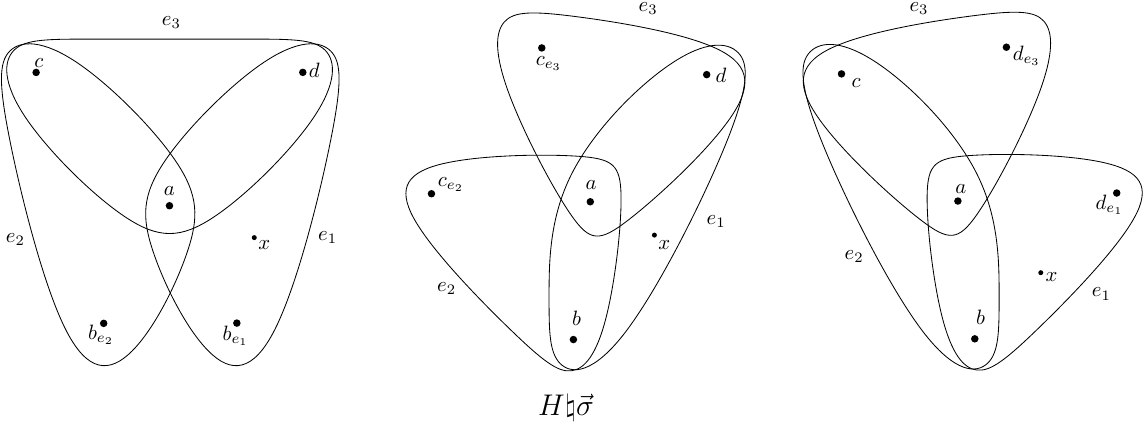}\\
~\\
\begin{tikzpicture}[scale=0.9, transform shape]
    \def\r{1.5}
    \coordinate (A) at (0,0);
    \coordinate (B) at (2,0);
    \coordinate (C) at (1,1.732);

    \begin{scope}
        \clip (A) circle (\r);
        \fill[pattern=\nel, pattern color=black] (A) circle (\r);
    \end{scope}
   
    \begin{scope}
        \clip (B) circle (\r);
        \fill[pattern=\nel, pattern color=black] (B) circle (\r);
    \end{scope}

    \begin{scope}
        \clip (C) circle (\r);
        \fill[pattern=none, pattern color=black] (C) circle (\r);
    \end{scope}

    \begin{scope}
        \clip (A) circle (\r);
        \clip (B) circle (\r);
        \fill[fill=white, pattern color=black] (-2,-2) rectangle (4,4);
    \end{scope}

    \begin{scope}
        \clip (A) circle (\r);
        \clip (C) circle (\r);
        \fill[pattern=\nel, pattern color=black] (-1,-1) rectangle (4,4);
    \end{scope}

    \begin{scope}
        \clip (B) circle (\r);
        \clip (C) circle (\r);
        \fill[pattern=\nel, pattern color=black] (0,0) rectangle (4,4);
    \end{scope}

    \begin{scope}
        \clip (A) circle (\r);
        \clip (B) circle (\r);
        \clip (C) circle (\r);
        \fill[pattern=\nel, pattern color=black] (0,0) rectangle (4,4);
    \end{scope}

    \draw[line width=1.2pt] (A) circle (\r);
    \draw[line width=1.2pt] (B) circle (\r);
    \draw[line width=1.2pt] (C) circle (\r);

    \node (V1) at (1.1,-2.5) {~~};

\end{tikzpicture}~~~~~~
\begin{tikzpicture}[scale=0.9, transform shape]
    \def\r{1.5}
    \coordinate (A) at (0,0);
    \coordinate (B) at (2,0);
    \coordinate (C) at (1,1.732);

    \begin{scope}
        \clip (A) circle (\r);
        \fill[pattern=\nel, pattern color=black] (A) circle (\r);
    \end{scope}
   
    \begin{scope}
        \clip (B) circle (\r);
        \fill[pattern=\nel, pattern color=black] (B) circle (\r);
    \end{scope}

    \begin{scope}
        \clip (C) circle (\r);
        \fill[pattern=\nel, pattern color=black] (C) circle (\r);
    \end{scope}

    \begin{scope}
        \clip (A) circle (\r);
        \clip (B) circle (\r);
        \fill[pattern=\nel, pattern color=black] (-2,-2) rectangle (4,4);
    \end{scope}

    \begin{scope}
        \clip (A) circle (\r);
        \clip (C) circle (\r);
        \fill[fill=white, pattern color=black] (-1,-1) rectangle (4,4);
    \end{scope}

    \begin{scope}
        \clip (B) circle (\r);
        \clip (C) circle (\r);
        \fill[pattern=\nel, pattern color=black] (0,0) rectangle (4,4);
    \end{scope}

    \begin{scope}
        \clip (A) circle (\r);
        \clip (B) circle (\r);
        \clip (C) circle (\r);
        \fill[pattern=\nel, pattern color=black] (0,0) rectangle (4,4);
    \end{scope}

    \draw[line width=1.2pt] (A) circle (\r);
    \draw[line width=1.2pt] (B) circle (\r);
    \draw[line width=1.2pt] (C) circle (\r);

    \node (V_14) at (1.1,-2.5) {\Large $\mathcal{V}_{14}$};
\end{tikzpicture}~~~~~~
\begin{tikzpicture}[scale=0.9, transform shape]
    \def\r{1.5}
    \coordinate (A) at (0,0);
    \coordinate (B) at (2,0);
    \coordinate (C) at (1,1.732);

    \begin{scope}
        \clip (A) circle (\r);
        \fill[pattern=none, pattern color=black] (A) circle (\r);
    \end{scope}
   
    \begin{scope}
        \clip (B) circle (\r);
        \fill[pattern=\nel, pattern color=black] (B) circle (\r);
    \end{scope}

    \begin{scope}
        \clip (C) circle (\r);
        \fill[pattern=\nel, pattern color=black] (C) circle (\r);
    \end{scope}

    \begin{scope}
        \clip (A) circle (\r);
        \clip (B) circle (\r);
        \fill[pattern=\nel, pattern color=black] (-2,-2) rectangle (4,4);
    \end{scope}

    \begin{scope}
        \clip (A) circle (\r);
        \clip (C) circle (\r);
        \fill[pattern=\nel, pattern color=black] (-1,-1) rectangle (4,4);
    \end{scope}

    \begin{scope}
        \clip (B) circle (\r);
        \clip (C) circle (\r);
        \fill[fill=white, pattern color=black] (0,0) rectangle (4,4);
    \end{scope}

    \begin{scope}
        \clip (A) circle (\r);
        \clip (B) circle (\r);
        \clip (C) circle (\r);
        \fill[pattern=\nel, pattern color=black] (0,0) rectangle (4,4);
    \end{scope}

    \draw[line width=1.2pt] (A) circle (\r);
    \draw[line width=1.2pt] (B) circle (\r);
    \draw[line width=1.2pt] (C) circle (\r);

    \node (V1) at (1.1,-2.5) {~~};

\end{tikzpicture}
\caption{The resulting fractured hypergraphs of $H=\hyperclique(1,0,0)$ from splitting vertices $b$, $c$, and $d$, respectively, and their representative Venn diagrams.}
\label{fig:V14_fractures}
\end{figure}
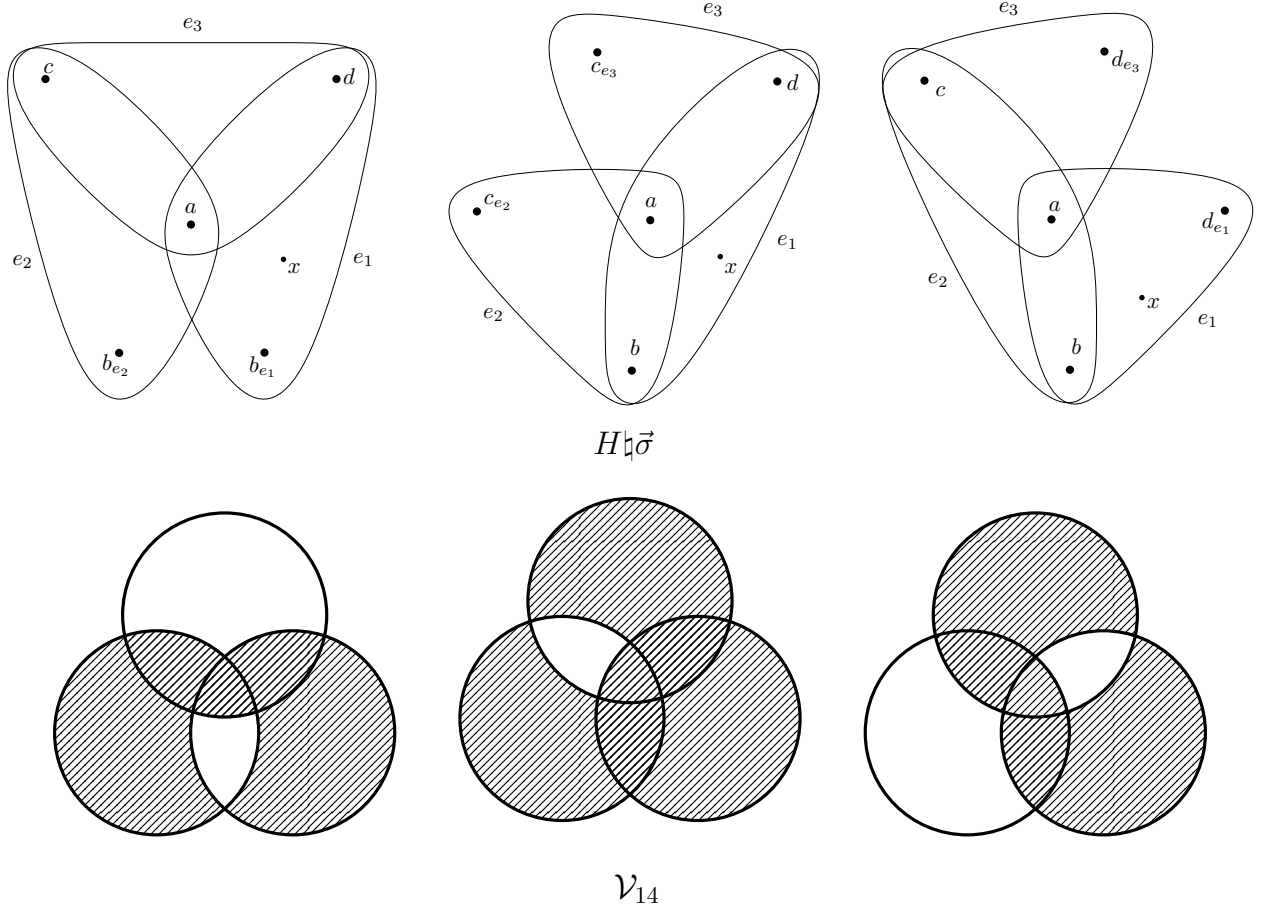

Thus, $\mathcal{L}(\hyperclique(1,0,0),\mathcal{V}_{14})$ contains precisely one fracture:
\[ \vecsigma=(\top,\top, \{\{e_2\},\{e_3\}\},\top,\top) \,.\]
Thus
$$\coeff_{\hyperclique(1,0,0),\mathcal{V}_{14}}(\vec{\top})=\big((-1)^0\cdot0!\big)^3\cdot\big((-1)^1\cdot1!\big)=-1\neq 0\,,$$
concluding the proof.
\end{proof}

\begin{corollary}\label{cor:coeff_summary}
    Let $\mathcal{V}$ be a non-degenerate Venn diagram. Then there is a hypergraph $H$ with $\mathsf{ghtw}(H)> 1$ such that $\coeff_{H,\mathcal{V}}(\vec\top)\neq 0$.
\end{corollary}
\begin{proof}
Lemmas~\ref{lem:coeff_V1_to_V4}-\ref{lem:coeff_V14} cover all non-degenerate Venn diagrams and show that, for each diagram $\mathcal{V}$ of them, there are $j_1,j_2,j_3\geq 0$ such that either $\coeff_{\cycle(j_1,j_2,j_3),\mathcal{V}}(\vec\top)\neq 0$ or $\coeff_{\hyperclique(j_1,j_2,j_3),\mathcal{V}}(\vec\top)\neq 0$. Finally, by Observation~\ref{obs:delta_gamma_cyclic}, we have that $\mathsf{ghtw}(\Delta(j_1,j_2,j_3))>1$ and $\mathsf{ghtw}(\Gamma(j_1,j_2,j_3))>1$.
\end{proof}

\section{Proof of the Main Result}
\label{sec5:proof}
\subsection*{Upper bounds}
For establishing our upper bounds, we rely on the following simple observation --- recall that $H\downarrow$ denotes the hypergraph obtained from $H$ by removing non-maximimal edges, that is, by deleting all edges $e$ such that $e \subsetneq e'$ for some $e' \in E(H)$; recall further that an isolated vertex of a hypergraph $H$ is a vertex not included in any edge.

\begin{lemma}\label{lem:easy_quotient_props}
    Let $H$ be a hypergraph without isolated vertices.
    \begin{enumerate}
        \item If $|E(H)|\leq 3$ then $\mathsf{ghtw}(H/\rho)\leq 2$ for each partition $\rho$ of $V(H)$.
        \item If $|E(H\downarrow)|\leq 2$ then $\mathsf{ghtw}(H/\rho)\leq 1$, that is, $H/\rho$ is acyclic, for each partition $\rho$ of $V(H)$
    \end{enumerate}
\end{lemma}
\begin{proof}
    \begin{enumerate}
        \item Note that $|E(H/\rho)|\leq |E(H)|$ for each $\rho$, and that taking quotients cannot create isolated vertices. Thus it suffices to show that any hypergraph with at most $3$ edges has generalised hypertreewidth at most $2$. To this end, let $e_1$, $e_2$, and $e_3$ denote the three edges. The hypertree decomposition then just consists of a two node tree $T$ with $V(T)=\{t,t'\}$ and $E(T)=\{(t,t')\}$ such that $B_t=e_1 \cup e_2$ and $B_{t'}= e_3$. Clearly, $B_t \cup B_{t'}$ is the set of all vertices of the hypergraph, and each bag is covered by at most $2$ edges. Furthermore, for each vertex $v$, the subtree $T_v=T[\{t \mid v \in B_t\}]$ is connected as $T_v$ is either an edge or a single node.
        \item Observe that $H/\rho$ cannot have more non-maximal edges than $H$ for any partition $\rho$. Hence, it suffices to show that any hypergraph with at most two maximal edges $e_1$ and $e_2$ has generalised hypertreewidth at most $1$. This is trivial as we can just pick bags $B_t=e_1$ and $B_{t'}=e_2$ and connect the two nodes $t$ and $t'$ to create a hypertree decomposition in which every bag ($B_t$ and $B_{t'}$) can be covered with at most one edge.
    \end{enumerate}
\end{proof}

Our algorithms will operate on the uncoloured problem $\motifprob{\mathcal{V}}$. To this end, we will show a simple transformation of $\motifs{\mathcal{V}}{G}$ into a linear combination of (non-coloured) homomorphism counts. While the transformation will be easier than compared to the coloured setting via fractures in the previous section, the coefficient formula will be much more intricate. However, for the properties of the coefficients required for our upper bounds will be much easier to prove, which is why we do not need to take the technical detour to fractured hypergraphs.

\begin{lemma}\label{lem:uncoloured_motif_to_hom}
    Let $r$ be a positive integer, and let $\mathcal{V}$ be a Venn diagram. Moreover, let $S(\mathcal{V},r)$ be the set of all (isomorphism types of) $3$-edge hypergraphs $H$ of rank at most $r$ such that $\mathcal{V}(H)=1$. Then, for all hypergraphs $G$ of rank at most $r$ we have
    \[ \motifs{\mathcal{V}}{G} = \sum_{F} \alpha_{\mathcal{V},r}(F) \cdot \homs{F}{G} \,,\]
    where
    \[\alpha_{\mathcal{V},r}(F) = \sum_{H \in S(\mathcal{V},r)} \auts{H}^{-1} \cdot \sum_{\substack{\rho \in P(V(H))\\H/\rho \cong F}} \mu(\rho)\,.\]
    Here $\mu(\rho)$ denotes the M\"obius function of the partition lattice of $V(H)$.
\end{lemma}
\begin{proof}
    We have
    \begin{align*}
        \motifs{\mathcal{V}}{G} &\stackrel{(a)}{=} \sum_{H \in S(\mathcal{V},r)} \subs{H}{G}\\
        ~&\stackrel{(b)}{=} \sum_{H \in S(\mathcal{V},r)}\auts{H}^{-1}\cdot \embs{H}{G}\\
        ~&\stackrel{(c)}{=} \sum_{H \in S(\mathcal{V},r)}\auts{H}^{-1}\cdot \sum_{\rho \in P(V(H))} \mu(\rho)\cdot \homs{H/\rho}{G}\,,
    \end{align*}
    where (a) follows from the observation that there is a canonical bijection between the elements of $\motif{\mathcal{V}}{G}$ and the $3$-edge subgraphs of $G$ that satisfy $\mathcal{V}$, (b) follows from Equation~\eqref{sub=emb/aut}, and (c) follows from Equation~\eqref{embintermsofhom}.
    The formula for the coefficient $\alpha_{\mathcal{V},r}(F)$ is the easily obtained by collection for all quotients $H/\rho$ isomorphic to $F$. 
\end{proof}

The final ingredient for our upper bounds is Yannakakis' algorithm~\cite{Yannakakis} which yields a near-linear-time algorithm for evaluating, and counting answers to, acyclic conjunctive queries, inducing immediately a near-linear-time algorithm for counting homomorphisms from an acyclic hypergraph, that is, a hypergraph of generalised hypertreewidth $1$ --- see Section 3.2 of Mengel's survey~\cite{Mengel,Mengel_full} for a detailed exposition of the counting problem. Moreover,  
it is well known that counting homomorphisms from a hypergraph $H$ into a hypergraph $G$ can, in time $O(|G|^{\mathsf{ghtw}(H)})$, be reduced to counting answers to an acyclic conjunctive query $Q(H)$ in a database of size at most $O(|G|^{\mathsf{ghtw}(H)})$ (see e.g.\ Scarcello's survey~\cite{Scarcello05}). In combination, this yields a near-linear-time algorithm for counting homomorphisms from hypergraphs of generalised hypertreewidth~$1$ and a near-quadratic-time algorithm for counting homomorphisms from hypergraphs of generalised hypertreewidth~$2$. Formally, we obtain the following algorithm --- recall that, for any hypergraph $H$, the problem $\#\textsc{Hom}(\{H\})$ expects as input a hypergraph $G$ and outputs $\homs{H}{G}$:
\begin{theorem}[\cite{Yannakakis,Scarcello05,BraultBaron13,Mengel_full}]\label{thm:upper_bound_oracle}
    The problem $\#\textsc{Hom}(\{H\})$ can be solved in time $\Tilde{O}(|G|^{\mathsf{ghtw}(H)})$.
\end{theorem}

Now recall that a Venn diagram is called degenerate if one of its sets is fully contained in another set --- see Definition~\ref{def:degenerate}. Recall further that all degenerate and non-degenerate Venn diagrams are depicted in Figure~\ref{fig:allVenns}.

\begin{theorem}[Main Theorem, Upper Bound]
    Let $\mathcal{V}$ be a Venn diagram. The problem $\motifprob{\mathcal{V}}$ can be solved in FPT near-quadratic time, that is, there is a computable function $f$ such that the problem can be solved in time
    \[ f(\mathsf{rank}(G))\cdot \tilde{O}(|G|^2) \,.\]
    Moreover, if $\mathcal{V}$ is degenerate, then $\motifprob{\mathcal{V}}$ can be solved in FPT near-linear time, that is, there is a computable function $f$ such that the problem can be solved in time
    \[ f(\mathsf{rank}(G))\cdot \tilde{O}(|G|) \,.\]
\end{theorem}
\begin{proof}
    On input $G$, both algorithms compute the transformation given by Lemma~\ref{lem:uncoloured_motif_to_hom} with $r=\mathsf{rank}(G)$:
    \[ \motifs{\mathcal{V}}{G} = \sum_{F} \alpha_{\mathcal{V},r}(F) \cdot \homs{F}{G} \,,\]
    where
    \[\alpha_{\mathcal{V},r}(F) = \sum_{H \in S(\mathcal{V},r)} \auts{H}^{-1} \cdot \sum_{\substack{\rho \in P(V(H))\\H/\rho \cong F}} \mu(\rho)\,.\]
    As $\mathcal{V}$ is fixed, the size of $S(\mathcal{V},r)$ only depends on the parameter $r$. Next, note that $\alpha_{\mathcal{V},r}(F)\neq 0$ implies that $F$ must be a quotient of $H$ for some $H \in S(\mathcal{V},r)$. This has three important consequences:
    \begin{enumerate}
        \item[(a)] The size of the support of $\alpha_{\mathcal{V},r}$ only depends on $S(\mathcal{V},r)$ and thus only on $r$. Hence all non-zero $\alpha_{\mathcal{V},r}(F)$ can be computed (by brute force) in time only depending on $r$.
        \item[(b)] As every $H \in S(\mathcal{V},r)$ has at most $3$ edges, it follows from Lemma~\ref{lem:easy_quotient_props}.1 that $\alpha_{\mathcal{V},r}(F)\neq 0$ only if $\mathsf{ghtw}(F)\leq 2$.
        \item[(c)] If $\mathcal{V}$ is degenerate, which implies that $|E(H\downarrow)|\leq 2$ for all $H \in S(\mathcal{V},r)$, then Lemma~\ref{lem:easy_quotient_props}.2 implies that $\alpha_{\mathcal{V},r}(F)\neq 0$ only if $\mathsf{ghtw}(F)\leq 1$.
    \end{enumerate}
    We can thus conclude by observing that the desired running times are achieved by running the algorithm from Theorem~\ref{thm:upper_bound_oracle} to compute $\homs{F}{G}$ for all $F$ with $\alpha_{\mathcal{V},r}(F)\neq 0$ and evaluating the linear combination.

    For general Venn diagrams, (2) yields thus an FPT near-quadratic algorithm, and for degenerate Venn diagrams, (3) yields thus an FPT near-linear-time algorithm.
\end{proof}

\begin{corollary}
    When restricted to input hypergraphs of bounded rank, $\motifprob{\mathcal{V}}$ is solvable in near-quadratic time for arbitrary Venn diagrams, and it is solvable in near-linear-time for degenerate Venn diagrams.
\end{corollary}

\subsection*{Lower bounds}
For the conditional lower bounds, we will rely on the following chain of reductions:
\begin{lemma}\label{lem:lower_bound_reduction_chain}
    Let $\mathcal{V}$ be a non-degenerate Venn diagram. Then there exists a hypergraph $H$ of rank at most $4$ with $\mathsf{ghtw}(H)>1$ such that
    \[\#\textsc{Hom}(\{H\}) \fptlinred \#\textsc{cpHom}(\{H\}) \fptlinred \colmotifprob{\mathcal{V}} \fptlinred \motifprob{\mathcal{V}}  \,.\]
\end{lemma}

We first establish the last step of the reduction chain via an inclusion-exclusion argument.
\begin{lemma}\label{lem:incl_excl}
    Let $\mathcal{V}$ be a Venn Diagram. Then 
    \[\colmotifprob{\mathcal{V}} \fptlinred \motifprob{\mathcal{V}} \,.\]
\end{lemma}
\begin{proof}
    Let $(G,c)$ be an input to $\colmotifprob{\mathcal{V}}$. Recall that $(G,c)$ is $H$-coloured for some hypergraph~$H$ with $|E(H)|=3$ and the task is to compute
    \[ \colmotifs{\mathcal{V}}{H}{(G,c)} = \#\{A \subseteq E(G) \mid |A|=3 ~\wedge~c(A)=E(H)~\wedge \mathcal{V}(G[A])=1 \} \]
    For each $J \subseteq E(H)$ we let $G \setminus J$ be the graph obtained from $G$ by deleting each edge $e \in E(G)$ with $c(e)\in J$ (and by deleting isolated vertices afterwards). Clearly, $G\setminus J$ can be constructed in $\Tilde{O}(|G|)$

    Observe that, by inclusion-exclusion, we obtain
    \begin{align*}
      \colmotifs{\mathcal{V}}{H}{(G,c)} = \sum_{J \subseteq E(H)} (-1)^{|J|} \cdot  \motifs{\mathcal{V}}{G\setminus J}\,,
    \end{align*}
    which can be computed via $2^{|E(H)|}=8$ calls to our oracle for $\motifprob{\mathcal{V}}$. 
\end{proof}

We are now able to establish the full reduction chain:
\begin{proof}[Proof of Lemma \ref{lem:lower_bound_reduction_chain}]
    Let $\mathcal{V}$ be any non-degenerate Venn diagram. By Corollary~\ref{cor:coeff_summary}, there is a hypergraph $H$ with $\mathsf{ghtw}(H)>1$ such that $\coeff_{H,\mathcal{V}}(\vec\top) \neq 0.$

    For the first part of the reduction chain $\#\textsc{Hom}(\{H\}) \fptlinred \#\textsc{cpHom}(\{H\})$, we note that the construction is analogous to the case of graphs --- see~\cite[Lemma 4]{DorflerRSW22} and observe that the construction yields edgewise-injectively $H$-coloured graphs as required.

    The second part of the reduction chain $\#\textsc{cpHom}(\{H\}) \fptlinred \colmotifprob{\mathcal{V}}$ is given by the oracle algorithm $R$ established in Lemma~\ref{lem:dedekind_applied}: specifically, since $\coeff_{H,\mathcal{V}}(\vec\top) \neq 0$, we have that $R$ computes
    \[ \cphoms{H \natural \vec{\top}}{H}{(G,c)} = \cphoms{H}{H}{(G,c)} \,.\]

    The third and last part of the reduction chain is established in Lemma~\ref{lem:incl_excl}.
\end{proof}

Finally, we rely on the following folklore result for the base of our chain of reductions; the first explicit proof of this lower bound\footnote{Note that this lower bound is usually stated in the language of database conjunctive query evaluation without projections, but it is well-known to apply directly to hypergraph homomorphisms} is due to Brault-Baron~\cite{BraultBaron13}, but we also refer the reader to the more concise overview and statement in Mengel's survey (see Theorem 3.8 in~\cite{Mengel,Mengel_full}).

\begin{theorem}[\cite{BraultBaron13,Mengel,Mengel_full}]\label{thm:lower_bound_base_mengel}
    Let $H$ be a non-acyclic hypergraph and assume that both the Triangle and the Hyperclique Hypotheses hold. Then $\#\textsc{Hom}(\{H\})$ cannot be solved in time $\Tilde{O}(|G|)$.
\end{theorem}

\begin{theorem}[Main Theorem; Lower Bound]
    Let $\mathcal{V}$ be a non-degenerate Venn diagram. Then the problem $\motifprob{\mathcal{V}}$ is not solvable in FPT-near-linear time (and thus also not in near-linear time) unless either the Triangle Hypothesis or the Hyperclique Hypothesis fails.
\end{theorem}
\begin{proof}
    By Lemma~\ref{lem:lower_bound_reduction_chain}, there is a hypergraph $H$ of rank at most $4$ with $\mathsf{ghtw}(H)>1$ such that
    \[\#\textsc{Hom}(\{H\}) \fptlinred \motifprob{\mathcal{V}}\,. \]
    Since $\mathsf{ghtw}(H)>1$ we have that $H$ is not acyclic. Hence, by Theorem~\ref{thm:lower_bound_base_mengel}, we have that $\#\textsc{Hom}(\{H\})$ is not solvable in near-linear-time unless either the Triangle Hypothesis or the Hyperclique Hypothesis fails. As $H$ has rank at most $4$, we can assume that the input $G$ to $\#\textsc{Hom}(\{H\})$ also has rank at most $4$ since edges of cardinality larger than $4$ in $G$ are irrelevant for homomorphisms from hypergraphs of rank $4$ or below and can thus be deleted. Thus, for the purpose of this lower bound, the parameter $\mathsf{rank}(G)$ is constant and, by definition of $\fptlinred$, all oracle queries made to $ \motifprob{\mathcal{V}}$ must have their parameter bounded by $f(\mathsf{rank}(G))$ for some computable function $f$. However, $f(\mathsf{rank}(G)) \in O(1)$, and thus, for $H$ having rank at most $4$, an FPT-near-linear time algorithm for $\motifprob{\mathcal{V}}$ yields a near linear-time algorithm for $\#\textsc{Hom}(\{H\})$.
\end{proof}

\section{Generalised Hypergraph Motifs}
\label{sec6:general_case}
In the last part of this work, we investigate hypergraph motifs beyond $3$ edges. Following the proposal of Lee et al.\ \cite[Section 5.2]{Lee}, and naturally extending our definition of Venn diagrams in the previous sections, a hypergraph motif with $k$ edges can be encoded as a binary vector of size $2^k-1$:
\begin{definition}[Generalised Hypergraph Motifs]
    A \emph{Generalised Hypergraph Motif} (``GHM'') of order $k$ is a vector $\mathfrak{H}\in \{0,1\}^{2^k-1}$, indexed by non-empty subsets of $[k]$, that is
    \[\mathfrak{H}=(\mathfrak{H}_J)_{\emptyset \neq J\subseteq [k]}\,.\]
    We say that a $k$-edge hypergraph $G$ satisfies $\mathfrak{H}$, denoted by $\mathcal{\mathfrak{H}}(G)=1$, if there is an ordering $e_1,\dots,e_k$ of $E(G)$ such that for all non-empty $J \subseteq [k]$ we have
    \[\mathfrak{H}_J=1 \Leftrightarrow \left(\bigcap_{i \in J}e_i \right)\setminus \left(\bigcup_{i \in [k]\setminus J} e_i\right) \neq \emptyset \,.\]
    Furthermore, we define
    \[\motif{\mathfrak{H}}{G}:=\{A \subseteq E(G) \mid \#A=k~\wedge~\mathfrak{H}(G[A])=1\}\]
    for the set of all \emph{motifs} of $G$ satisfying $\mathfrak{H}$.
\end{definition}

Our aim in this section is to investigate the parameterised complexity of counting generalised hypergraph motifs. To this end, we first note that, for any GHM $\mathfrak{H}$ of order $k$, we can compute the function $G \mapsto \#\motif{\mathfrak{H}}{G}$ in time $|G|^{O(k)}$ up to a factor only depending on $\mathfrak{H}$ by brute force, leading to a polynomial-time algorithm if $\mathfrak{H}$ (and thus $k$) are fixed. For this reason, we ask for which generalised hypergraph motifs we can obtain an running time in which the exponent of the host graph $G$ does not depend on the order of the motif. In order to state and answer this question formally, we follow the standard approach for parameterised motif detection and counting problems (cf.\ \cite{Grohe07,Marx13,Bressanetal26}) and restrict the problem to classes of allowed hypergraph motifs:

\begin{definition}
Let $\mathfrak{C}$ be a recursively enumerable class of generalised hypergraph motifs. We define the following problem:

    \begin{mdframed}
\#\textsc{GeneralisedHyperMotif}($\mathfrak{C}$) \\
\textbf{Input:} A GHM $\mathfrak{H}\in \mathfrak{C}$, and a hypergraph $G$ \\
\textbf{Parameter:} $|\mathfrak{H}|+\mathsf{rank}(G)$ \\
\textbf{Output:} $\#\motif{\mathfrak{H}}{G}$
\end{mdframed}
\end{definition}

The aforementioned question on the tractability of counting generalised hypergraph motifs in a hypergraph $G$ can now be stated rigorously: For which classes $\mathfrak{C}$ is $\#\textsc{GeneralisedHyperMotif}(\mathfrak{C})$ fixed-parameter tractable? For attempting to answer this question, we need to take a brief tour to more intricate hypergraph width measures.

\subsection{Fractional Hypertreewidth and Adaptive Width}
Given a hypergraph $H$, a \emph{fractional independent set} of $H$ is a mapping $\iota: V(H) \to [0,1]$ such that, for each $e \in E(H)$, we have $\sum_{v \in e} \iota(v)\leq 1$. For a subset $X \subseteq V(H)$, we set $\iota(X):=\sum_{v \in X}\iota(v)$. A \emph{fractional edge cover} of a subset $X\subseteq V(H)$ is a mapping $\gamma: E(H) \to [0,1]$ such that, for each $v \in X$, we have $\sum_{e \in E(H):v \in e}\gamma(e)\geq 1$,
and the total weight of $\gamma$ is defined as $\sum_{e \in E(H)}\gamma(e)$. We write $\gamma^\ast_H:2^{V(H)}\to \mathbb{R}_{\geq 0}$ for the function that maps $X$ to the minimum total weight of any fractional edge cover of $X$.

\begin{definition}[$f$-width and $\mathcal{F}$-width]
    Let $H$ be a hypergraph, let $f:2^{V(H)}\to \mathbb{R}_{\geq 0}$, and let $\mathcal{T}=(T,\{B_t\}_{t \in V(T)})$ be a hypertree decomposition of $H$. The $f$\emph{-width} of $\mathcal{T}$  is $\max\{f(B_t) \mid t \in V(T)\}$. The  $f$\emph{-width} if $H$ is the minimum $f$-width of any hypertree decomposition of $H$.

    Given a (possibly infinite) set of functions $\mathcal{F}$ such that each $f \in \mathcal{F}$ maps vertices of $H$ to non-negative reals, the  $\mathcal{F}$\emph{-width} of $H$ is defined as
    \[ \mathcal{F}\text{-}\mathrm{width}(H):= \sup\{f\text{-}\mathrm{width}(H) \mid f \in \mathcal{F}\},. \]
\end{definition}

The fractional hypertreewidth and the adaptive width of a hypergraph can be defined via fractional edge covers and independent sets as follows (cf.\ \cite{Marx13}):
\begin{definition}[Fractional Hypertreewidth and Adaptive width]
    Let $H$ be a hypergraph. The \emph{fractional hypertreewidth} of $H$, denoted by $\mathsf{fhtw}(H)$, is defined as
    \[\mathsf{fhtw}(H):= \gamma^\ast_H\text{-}\mathrm{width}(H)\,.\]
    Moreover, writing $I$ for the set of all fractional independent sets of $H$, the \emph{adaptive width} of $H$, denoted by $\mathsf{aw}(H)$, is defined as
    \[\mathsf{aw}(H) := I\text{-}\mathrm{width}(H) \,.\]
\end{definition}

Given a class of hypergraphs $\mathcal{H}$, we say that $\mathcal{H}$ has bounded fractional hypertreewidth (resp.\ bounded adaptive width) if there is a constant $c$ such that, for all $H \in \mathcal{H}$ we have $\mathsf{fhtw}(H)\leq c$ (resp.\ $\mathsf{aw}(H)\leq c$).
It is well-known (cf.\ \cite{Marx13}) that any class $\mathcal{H}$ of bounded fractional hypertreewidth also has bounded adaptive width, but the backwards direction is not true.

It follows from the work of Grohe and Marx on constraint solving via fractional edge-covers~\cite{GroheM14} that $\#\textsc{Hom}(\mathcal{H})$ is efficiently solvable whenever $\mathcal{H}$ has bounded fractional hypertreewidth:\footnote{We note that Grohe and Marx state their result for the decision problem only; however, their argument applies to counting with only trivial modifications.}   
\begin{theorem}[\cite{GroheM14}]\label{thm:hom_upper_bound}
    Let $\mathcal{H}$ be a recursively enumerable class of hypergraphs of bounded fractional hypertreewidth. Then $\#\textsc{Hom}(\mathcal{H})$ is solvable in polynomial time.
\end{theorem}

For our lower bounds, we rely on the following result due to Marx~\cite[Theorem 7.1]{Marx13}:
\begin{theorem}[\cite{Marx13}]\label{thm:hom_lower_bound}
    Let $\mathcal{H}$ be a recursively enumerable class of hypergraphs of unbounded adaptive width. Then $\#\textsc{Hom}(\mathcal{H})$ is not fixed-parameter tractable, unless ETH fails.
\end{theorem}

We note that there is a gap between the previous two theorems in the sense that there are classes of hypergraphs of unbounded fractional hypertreewidth but bounded adaptive width. For those classes, the parameterised complexity of counting homomorphisms is still unresolved. Nevertheless, we are able to use Theorems~\ref{thm:hom_upper_bound} and ~\ref{thm:hom_lower_bound} to provide a partial classification of $\#\textsc{GeneralisedHyperMotif}(\mathfrak{C})$, but we remark that a full resolution of the complexity of the hypergraph homomorphism counting problem might be required for an exhaustive classification of $\#\textsc{GeneralisedHyperMotif}(\mathfrak{C})$.

\subsection{The Parameterised Complexity of $\#\textsc{GeneralisedHyperMotif}(\mathfrak{C})$}
Given a generalised hypergraph motif $\mathfrak{H}$ of order $k$, and a positive integer $r$, we write $S(\mathfrak{H},r)$ for the set of all (isomorphism types of) $k$-edge hypergraphs $H$ of rank at most $r$ such that $\mathfrak{H}(H)=1$. Moreover, we define
\[ \alpha_{\mathfrak{}{H},r}(F) := \sum_{H \in S(\mathfrak{H},r)} \auts{H}^{-1} \cdot \sum_{\substack{\rho \in P(V(H))\\H/\rho \cong F}} \mu(\rho)\,,\]
where $\mu(\rho)$ denotes the M\"obius function of the partition lattice of $V(H)$.

\begin{lemma}\label{lem:uncoloured_motif_to_hom_generalised}
    Let $r$ be a positive integer, and let $\mathfrak{H}$ be generalised hypergraph motif of order $k$. For all hypergraphs~$G$ of rank at most~$r$ we have
    \[ \motifs{\mathfrak{H}}{G} = \sum_{F} \alpha_{\mathfrak{H},r}(F) \cdot \homs{F}{G} \,.\]
\end{lemma}
\begin{proof}
    The proof is analogous to the case $k=3$ in Lemma~\ref{lem:uncoloured_motif_to_hom}.
\end{proof}

\begin{definition}[Hereditary width measures]
    Let $\mathfrak{H}$ be a generalised hypergraph motif of order~$k$ The \emph{hereditary fractional hypertreewidth} of $\mathfrak{H}$ is defined as $\max\{\mathsf{fhtw}(F)\mid \exists r>0:\alpha_{\mathfrak{H},r}(F)\neq 0\}$.
    Similarly, we define the \emph{hereditary adaptive width} of 
    $\mathfrak{H}$ as $\max\{\mathsf{aw}(F)\mid \exists r>0: \alpha_{\mathfrak{H},r}(F)\neq 0\}$.
\end{definition}

We are now able to prove the following partial classification for $\#\textsc{GeneralisedHyperMotif}(\mathfrak{C})$ via an easy application of uncoloured complexity monotonicity for computing linear combinations of hypergraph homomorphisms established by Bressan et al.\ \cite{Bressanetal26}.
\begin{theorem}\label{thm:main_GHMs}
    Let $\mathfrak{C}$ be a recursively enumerable class of generalised hypergraph motifs.
    \begin{itemize}
        \item If $\mathfrak{C}$ has bounded hereditary fractional hypertreewidth, then $\#\textsc{GeneralisedHyperMotif}(\mathfrak{C})$ is fixed-parameter tractable.
        \item If $\mathfrak{C}$ has unbounded hereditary adaptive width, then $\#\textsc{GeneralisedHyperMotif}(\mathfrak{C})$ is not fixed-parameter tractable, unless ETH fails.
    \end{itemize}
\end{theorem}
\begin{proof}
    For the upper bound, on input $\mathfrak{H}\in \mathfrak{C}$ and a hypergraph $G$ or rank $r$, we can compute in time only depending on $\mathfrak{H}$ and $r$, the hypergraphs $F$ and the coefficients $\alpha_{\mathfrak{H},r}(F)$ for all $F$ with $\alpha_{\mathfrak{H},r}(F)\neq 0$ --- note that $\alpha_{\mathfrak{H},r}(F)\neq 0$ implies that $F$  is a quotient of a hypergraph $H \in S(\mathfrak{H},r)$, the latter of which only depends on $\mathfrak{H}$ and $r$; hence we can iterate by bruteforce over all possible partitions $\rho \in P(V(H))$ for all $H \in S(\mathfrak{H},r)$. Afterwards, we can evaluate the linear combination in Lemma~\ref{lem:uncoloured_motif_to_hom_generalised} in FPT time by using the algorithm from Theorem~\ref{thm:hom_upper_bound} for each term $\homs{F}{G}$ with $\alpha_{\mathfrak{H},k}(F)\neq 0$; recall that those $F$ have bounded fractional hypertreewidth as $\mathfrak{C}$ has bounded hereditary fractional hypertreewidth.

    For the lower bound, let $\mathcal{H}$ be the class of all hypergraphs $F$ such that $\alpha_{\mathfrak{H},r}(F)\neq 0$ for some $\mathfrak{H}\in \mathfrak{C}$ and $r \geq 0$. Observe that $\mathcal{H}$ has unbounded adaptive width as $\mathfrak{C}$ has unbounded hereditary adaptive width. Hence, by Theorem~\ref{thm:hom_lower_bound}, the problem $\#\textsc{Hom}(\mathcal{H})$ is not fixed-parameter tractable unless ETH fails. Finally, for our choice of $\mathcal{H}$, we obtain
    \[\#\textsc{Hom}(\mathcal{H}) \leq^{\mathsf{FPT}} \#\textsc{GeneralisedHyperMotif}(\mathfrak{C}) \]
    via the uncoloured version of hypergraph complexity monotonicity as shown by Bressan et al.\ \cite{Bressanetal26}. An explicit statement of the reduction can be found in the full version~\cite[Corollary 4.11]{Bressanetal26full}, which applies to our setting despite $\#\textsc{GeneralisedHyperMotif}(\mathfrak{C})$ being parameterised by the rank of the input graph $G$; this is due to the fact that any input $(H,G)$ of $\#\textsc{Hom}(\mathcal{H})$ can be assumed to satisfy that the rank of $G$ is at most the rank of $H$ as edges in $G$ of higher cardinality do not affect homomorphisms from $H$ and can thus be deleted.
\end{proof}

\bibliographystyle{plain}
\bibliography{references}

\end{document}

%% file: big_venn_diagram_picture.tex
\begin{center}
\begin{tikzpicture}[scale=0.4, transform shape]
    
    \def\r{1.5}
    \coordinate (A) at (0,0);
    \coordinate (B) at (2,0);
    \coordinate (C) at (1,1.732);

    \begin{scope}
        \clip (A) circle (\r);
        \fill[pattern=\nel, pattern color=\redish] (A) circle (\r);
    \end{scope}
   
    \begin{scope}
        \clip (B) circle (\r);
        \fill[pattern=\nel, pattern color=\redish] (B) circle (\r);
    \end{scope}

    \begin{scope}
        \clip (C) circle (\r);
        \fill[pattern=\nel, pattern color=\redish] (C) circle (\r);
    \end{scope}

    \begin{scope}
        \clip (A) circle (\r);
        \clip (B) circle (\r);
        \fill[pattern=\nel, pattern color=\redish] (-2,-2) rectangle (4,4);
    \end{scope}

    \begin{scope}
        \clip (A) circle (\r);
        \clip (C) circle (\r);
        \fill[pattern=\nel, pattern color=\redish] (-1,-1) rectangle (4,4);
    \end{scope}

    \begin{scope}
        \clip (B) circle (\r);
        \clip (C) circle (\r);
        \fill[pattern=\nel, pattern color=\redish] (0,0) rectangle (4,4);
    \end{scope}

    \begin{scope}
        \clip (A) circle (\r);
        \clip (B) circle (\r);
        \clip (C) circle (\r);
        \fill[fill=white, pattern color=\redish] (0,0) rectangle (4,4);
    \end{scope}

    \draw[line width=1.2pt] (A) circle (\r);
    \draw[line width=1.2pt] (B) circle (\r);
    \draw[line width=1.2pt] (C) circle (\r);

    \node (V_1) at (1.1,-2.5) {\huge $\mathcal{V}_1$};
\end{tikzpicture} 
\begin{tikzpicture}[scale=0.4, transform shape]
    \def\r{1.5}
    \coordinate (A) at (0,0);
    \coordinate (B) at (2,0);
    \coordinate (C) at (1,1.732);

    \begin{scope}
        \clip (A) circle (\r);
        \fill[pattern=\nel, pattern color=\redish] (A) circle (\r);
    \end{scope}
   
    \begin{scope}
        \clip (B) circle (\r);
        \fill[pattern=none, pattern color=\redish] (B) circle (\r);
    \end{scope}

    \begin{scope}
        \clip (C) circle (\r);
        \fill[pattern=\nel, pattern color=\redish] (C) circle (\r);
    \end{scope}

    \begin{scope}
        \clip (A) circle (\r);
        \clip (B) circle (\r);
        \fill[pattern=\nel, pattern color=\redish] (-2,-2) rectangle (4,4);
    \end{scope}

    \begin{scope}
        \clip (A) circle (\r);
        \clip (C) circle (\r);
        \fill[pattern=\nel, pattern color=\redish] (-1,-1) rectangle (4,4);
    \end{scope}

    \begin{scope}
        \clip (B) circle (\r);
        \clip (C) circle (\r);
        \fill[pattern=\nel, pattern color=\redish] (0,0) rectangle (4,4);
    \end{scope}

    \begin{scope}
        \clip (A) circle (\r);
        \clip (B) circle (\r);
        \clip (C) circle (\r);
        \fill[fill=white, pattern color=\redish] (0,0) rectangle (4,4);
    \end{scope}

    \draw[line width=1.2pt] (A) circle (\r);
    \draw[line width=1.2pt] (B) circle (\r);
    \draw[line width=1.2pt] (C) circle (\r);

    \node (V_2) at (1.1,-2.5) {\huge $\mathcal{V}_2$};
\end{tikzpicture}
\begin{tikzpicture}[scale=0.4, transform shape]
    \def\r{1.5}
    \coordinate (A) at (0,0);
    \coordinate (B) at (2,0);
    \coordinate (C) at (1,1.732);

    \begin{scope}
        \clip (A) circle (\r);
        \fill[pattern=none, pattern color=\redish] (A) circle (\r);
    \end{scope}
   
    \begin{scope}
        \clip (B) circle (\r);
        \fill[pattern=none, pattern color=\redish] (B) circle (\r);
    \end{scope}

    \begin{scope}
        \clip (C) circle (\r);
        \fill[pattern=\nel, pattern color=\redish] (C) circle (\r);
    \end{scope}

    \begin{scope}
        \clip (A) circle (\r);
        \clip (B) circle (\r);
        \fill[pattern=\nel, pattern color=\redish] (-2,-2) rectangle (4,4);
    \end{scope}

    \begin{scope}
        \clip (A) circle (\r);
        \clip (C) circle (\r);
        \fill[pattern=\nel, pattern color=\redish] (-1,-1) rectangle (4,4);
    \end{scope}

    \begin{scope}
        \clip (B) circle (\r);
        \clip (C) circle (\r);
        \fill[pattern=\nel, pattern color=\redish] (0,0) rectangle (4,4);
    \end{scope}

    \begin{scope}
        \clip (A) circle (\r);
        \clip (B) circle (\r);
        \clip (C) circle (\r);
        \fill[fill=white, pattern color=\redish] (0,0) rectangle (4,4);
    \end{scope}

    \draw[line width=1.2pt] (A) circle (\r);
    \draw[line width=1.2pt] (B) circle (\r);
    \draw[line width=1.2pt] (C) circle (\r);

    \node (V_3) at (1.1,-2.5) {\huge $\mathcal{V}_3$};
\end{tikzpicture}
\begin{tikzpicture}[scale=0.4, transform shape]
    \def\r{1.5}
    \coordinate (A) at (0,0);
    \coordinate (B) at (2,0);
    \coordinate (C) at (1,1.732);

    \begin{scope}
        \clip (A) circle (\r);
        \fill[pattern=none, pattern color=\redish] (A) circle (\r);
    \end{scope}
   
    \begin{scope}
        \clip (B) circle (\r);
        \fill[pattern=none, pattern color=\redish] (B) circle (\r);
    \end{scope}

    \begin{scope}
        \clip (C) circle (\r);
        \fill[pattern=none, pattern color=\redish] (C) circle (\r);
    \end{scope}

    \begin{scope}
        \clip (A) circle (\r);
        \clip (B) circle (\r);
        \fill[pattern=\nel, pattern color=\redish] (-2,-2) rectangle (4,4);
    \end{scope}

    \begin{scope}
        \clip (A) circle (\r);
        \clip (C) circle (\r);
        \fill[pattern=\nel, pattern color=\redish] (-1,-1) rectangle (4,4);
    \end{scope}

    \begin{scope}
        \clip (B) circle (\r);
        \clip (C) circle (\r);
        \fill[pattern=\nel, pattern color=\redish] (0,0) rectangle (4,4);
    \end{scope}

    \begin{scope}
        \clip (A) circle (\r);
        \clip (B) circle (\r);
        \clip (C) circle (\r);
        \fill[fill=white, pattern color=\redish] (0,0) rectangle (4,4);
    \end{scope}

    \draw[line width=1.2pt] (A) circle (\r);
    \draw[line width=1.2pt] (B) circle (\r);
    \draw[line width=1.2pt] (C) circle (\r);

    \node (V_4) at (1.1,-2.5) {\huge $\mathcal{V}_4$};
\end{tikzpicture}
\begin{tikzpicture}[scale=0.4, transform shape]
    \def\r{1.5}
    \coordinate (A) at (0,0);
    \coordinate (B) at (2,0);
    \coordinate (C) at (1,1.732);

    \begin{scope}
        \clip (A) circle (\r);
        \fill[pattern=\nel, pattern color=\redish] (A) circle (\r);
    \end{scope}
   
    \begin{scope}
        \clip (B) circle (\r);
        \fill[pattern=\nel, pattern color=\redish] (B) circle (\r);
    \end{scope}

    \begin{scope}
        \clip (C) circle (\r);
        \fill[pattern=\nel, pattern color=\redish] (C) circle (\r);
    \end{scope}

    \begin{scope}
        \clip (A) circle (\r);
        \clip (B) circle (\r);
        \fill[pattern=\nel, pattern color=\redish] (-2,-2) rectangle (4,4);
    \end{scope}

    \begin{scope}
        \clip (A) circle (\r);
        \clip (C) circle (\r);
        \fill[fill=white, pattern color=\redish] (-1,-1) rectangle (4,4);
    \end{scope}

    \begin{scope}
        \clip (B) circle (\r);
        \clip (C) circle (\r);
        \fill[pattern=\nel, pattern color=\redish] (0,0) rectangle (4,4);
    \end{scope}

    \begin{scope}
        \clip (A) circle (\r);
        \clip (B) circle (\r);
        \clip (C) circle (\r);
        \fill[fill=white, pattern color=\redish] (0,0) rectangle (4,4);
    \end{scope}

    \draw[line width=1.2pt] (A) circle (\r);
    \draw[line width=1.2pt] (B) circle (\r);
    \draw[line width=1.2pt] (C) circle (\r);

    \node (V_5) at (1.1,-2.5) {\huge $\mathcal{V}_5$};
\end{tikzpicture}
\begin{tikzpicture}[scale=0.4, transform shape]
    \def\r{1.5}
    \coordinate (A) at (0,0);
    \coordinate (B) at (2,0);
    \coordinate (C) at (1,1.732);

    \begin{scope}
        \clip (A) circle (\r);
        \fill[pattern=\nel, pattern color=\redish] (A) circle (\r);
    \end{scope}
   
    \begin{scope}
        \clip (B) circle (\r);
        \fill[pattern=none, pattern color=\redish] (B) circle (\r);
    \end{scope}

    \begin{scope}
        \clip (C) circle (\r);
        \fill[pattern=\nel, pattern color=\redish] (C) circle (\r);
    \end{scope}

    \begin{scope}
        \clip (A) circle (\r);
        \clip (B) circle (\r);
        \fill[pattern=\nel, pattern color=\redish] (-2,-2) rectangle (4,4);
    \end{scope}

    \begin{scope}
        \clip (A) circle (\r);
        \clip (C) circle (\r);
        \fill[fill=white, pattern color=\redish] (-1,-1) rectangle (4,4);
    \end{scope}

    \begin{scope}
        \clip (B) circle (\r);
        \clip (C) circle (\r);
        \fill[pattern=\nel, pattern color=\redish] (0,0) rectangle (4,4);
    \end{scope}

    \begin{scope}
        \clip (A) circle (\r);
        \clip (B) circle (\r);
        \clip (C) circle (\r);
        \fill[fill=white, pattern color=\redish] (0,0) rectangle (4,4);
    \end{scope}

    \draw[line width=1.2pt] (A) circle (\r);
    \draw[line width=1.2pt] (B) circle (\r);
    \draw[line width=1.2pt] (C) circle (\r);

    \node (V_6) at (1.1,-2.5) {\huge $\mathcal{V}_6$};
\end{tikzpicture}
\begin{tikzpicture}[scale=0.4, transform shape]
    \def\r{1.5}
    \coordinate (A) at (0,0);
    \coordinate (B) at (2,0);
    \coordinate (C) at (1,1.732);

    \begin{scope}
        \clip (A) circle (\r);
        \fill[pattern=\nel, pattern color=\redish] (A) circle (\r);
    \end{scope}
   
    \begin{scope}
        \clip (B) circle (\r);
        \fill[pattern=\nel, pattern color=\redish] (B) circle (\r);
    \end{scope}

    \begin{scope}
        \clip (C) circle (\r);
        \fill[pattern=\nel, pattern color=\redish] (C) circle (\r);
    \end{scope}

    \begin{scope}
        \clip (A) circle (\r);
        \clip (B) circle (\r);
        \fill[pattern=\nel, pattern color=\redish] (-2,-2) rectangle (4,4);
    \end{scope}

    \begin{scope}
        \clip (A) circle (\r);
        \clip (C) circle (\r);
        \fill[pattern=\nel, pattern color=\redish] (-1,-1) rectangle (4,4);
    \end{scope}

    \begin{scope}
        \clip (B) circle (\r);
        \clip (C) circle (\r);
        \fill[pattern=\nel, pattern color=\redish] (0,0) rectangle (4,4);
    \end{scope}

    \begin{scope}
        \clip (A) circle (\r);
        \clip (B) circle (\r);
        \clip (C) circle (\r);
        \fill[pattern=\nel, pattern color=\redish] (0,0) rectangle (4,4);
    \end{scope}

    \draw[line width=1.2pt] (A) circle (\r);
    \draw[line width=1.2pt] (B) circle (\r);
    \draw[line width=1.2pt] (C) circle (\r);

    \node (V_7) at (1.1,-2.5) {\huge $\mathcal{V}_7$};
\end{tikzpicture}\\ 
~\\
\begin{tikzpicture}[scale=0.4, transform shape]
    \def\r{1.5}
    \coordinate (A) at (0,0);
    \coordinate (B) at (2,0);
    \coordinate (C) at (1,1.732);

    \begin{scope}
        \clip (A) circle (\r);
        \fill[pattern=\nel, pattern color=\redish] (A) circle (\r);
    \end{scope}
   
    \begin{scope}
        \clip (B) circle (\r);
        \fill[pattern=none, pattern color=\redish] (B) circle (\r);
    \end{scope}

    \begin{scope}
        \clip (C) circle (\r);
        \fill[pattern=\nel, pattern color=\redish] (C) circle (\r);
    \end{scope}

    \begin{scope}
        \clip (A) circle (\r);
        \clip (B) circle (\r);
        \fill[pattern=\nel, pattern color=\redish] (-2,-2) rectangle (4,4);
    \end{scope}

    \begin{scope}
        \clip (A) circle (\r);
        \clip (C) circle (\r);
        \fill[pattern=\nel, pattern color=\redish] (-1,-1) rectangle (4,4);
    \end{scope}

    \begin{scope}
        \clip (B) circle (\r);
        \clip (C) circle (\r);
        \fill[pattern=\nel, pattern color=\redish] (0,0) rectangle (4,4);
    \end{scope}

    \begin{scope}
        \clip (A) circle (\r);
        \clip (B) circle (\r);
        \clip (C) circle (\r);
        \fill[pattern=\nel, pattern color=\redish] (0,0) rectangle (4,4);
    \end{scope}

    \draw[line width=1.2pt] (A) circle (\r);
    \draw[line width=1.2pt] (B) circle (\r);
    \draw[line width=1.2pt] (C) circle (\r);

    \node (V_8) at (1.1,-2.5) {\huge $\mathcal{V}_8$};
\end{tikzpicture}
\begin{tikzpicture}[scale=0.4, transform shape]
    \def\r{1.5}
    \coordinate (A) at (0,0);
    \coordinate (B) at (2,0);
    \coordinate (C) at (1,1.732);

    \begin{scope}
        \clip (A) circle (\r);
        \fill[pattern=none, pattern color=\redish] (A) circle (\r);
    \end{scope}
   
    \begin{scope}
        \clip (B) circle (\r);
        \fill[pattern=none, pattern color=\redish] (B) circle (\r);
    \end{scope}

    \begin{scope}
        \clip (C) circle (\r);
        \fill[pattern=\nel, pattern color=\redish] (C) circle (\r);
    \end{scope}

    \begin{scope}
        \clip (A) circle (\r);
        \clip (B) circle (\r);
        \fill[pattern=\nel, pattern color=\redish] (-2,-2) rectangle (4,4);
    \end{scope}

    \begin{scope}
        \clip (A) circle (\r);
        \clip (C) circle (\r);
        \fill[pattern=\nel, pattern color=\redish] (-1,-1) rectangle (4,4);
    \end{scope}

    \begin{scope}
        \clip (B) circle (\r);
        \clip (C) circle (\r);
        \fill[pattern=\nel, pattern color=\redish] (0,0) rectangle (4,4);
    \end{scope}

    \begin{scope}
        \clip (A) circle (\r);
        \clip (B) circle (\r);
        \clip (C) circle (\r);
        \fill[pattern=\nel, pattern color=\redish] (0,0) rectangle (4,4);
    \end{scope}

    \draw[line width=1.2pt] (A) circle (\r);
    \draw[line width=1.2pt] (B) circle (\r);
    \draw[line width=1.2pt] (C) circle (\r);

    \node (V_9) at (1.1,-2.5) {\huge $\mathcal{V}_9$};
\end{tikzpicture}
\begin{tikzpicture}[scale=0.4, transform shape]
    \def\r{1.5}
    \coordinate (A) at (0,0);
    \coordinate (B) at (2,0);
    \coordinate (C) at (1,1.732);

    \begin{scope}
        \clip (A) circle (\r);
        \fill[pattern=none, pattern color=\redish] (A) circle (\r);
    \end{scope}
   
    \begin{scope}
        \clip (B) circle (\r);
        \fill[pattern=none, pattern color=\redish] (B) circle (\r);
    \end{scope}

    \begin{scope}
        \clip (C) circle (\r);
        \fill[pattern=none, pattern color=\redish] (C) circle (\r);
    \end{scope}

    \begin{scope}
        \clip (A) circle (\r);
        \clip (B) circle (\r);
        \fill[pattern=\nel, pattern color=\redish] (-2,-2) rectangle (4,4);
    \end{scope}

    \begin{scope}
        \clip (A) circle (\r);
        \clip (C) circle (\r);
        \fill[pattern=\nel, pattern color=\redish] (-1,-1) rectangle (4,4);
    \end{scope}

    \begin{scope}
        \clip (B) circle (\r);
        \clip (C) circle (\r);
        \fill[pattern=\nel, pattern color=\redish] (0,0) rectangle (4,4);
    \end{scope}

    \begin{scope}
        \clip (A) circle (\r);
        \clip (B) circle (\r);
        \clip (C) circle (\r);
        \fill[pattern=\nel, pattern color=\redish] (0,0) rectangle (4,4);
    \end{scope}

    \draw[line width=1.2pt] (A) circle (\r);
    \draw[line width=1.2pt] (B) circle (\r);
    \draw[line width=1.2pt] (C) circle (\r);

    \node (V_10) at (1.1,-2.5) {\huge $\mathcal{V}_{10}$};
\end{tikzpicture}
\begin{tikzpicture}[scale=0.4, transform shape]
    \def\r{1.5}
    \coordinate (A) at (0,0);
    \coordinate (B) at (2,0);
    \coordinate (C) at (1,1.732);

    \begin{scope}
        \clip (A) circle (\r);
        \fill[pattern=\nel, pattern color=\redish] (A) circle (\r);
    \end{scope}
   
    \begin{scope}
        \clip (B) circle (\r);
        \fill[pattern=none, pattern color=\redish] (B) circle (\r);
    \end{scope}

    \begin{scope}
        \clip (C) circle (\r);
        \fill[pattern=\nel, pattern color=\redish] (C) circle (\r);
    \end{scope}

    \begin{scope}
        \clip (A) circle (\r);
        \clip (B) circle (\r);
        \fill[pattern=\nel, pattern color=\redish] (-2,-2) rectangle (4,4);
    \end{scope}

    \begin{scope}
        \clip (A) circle (\r);
        \clip (C) circle (\r);
        \fill[fill=white, pattern color=\redish] (-1,-1) rectangle (4,4);
    \end{scope}

    \begin{scope}
        \clip (B) circle (\r);
        \clip (C) circle (\r);
        \fill[pattern=\nel, pattern color=\redish] (0,0) rectangle (4,4);
    \end{scope}

    \begin{scope}
        \clip (A) circle (\r);
        \clip (B) circle (\r);
        \clip (C) circle (\r);
        \fill[pattern=\nel, pattern color=\redish] (0,0) rectangle (4,4);
    \end{scope}

    \draw[line width=1.2pt] (A) circle (\r);
    \draw[line width=1.2pt] (B) circle (\r);
    \draw[line width=1.2pt] (C) circle (\r);

    \node (V_11) at (1.1,-2.5) {\huge $\mathcal{V}_{11}$};
\end{tikzpicture} 
\begin{tikzpicture}[scale=0.4, transform shape]
    \def\r{1.5}
    \coordinate (A) at (0,0);
    \coordinate (B) at (2,0);
    \coordinate (C) at (1,1.732);

    \begin{scope}
        \clip (A) circle (\r);
        \fill[pattern=\nel, pattern color=\redish] (A) circle (\r);
    \end{scope}
   
    \begin{scope}
        \clip (B) circle (\r);
        \fill[pattern=\nel, pattern color=\redish] (B) circle (\r);
    \end{scope}

    \begin{scope}
        \clip (C) circle (\r);
        \fill[pattern=\nel, pattern color=\redish] (C) circle (\r);
    \end{scope}

    \begin{scope}
        \clip (A) circle (\r);
        \clip (B) circle (\r);
        \fill[pattern=\nel, pattern color=\redish] (-2,-2) rectangle (4,4);
    \end{scope}

    \begin{scope}
        \clip (A) circle (\r);
        \clip (C) circle (\r);
        \fill[fill=white, pattern color=\redish] (-1,-1) rectangle (4,4);
    \end{scope}

    \begin{scope}
        \clip (B) circle (\r);
        \clip (C) circle (\r);
        \fill[fill=white, pattern color=\redish] (0,0) rectangle (4,4);
    \end{scope}

    \begin{scope}
        \clip (A) circle (\r);
        \clip (B) circle (\r);
        \clip (C) circle (\r);
        \fill[pattern=\nel, pattern color=\redish] (0,0) rectangle (4,4);
    \end{scope}

    \draw[line width=1.2pt] (A) circle (\r);
    \draw[line width=1.2pt] (B) circle (\r);
    \draw[line width=1.2pt] (C) circle (\r);

    \node (V_12) at (1.1,-2.5) {\huge $\mathcal{V}_{12}$};
\end{tikzpicture}
\begin{tikzpicture}[scale=0.4, transform shape]
    \def\r{1.5}
    \coordinate (A) at (0,0);
    \coordinate (B) at (2,0);
    \coordinate (C) at (1,1.732);

    \begin{scope}
        \clip (A) circle (\r);
        \fill[pattern=\nel, pattern color=\redish] (A) circle (\r);
    \end{scope}
   
    \begin{scope}
        \clip (B) circle (\r);
        \fill[pattern=\nel, pattern color=\redish] (B) circle (\r);
    \end{scope}

    \begin{scope}
        \clip (C) circle (\r);
        \fill[pattern=\nel, pattern color=\redish] (C) circle (\r);
    \end{scope}

    \begin{scope}
        \clip (A) circle (\r);
        \clip (B) circle (\r);
        \fill[fill=white, pattern color=\redish] (-2,-2) rectangle (4,4);
    \end{scope}

    \begin{scope}
        \clip (A) circle (\r);
        \clip (C) circle (\r);
        \fill[fill=white, pattern color=\redish] (-1,-1) rectangle (4,4);
    \end{scope}

    \begin{scope}
        \clip (B) circle (\r);
        \clip (C) circle (\r);
        \fill[fill=white, pattern color=\redish] (0,0) rectangle (4,4);
    \end{scope}

    \begin{scope}
        \clip (A) circle (\r);
        \clip (B) circle (\r);
        \clip (C) circle (\r);
        \fill[pattern=\nel, pattern color=\redish] (0,0) rectangle (4,4);
    \end{scope}

    \draw[line width=1.2pt] (A) circle (\r);
    \draw[line width=1.2pt] (B) circle (\r);
    \draw[line width=1.2pt] (C) circle (\r);

    \node (V_13) at (1.1,-2.5) {\huge $\mathcal{V}_{13}$};
\end{tikzpicture}
\begin{tikzpicture}[scale=0.4, transform shape]
    \def\r{1.5}
    \coordinate (A) at (0,0);
    \coordinate (B) at (2,0);
    \coordinate (C) at (1,1.732);

    \begin{scope}
        \clip (A) circle (\r);
        \fill[pattern=\nel, pattern color=\redish] (A) circle (\r);
    \end{scope}
   
    \begin{scope}
        \clip (B) circle (\r);
        \fill[pattern=\nel, pattern color=\redish] (B) circle (\r);
    \end{scope}

    \begin{scope}
        \clip (C) circle (\r);
        \fill[pattern=\nel, pattern color=\redish] (C) circle (\r);
    \end{scope}

    \begin{scope}
        \clip (A) circle (\r);
        \clip (B) circle (\r);
        \fill[pattern=\nel, pattern color=\redish] (-2,-2) rectangle (4,4);
    \end{scope}

    \begin{scope}
        \clip (A) circle (\r);
        \clip (C) circle (\r);
        \fill[fill=white, pattern color=\redish] (-1,-1) rectangle (4,4);
    \end{scope}

    \begin{scope}
        \clip (B) circle (\r);
        \clip (C) circle (\r);
        \fill[pattern=\nel, pattern color=\redish] (0,0) rectangle (4,4);
    \end{scope}

    \begin{scope}
        \clip (A) circle (\r);
        \clip (B) circle (\r);
        \clip (C) circle (\r);
        \fill[pattern=\nel, pattern color=\redish] (0,0) rectangle (4,4);
    \end{scope}

    \draw[line width=1.2pt] (A) circle (\r);
    \draw[line width=1.2pt] (B) circle (\r);
    \draw[line width=1.2pt] (C) circle (\r);

    \node (V_14) at (1.1,-2.5) {\huge $\mathcal{V}_{14}$};
\end{tikzpicture}\\
~\\
\begin{tikzpicture}[scale=0.4, transform shape]
    \def\r{1.5}
    \coordinate (A) at (0,0);
    \coordinate (B) at (2,0);
    \coordinate (C) at (1,1.732);

    \begin{scope}
        \clip (A) circle (\r);
        \fill[pattern=none, pattern color=\greenish] (A) circle (\r);
    \end{scope}
   
    \begin{scope}
        \clip (B) circle (\r);
        \fill[pattern=none, pattern color=\greenish] (B) circle (\r);
    \end{scope}

    \begin{scope}
        \clip (C) circle (\r);
        \fill[pattern=none, pattern color=\greenish] (C) circle (\r);
    \end{scope}

    \begin{scope}
        \clip (A) circle (\r);
        \clip (B) circle (\r);
        \fill[pattern=\nel, pattern color=\greenish] (-2,-2) rectangle (4,4);
    \end{scope}

    \begin{scope}
        \clip (A) circle (\r);
        \clip (C) circle (\r);
        \fill[fill=white, pattern color=\greenish] (-1,-1) rectangle (4,4);
    \end{scope}

    \begin{scope}
        \clip (B) circle (\r);
        \clip (C) circle (\r);
        \fill[pattern=\nel, pattern color=\greenish] (0,0) rectangle (4,4);
    \end{scope}

    \begin{scope}
        \clip (A) circle (\r);
        \clip (B) circle (\r);
        \clip (C) circle (\r);
        \fill[fill=white, pattern color=\greenish] (0,0) rectangle (4,4);
    \end{scope}

    \draw[line width=1.2pt] (A) circle (\r);
    \draw[line width=1.2pt] (B) circle (\r);
    \draw[line width=1.2pt] (C) circle (\r);

    \node (V_15) at (1.1,-2.5) {\huge $\mathcal{V}_{15}$};
\end{tikzpicture}
\begin{tikzpicture}[scale=0.4, transform shape]
    \def\r{1.5}
    \coordinate (A) at (0,0);
    \coordinate (B) at (2,0);
    \coordinate (C) at (1,1.732);

    \begin{scope}
        \clip (A) circle (\r);
        \fill[pattern=none, pattern color=\greenish] (A) circle (\r);
    \end{scope}
   
    \begin{scope}
        \clip (B) circle (\r);
        \fill[pattern=\nel, pattern color=\greenish] (B) circle (\r);
    \end{scope}

    \begin{scope}
        \clip (C) circle (\r);
        \fill[pattern=none, pattern color=\greenish] (C) circle (\r);
    \end{scope}

    \begin{scope}
        \clip (A) circle (\r);
        \clip (B) circle (\r);
        \fill[pattern=\nel, pattern color=\greenish] (-2,-2) rectangle (4,4);
    \end{scope}

    \begin{scope}
        \clip (A) circle (\r);
        \clip (C) circle (\r);
        \fill[fill=white, pattern color=\greenish] (-1,-1) rectangle (4,4);
    \end{scope}

    \begin{scope}
        \clip (B) circle (\r);
        \clip (C) circle (\r);
        \fill[pattern=\nel, pattern color=\greenish] (0,0) rectangle (4,4);
    \end{scope}

    \begin{scope}
        \clip (A) circle (\r);
        \clip (B) circle (\r);
        \clip (C) circle (\r);
        \fill[fill=white, pattern color=\greenish] (0,0) rectangle (4,4);
    \end{scope}

    \draw[line width=1.2pt] (A) circle (\r);
    \draw[line width=1.2pt] (B) circle (\r);
    \draw[line width=1.2pt] (C) circle (\r);

    \node (V_16) at (1.1,-2.5) {\huge $\mathcal{V}_{16}$};
\end{tikzpicture}
\begin{tikzpicture}[scale=0.4, transform shape]
    \def\r{1.5}
    \coordinate (A) at (0,0);
    \coordinate (B) at (2,0);
    \coordinate (C) at (1,1.732);

    \begin{scope}
        \clip (A) circle (\r);
        \fill[pattern=none, pattern color=\greenish] (A) circle (\r);
    \end{scope}
   
    \begin{scope}
        \clip (B) circle (\r);
        \fill[pattern=none, pattern color=\greenish] (B) circle (\r);
    \end{scope}

    \begin{scope}
        \clip (C) circle (\r);
        \fill[pattern=\nel, pattern color=\greenish] (C) circle (\r);
    \end{scope}

    \begin{scope}
        \clip (A) circle (\r);
        \clip (B) circle (\r);
        \fill[pattern=\nel, pattern color=\greenish] (-2,-2) rectangle (4,4);
    \end{scope}

    \begin{scope}
        \clip (A) circle (\r);
        \clip (C) circle (\r);
        \fill[fill=white, pattern color=\greenish] (-1,-1) rectangle (4,4);
    \end{scope}

    \begin{scope}
        \clip (B) circle (\r);
        \clip (C) circle (\r);
        \fill[pattern=\nel, pattern color=\greenish] (0,0) rectangle (4,4);
    \end{scope}

    \begin{scope}
        \clip (A) circle (\r);
        \clip (B) circle (\r);
        \clip (C) circle (\r);
        \fill[fill=white, pattern color=\greenish] (0,0) rectangle (4,4);
    \end{scope}

    \draw[line width=1.2pt] (A) circle (\r);
    \draw[line width=1.2pt] (B) circle (\r);
    \draw[line width=1.2pt] (C) circle (\r);

    \node (V_17) at (1.1,-2.5) {\huge $\mathcal{V}_{17}$};
\end{tikzpicture}
\begin{tikzpicture}[scale=0.4, transform shape]
    \def\r{1.5}
    \coordinate (A) at (0,0);
    \coordinate (B) at (2,0);
    \coordinate (C) at (1,1.732);

    \begin{scope}
        \clip (A) circle (\r);
        \fill[pattern=none, pattern color=\greenish] (A) circle (\r);
    \end{scope}
   
    \begin{scope}
        \clip (B) circle (\r);
        \fill[pattern=\nel, pattern color=\greenish] (B) circle (\r);
    \end{scope}

    \begin{scope}
        \clip (C) circle (\r);
        \fill[pattern=\nel, pattern color=\greenish] (C) circle (\r);
    \end{scope}

    \begin{scope}
        \clip (A) circle (\r);
        \clip (B) circle (\r);
        \fill[pattern=\nel, pattern color=\greenish] (-2,-2) rectangle (4,4);
    \end{scope}

    \begin{scope}
        \clip (A) circle (\r);
        \clip (C) circle (\r);
        \fill[fill=white, pattern color=\greenish] (-1,-1) rectangle (4,4);
    \end{scope}

    \begin{scope}
        \clip (B) circle (\r);
        \clip (C) circle (\r);
        \fill[pattern=\nel, pattern color=\greenish] (0,0) rectangle (4,4);
    \end{scope}

    \begin{scope}
        \clip (A) circle (\r);
        \clip (B) circle (\r);
        \clip (C) circle (\r);
        \fill[fill=white, pattern color=\greenish] (0,0) rectangle (4,4);
    \end{scope}

    \draw[line width=1.2pt] (A) circle (\r);
    \draw[line width=1.2pt] (B) circle (\r);
    \draw[line width=1.2pt] (C) circle (\r);

    \node (V_18) at (1.1,-2.5) {\huge $\mathcal{V}_{18}$};
\end{tikzpicture}
\begin{tikzpicture}[scale=0.4, transform shape]
    
    \def\r{1.5}
    \coordinate (A) at (0,0);
    \coordinate (B) at (2,0);
    \coordinate (C) at (1,1.732);

    \begin{scope}
        \clip (A) circle (\r);
        \fill[pattern=\nel, pattern color=\greenish] (A) circle (\r);
    \end{scope}
   
    \begin{scope}
        \clip (B) circle (\r);
        \fill[pattern=\nel, pattern color=\greenish] (B) circle (\r);
    \end{scope}

    \begin{scope}
        \clip (C) circle (\r);
        \fill[pattern=none, pattern color=\greenish] (C) circle (\r);
    \end{scope}

    \begin{scope}
        \clip (A) circle (\r);
        \clip (B) circle (\r);
        \fill[fill=white, pattern color=\greenish] (-2,-2) rectangle (4,4);
    \end{scope}

    \begin{scope}
        \clip (A) circle (\r);
        \clip (C) circle (\r);
        \fill[fill=white, pattern color=\greenish] (-1,-1) rectangle (4,4);
    \end{scope}

    \begin{scope}
        \clip (B) circle (\r);
        \clip (C) circle (\r);
        \fill[fill=white, pattern color=\greenish] (0,0) rectangle (4,4);
    \end{scope}

    \begin{scope}
        \clip (A) circle (\r);
        \clip (B) circle (\r);
        \clip (C) circle (\r);
        \fill[pattern=\nel, pattern color=\greenish] (0,0) rectangle (4,4);
    \end{scope}

    \draw[line width=1.2pt] (A) circle (\r);
    \draw[line width=1.2pt] (B) circle (\r);
    \draw[line width=1.2pt] (C) circle (\r);

    \node (V_19) at (1.1,-2.5) {\huge $\mathcal{V}_{19}$};
\end{tikzpicture} 
\begin{tikzpicture}[scale=0.4, transform shape]
    \def\r{1.5}
    \coordinate (A) at (0,0);
    \coordinate (B) at (2,0);
    \coordinate (C) at (1,1.732);

    \begin{scope}
        \clip (A) circle (\r);
        \fill[pattern=\nel, pattern color=\greenish] (A) circle (\r);
    \end{scope}
   
    \begin{scope}
        \clip (B) circle (\r);
        \fill[pattern=none, pattern color=\greenish] (B) circle (\r);
    \end{scope}

    \begin{scope}
        \clip (C) circle (\r);
        \fill[pattern=none, pattern color=\greenish] (C) circle (\r);
    \end{scope}

    \begin{scope}
        \clip (A) circle (\r);
        \clip (B) circle (\r);
        \fill[pattern=\nel, pattern color=\greenish] (-2,-2) rectangle (4,4);
    \end{scope}

    \begin{scope}
        \clip (A) circle (\r);
        \clip (C) circle (\r);
        \fill[fill=white, pattern color=\greenish] (-1,-1) rectangle (4,4);
    \end{scope}

    \begin{scope}
        \clip (B) circle (\r);
        \clip (C) circle (\r);
        \fill[fill=white, pattern color=\greenish] (0,0) rectangle (4,4);
    \end{scope}

    \begin{scope}
        \clip (A) circle (\r);
        \clip (B) circle (\r);
        \clip (C) circle (\r);
        \fill[pattern=\nel, pattern color=\greenish] (0,0) rectangle (4,4);
    \end{scope}

    \draw[line width=1.2pt] (A) circle (\r);
    \draw[line width=1.2pt] (B) circle (\r);
    \draw[line width=1.2pt] (C) circle (\r);

    \node (V_20) at (1.1,-2.5) {\huge $\mathcal{V}_{20}$};
\end{tikzpicture} \\
~\\
\begin{tikzpicture}[scale=0.4, transform shape]
    \def\r{1.5}
    \coordinate (A) at (0,0);
    \coordinate (B) at (2,0);
    \coordinate (C) at (1,1.732);

    \begin{scope}
        \clip (A) circle (\r);
        \fill[pattern=\nel, pattern color=\greenish] (A) circle (\r);
    \end{scope}
   
    \begin{scope}
        \clip (B) circle (\r);
        \fill[pattern=\nel, pattern color=\greenish] (B) circle (\r);
    \end{scope}

    \begin{scope}
        \clip (C) circle (\r);
        \fill[pattern=none, pattern color=\greenish] (C) circle (\r);
    \end{scope}

    \begin{scope}
        \clip (A) circle (\r);
        \clip (B) circle (\r);
        \fill[pattern=\nel, pattern color=\greenish] (-2,-2) rectangle (4,4);
    \end{scope}

    \begin{scope}
        \clip (A) circle (\r);
        \clip (C) circle (\r);
        \fill[fill=white, pattern color=\greenish] (-1,-1) rectangle (4,4);
    \end{scope}

    \begin{scope}
        \clip (B) circle (\r);
        \clip (C) circle (\r);
        \fill[fill=white, pattern color=\greenish] (0,0) rectangle (4,4);
    \end{scope}

    \begin{scope}
        \clip (A) circle (\r);
        \clip (B) circle (\r);
        \clip (C) circle (\r);
        \fill[pattern=\nel, pattern color=\greenish] (0,0) rectangle (4,4);
    \end{scope}

    \draw[line width=1.2pt] (A) circle (\r);
    \draw[line width=1.2pt] (B) circle (\r);
    \draw[line width=1.2pt] (C) circle (\r);

    \node (V_21) at (1.1,-2.5) {\huge $\mathcal{V}_{21}$};
\end{tikzpicture}
\begin{tikzpicture}[scale=0.4, transform shape]
    \def\r{1.5}
    \coordinate (A) at (0,0);
    \coordinate (B) at (2,0);
    \coordinate (C) at (1,1.732);

    \begin{scope}
        \clip (A) circle (\r);
        \fill[pattern=\nel, pattern color=\greenish] (A) circle (\r);
    \end{scope}
   
    \begin{scope}
        \clip (B) circle (\r);
        \fill[pattern=none, pattern color=\greenish] (B) circle (\r);
    \end{scope}

    \begin{scope}
        \clip (C) circle (\r);
        \fill[pattern=\nel, pattern color=\greenish] (C) circle (\r);
    \end{scope}

    \begin{scope}
        \clip (A) circle (\r);
        \clip (B) circle (\r);
        \fill[pattern=\nel, pattern color=\greenish] (-2,-2) rectangle (4,4);
    \end{scope}

    \begin{scope}
        \clip (A) circle (\r);
        \clip (C) circle (\r);
        \fill[fill=white, pattern color=\greenish] (-1,-1) rectangle (4,4);
    \end{scope}

    \begin{scope}
        \clip (B) circle (\r);
        \clip (C) circle (\r);
        \fill[fill=white, pattern color=\greenish] (0,0) rectangle (4,4);
    \end{scope}

    \begin{scope}
        \clip (A) circle (\r);
        \clip (B) circle (\r);
        \clip (C) circle (\r);
        \fill[pattern=\nel, pattern color=\greenish] (0,0) rectangle (4,4);
    \end{scope}

    \draw[line width=1.2pt] (A) circle (\r);
    \draw[line width=1.2pt] (B) circle (\r);
    \draw[line width=1.2pt] (C) circle (\r);

    \node (V_22) at (1.1,-2.5) {\huge $\mathcal{V}_{22}$};
\end{tikzpicture}
\begin{tikzpicture}[scale=0.4, transform shape]
    \def\r{1.5}
    \coordinate (A) at (0,0);
    \coordinate (B) at (2,0);
    \coordinate (C) at (1,1.732);

    \begin{scope}
        \clip (A) circle (\r);
        \fill[pattern=none, pattern color=\greenish] (A) circle (\r);
    \end{scope}
   
    \begin{scope}
        \clip (B) circle (\r);
        \fill[pattern=none, pattern color=\greenish] (B) circle (\r);
    \end{scope}

    \begin{scope}
        \clip (C) circle (\r);
        \fill[pattern=none, pattern color=\greenish] (C) circle (\r);
    \end{scope}

    \begin{scope}
        \clip (A) circle (\r);
        \clip (B) circle (\r);
        \fill[pattern=\nel, pattern color=\greenish] (-2,-2) rectangle (4,4);
    \end{scope}

    \begin{scope}
        \clip (A) circle (\r);
        \clip (C) circle (\r);
        \fill[fill=white, pattern color=\greenish] (-1,-1) rectangle (4,4);
    \end{scope}

    \begin{scope}
        \clip (B) circle (\r);
        \clip (C) circle (\r);
        \fill[pattern=\nel, pattern color=\greenish] (0,0) rectangle (4,4);
    \end{scope}

    \begin{scope}
        \clip (A) circle (\r);
        \clip (B) circle (\r);
        \clip (C) circle (\r);
        \fill[pattern=\nel, pattern color=\greenish] (0,0) rectangle (4,4);
    \end{scope}

    \draw[line width=1.2pt] (A) circle (\r);
    \draw[line width=1.2pt] (B) circle (\r);
    \draw[line width=1.2pt] (C) circle (\r);

    \node (V_23) at (1.1,-2.5) {\huge $\mathcal{V}_{23}$};
\end{tikzpicture}
\begin{tikzpicture}[scale=0.4, transform shape]
    \def\r{1.5}
    \coordinate (A) at (0,0);
    \coordinate (B) at (2,0);
    \coordinate (C) at (1,1.732);

    \begin{scope}
        \clip (A) circle (\r);
        \fill[pattern=none, pattern color=\greenish] (A) circle (\r);
    \end{scope}
   
    \begin{scope}
        \clip (B) circle (\r);
        \fill[pattern=\nel, pattern color=\greenish] (B) circle (\r);
    \end{scope}

    \begin{scope}
        \clip (C) circle (\r);
        \fill[pattern=none, pattern color=\greenish] (C) circle (\r);
    \end{scope}

    \begin{scope}
        \clip (A) circle (\r);
        \clip (B) circle (\r);
        \fill[pattern=\nel, pattern color=\greenish] (-2,-2) rectangle (4,4);
    \end{scope}

    \begin{scope}
        \clip (A) circle (\r);
        \clip (C) circle (\r);
        \fill[fill=white, pattern color=\greenish] (-1,-1) rectangle (4,4);
    \end{scope}

    \begin{scope}
        \clip (B) circle (\r);
        \clip (C) circle (\r);
        \fill[pattern=\nel, pattern color=\greenish] (0,0) rectangle (4,4);
    \end{scope}

    \begin{scope}
        \clip (A) circle (\r);
        \clip (B) circle (\r);
        \clip (C) circle (\r);
        \fill[pattern=\nel, pattern color=\greenish] (0,0) rectangle (4,4);
    \end{scope}

    \draw[line width=1.2pt] (A) circle (\r);
    \draw[line width=1.2pt] (B) circle (\r);
    \draw[line width=1.2pt] (C) circle (\r);

    \node (V_24) at (1.1,-2.5) {\huge $\mathcal{V}_{24}$};
\end{tikzpicture}
\begin{tikzpicture}[scale=0.4, transform shape]
    \def\r{1.5}
    \coordinate (A) at (0,0);
    \coordinate (B) at (2,0);
    \coordinate (C) at (1,1.732);

    \begin{scope}
        \clip (A) circle (\r);
        \fill[pattern=none, pattern color=\greenish] (A) circle (\r);
    \end{scope}
   
    \begin{scope}
        \clip (B) circle (\r);
        \fill[pattern=none, pattern color=\greenish] (B) circle (\r);
    \end{scope}

    \begin{scope}
        \clip (C) circle (\r);
        \fill[pattern=\nel, pattern color=\greenish] (C) circle (\r);
    \end{scope}

    \begin{scope}
        \clip (A) circle (\r);
        \clip (B) circle (\r);
        \fill[pattern=\nel, pattern color=\greenish] (-2,-2) rectangle (4,4);
    \end{scope}

    \begin{scope}
        \clip (A) circle (\r);
        \clip (C) circle (\r);
        \fill[fill=white, pattern color=\greenish] (-1,-1) rectangle (4,4);
    \end{scope}

    \begin{scope}
        \clip (B) circle (\r);
        \clip (C) circle (\r);
        \fill[pattern=\nel, pattern color=\greenish] (0,0) rectangle (4,4);
    \end{scope}

    \begin{scope}
        \clip (A) circle (\r);
        \clip (B) circle (\r);
        \clip (C) circle (\r);
        \fill[pattern=\nel, pattern color=\greenish] (0,0) rectangle (4,4);
    \end{scope}

    \draw[line width=1.2pt] (A) circle (\r);
    \draw[line width=1.2pt] (B) circle (\r);
    \draw[line width=1.2pt] (C) circle (\r);

    \node (V_25) at (1.1,-2.5) {\huge $\mathcal{V}_{25}$};
\end{tikzpicture}
\begin{tikzpicture}[scale=0.4, transform shape]
    \def\r{1.5}
    \coordinate (A) at (0,0);
    \coordinate (B) at (2,0);
    \coordinate (C) at (1,1.732);

    \begin{scope}
        \clip (A) circle (\r);
        \fill[pattern=none, pattern color=\greenish] (A) circle (\r);
    \end{scope}
   
    \begin{scope}
        \clip (B) circle (\r);
        \fill[pattern=\nel, pattern color=\greenish] (B) circle (\r);
    \end{scope}

    \begin{scope}
        \clip (C) circle (\r);
        \fill[pattern=\nel, pattern color=\greenish] (C) circle (\r);
    \end{scope}

    \begin{scope}
        \clip (A) circle (\r);
        \clip (B) circle (\r);
        \fill[pattern=\nel, pattern color=\greenish] (-2,-2) rectangle (4,4);
    \end{scope}

    \begin{scope}
        \clip (A) circle (\r);
        \clip (C) circle (\r);
        \fill[fill=white, pattern color=\greenish] (-1,-1) rectangle (4,4);
    \end{scope}

    \begin{scope}
        \clip (B) circle (\r);
        \clip (C) circle (\r);
        \fill[pattern=\nel, pattern color=\greenish] (0,0) rectangle (4,4);
    \end{scope}

    \begin{scope}
        \clip (A) circle (\r);
        \clip (B) circle (\r);
        \clip (C) circle (\r);
        \fill[pattern=\nel, pattern color=\greenish] (0,0) rectangle (4,4);
    \end{scope}

    \draw[line width=1.2pt] (A) circle (\r);
    \draw[line width=1.2pt] (B) circle (\r);
    \draw[line width=1.2pt] (C) circle (\r);

    \node (V_26) at (1.1,-2.5) {\huge $\mathcal{V}_{26}$};
\end{tikzpicture}
\\
\end{center}